\documentclass[%
 reprint,
 amsmath,amssymb,
 aps,
 prx,
]{revtex4-2}

\usepackage{graphicx}
\usepackage{dcolumn}
\usepackage{bm}

\usepackage{braket}
\usepackage{xcolor}

\definecolor{newColor}{RGB}{152,63,16}
\usepackage{amsthm}
\usepackage{mdframed}
\usepackage{multirow}

\usepackage{bbold}
\usepackage{hyperref}
\makeatletter
\newcommand{\neutralize}[1]{\expandafter\let\csname c@#1\endcsname\count@}
\makeatother

\makeatletter
\AtBeginDocument{\addtocontents{apc}{}}

\newcommand{\appendixcontents}{%
  \section*{Appendix Contents}           
  \@starttoc{apc}%
}

\newcommand{\appsection}[1]{%
  \section{#1}
  \addcontentsline{apc}{section}{
    \protect\numberline{\thesection}#1}%
}

\newcommand{\appsubsection}[1]{%
  \subsection{#1}%
  \addcontentsline{apc}{subsection}{%
    \protect\numberline{\thesubsection}#1}%
}
\makeatother

\theoremstyle{definition}
\newtheorem{thm}{Theorem}
\surroundwithmdframed[backgroundcolor=gray!10]{thm}

\newtheorem{thmapp}{Theorem}
\surroundwithmdframed[backgroundcolor=gray!10]{thmapp}

\newtheorem{pro}{Proposition}
\numberwithin{pro}{section}
\surroundwithmdframed[backgroundcolor=gray!10]{pro}
\newtheorem{Cor}{Corollary}
\surroundwithmdframed[backgroundcolor=gray!10]{Cor}

\numberwithin{cor}{section}

\newtheorem{lem}{Lemma}
\numberwithin{lem}{section}
\surroundwithmdframed[backgroundcolor=gray!10]{lem}

\numberwithin{conj}{section}
\newtheorem{defi}{Definition}
\surroundwithmdframed[backgroundcolor=gray!10]{defi}

\newtheorem{cas}{Case}
\newenvironment{casbis}[1]
  {\renewcommand{\thecas}{\ref*{#1}$'$}%
   \addtocounter{cas}{-1}%
   \begin{cas}}
  {\end{cas}}

\usepackage{CJK}

\begin{document}
\begin{CJK*}{UTF8}{gbsn}
\preprint{APS/123-QEd}

\title{Universal energy-space localization and stable quantum phases against time-dependent perturbations}

\author{Hongye Yu (余泓烨)}
\affiliation{C. N. Yang Institute for Theoretical Physics,
State University of New York at
Stony Brook, Stony Brook, NY 11794-3840, USA}
\affiliation{Department of Physics and Astronomy, State University of New York at
Stony Brook, Stony Brook, NY 11794-3800, USA}
\author{Tzu-Chieh Wei (\CJKfamily{bsmi}魏子傑)}
\affiliation{C. N. Yang Institute for Theoretical Physics,
State University of New York at
Stony Brook, Stony Brook, NY 11794-3840, USA}
\affiliation{Department of Physics and Astronomy, State University of New York at
Stony Brook, Stony Brook, NY 11794-3800, USA}


\begin{abstract} 
Stability against perturbations is a highly nontrivial property of quantum systems and is often a requirement to define new phases. In most systems where stability can be rigorously established, only static perturbations are considered; whether any system can remain stable against generic time-dependent perturbations is largely elusive. In this work, we identify a universal phenomenon in driving $q$-local Hamiltonians called energy-space localization and prove that it can survive under generic time-dependent perturbations, where the evolving state is exponentially localized in an energy window of the instantaneous spectrum. For spin glass models where the configuration spaces are separated by large energy barriers, the localization in energy spaces can induce a true localization in configuration spaces and robustly break ergodicity. We then demonstrate its applications in several systems with such barriers. For certain LDPC codes, we show that the system remains localized near the original codeword for an exponentially long time even under generic time-dependent perturbations. For classical optimization problems with clustered solution space, the stability becomes an obstacle for quantum Hamiltonian-based algorithms to escape local minima. Our work provides a new lens for analyzing non-equilibrium dynamics of generic quantum systems, and versatile mathematical tools for establishing stability and for designing quantum algorithms.
\end{abstract}

\maketitle
\end{CJK*}

\section{Introduction}

Over the past few decades, many new phases~\cite{anderson1958absence,RMP-ZooTopoOrder,RevModPhys.91.021001-MBL,nandkishore2019fractons,PhysRevLett.116.120401} in condensed matter physics have been identified and discovered. Stability against perturbations is often a requirement and a signal to define a new phase of matter.
Many of these phases, such as topologically ordered phases~\cite{PhysRevB.72.045141,michalakis2013stability,RMP-ZooTopoOrder}, many-body localization~\cite{imbrie2016many,RevModPhys.91.021001-MBL}, prethermalization~\cite{PhysRevLett.116.120401,abanin2017rigorous,else2017prethermal,PhysRevX.10.011043}, and some open quantum systems~\cite{PhysRevLett.120.040404,hong2025quantum,rakovszky2024bottlenecks},  have been proven to be stable against certain types of perturbations. Some of these stable phases have also been experimentally realized~\cite{wray2010observation,gring2012relaxation,choi2017observation,peng2021floquet,kongkhambut2022observation,guo2025observation}, but most theoretical proofs of their stability are limited to static perturbations, whereas most perturbations in reality are time-dependent. However, whether these phases can be stable against time-dependent perturbations is not known, and rigorous bounds and mathematical tools for analyzing generic time-dependent perturbations are lacking in the literature. For time-dependent perturbations, most existing methods~\cite{dyson1949s,lieb1972finite,kuwahara2016floquet} give bounds in power series of evolution time $t$, which will eventually diverge with increasing $t$ and thus can only control errors within short-time evolutions. These obstacles correspond to uncontrollable effects of generic time-dependent perturbations, such as high excitations and large heat absorption, which can drive the states globally away from their initial configurations. Thus, in general, we expect that many of those stable phases will no longer be robust against generic time-dependent perturbations under long-time driving.

Therefore, it is fundamentally important to ask: \textit{are there any systems that do possess provable stability against generic time-dependent perturbations}? 
In this work, we identify one property called \textit{energy-space localization} that survives under generic time-dependent perturbations, and can be used as a guide to search for dynamically stable models. For an initial eigenstate that evolves under a generic time-dependent Hamiltonian $H(t)$, it will become a superposition of instantaneous eigenstates and spread in the instantaneous energy space. If $H(t)$ is $q$-local (not necessarily geometrically local) and has a bounded local norm, we can rigorously prove that the spreading is exponentially localized in an energy window near the initial energy, with a width roughly being the total variation of $H(t)$, which is consistent with intuitions from classical mechanics. In addition, all the bounds we obtain are invariant under monotonic rescaling of $t$. This indicates that the far region of the spectrum remains inaccessible for all times to quantum evolution $H(t)$ with small total variation, regardless of how rapidly or slowly the $H(t)$ varies.

Since energy-space localization will be shown to hold ubiquitously for physical $q$-local Hamiltonians $H(t)$, itself alone may not necessarily yield additional non-trivial properties of the system. Indeed, according to the eigenstate thermalization hypothesis (ETH)~\cite{deutsch1991quantum,srednicki1994chaos,rigol2008thermalization,nandkishore2015many,deutsch2018eigenstate}, a state localized at one single energy level of a generic Hamiltonian can behave like a thermal ensemble. Thus, localization in an energy window alone does not prevent states from thermalization or ergodicity in generic systems. On the other hand, systems that can break ergodicity are highly nontrivial, especially when such breaking persists under generic perturbations. This stability is indeed the property of some exotic phenomena such as MBL~\cite{RevModPhys.91.021001-MBL} and pretheramlization~\cite{abanin2017effective,else2017prethermal}. However, these existing models and theoretical proofs are not applicable to generic time-dependent perturbations, and it is likely that they are no longer stable in this scenario. As we will prove that energy-space localization can survive under generic time-dependent perturbations for an arbitrarily long time, could it uncover new stability? How could stability emerge beyond energy space?

To search for a stable phase via energy-space localization, we need to relate the closeness in the energy space to that in the configuration space. This property indeed exists in certain spin glass models, which can be constructed from error-correcting codes~\cite{DinurGoodQuantum,PanteleevAsymptoticallygood,PRL-EnergyBarrierHyperCode,yin2024eigenstate} and hard classical optimization problems~\cite{gamarnik2021OGP-pnas}, where the low-energy eigenspace is well separated into distant clusters, each of which contains eigenstates that are close. In these cases, eigenstates require overcoming a significant energy barrier to reach other clusters. Using the energy-space localization that we shall derive, we can prove that the localization time in each cluster is exponentially long for linear barriers and superpolynomially long for sublinearly increasing barriers, provided the total variation of the time-dependent perturbation is smaller than half of the energy barrier. Such localization is a hallmark of the ergodicity breaking, and is stronger than the usual context of ergodicity defined in static Hamiltonians.

Such localization is useful for quantum error-correcting codes, especially when used as fault-tolerant quantum memory~\cite{knill1998resilient,dennis2002topological,wang2003confinement,bravyi2024high}. A central question in this context is whether the error-correcting codes can protect quantum information under generic perturbations for a long time. Although realistic noises are generally time-dependent, rigorously proving stability under generic time-dependent perturbations was considered theoretically intractable, as there are rarely theoretical tools that can bound unitary dynamics of $H(t)$ for long-time evolutions. Even for static perturbations, proofs for dynamical stability of quantum codes mostly reply on ground-state (or zero-temperature) stability~\cite{bravyi2010topological,yin2025low,de2025low}, which no longer works if the initial state already contains small correctable errors. Can erroneous initial states of certain qLDPC codes remain correctable after a long-time noisy evolution? Our results give an affirmative answer in the most general time-dependent setting: If the code exhibits linear energy barriers, errors induced by weak time-dependent perturbations remain correctable if the total variation imparted by the latter is small.

Our results can also be directly reduced to static cases by taking the time duration of time-dependent perturbations to the infinitesimal limit in a quench dynamics. This allows us to prove that the perturbed eigenstates are universally localized inside the energy-space of the unperturbed spectrum. One of its special cases is when the unperturbed Hamiltonian is mutually commuting. For such Hamiltonians, their energy-space localization was originally obtained in~\cite{yin2024eigenstate}, and was used as a stepstone to prove the infinite-time dynamical localization of cLDPC codes via proving the localization of perturbed eigenstates. However, it remains unknown whether qLDPC codes can have eigenstate localization; thus, proving dynamical localization of qLDPC codes from static eigenstate localization is still not feasible. In this work, we bypass this obstacle by directly proving the dynamical localization of qLDPC codes in time-dependent cases, which naturally include static perturbations and yield an exponentially long localization time regardless of whether the systems exhibit eigenstate localization. The static reduction can also reproduce many existing stability results~\cite{yin2024eigenstate,rakovszky2024bottlenecks} of spin glasses against static $q$-local perturbations and generalize them to quasi-$q$-local perturbations. 

On the other hand, localization can also pose an obstacle to Hamiltonian-based quantum algorithms aiming to solve difficult classical optimization problems~\cite{gamarnik2021OGP-pnas} whose near-optimal solution space is clustered. The hardness from their clustered solutions was previously leveraged to rule out logarithmic-time quantum algorithms~\cite{Moosavian2022limitsofshorttime}. Although it is believed that such NP-hard problems remain hard within polynomial-time quantum algorithms, no rigorous results are known to rule out algorithms longer than linear time. In this work, we will show that for Hamiltonian-based quantum algorithms, the system will become trapped in local minima for an exponentially long time due to the energy barriers if the algorithmic Hamiltonian's total variation is insufficient.

Finally, we will discuss how to experimentally detect the results obtained in this work. For energy-space localization itself, its exponentially small tail bounds are hard to directly measure. Instead, we can probe the localization behavior of systems with extensive energy barriers to confirm our theory in potential experimental platforms.

\section{Energy-space localization}
\label{sec:Energy-spaceLoc}

\begin{figure}[ht]
    \centering
    \includegraphics[width=0.49\textwidth]{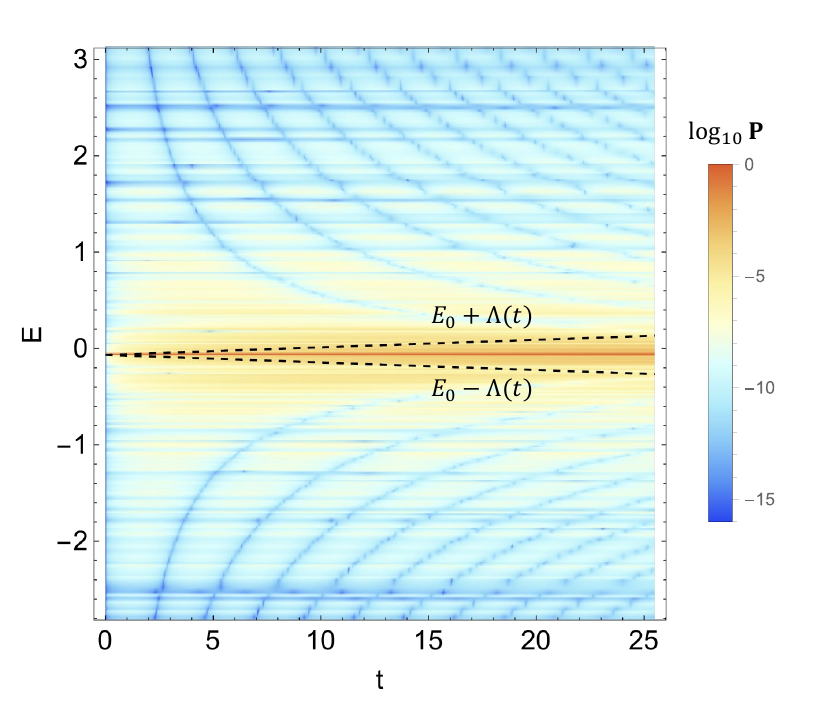}
    \caption{An illustrative example of the energy-space localization. $\Lambda(t)$ denotes the total variation of $H(t)$ from $0$ to $t$. The results are from simulating an 8-qubit $H(t)$ consisting of random $2$-local Pauli operators with all-to-all couplings, where the initial state is chosen to be one eigenstate of $H(0)$. The quantity $\mathbf{P}$ for $j$-th eigenstates of the instantaneous Hamiltonian $H(t)$ is its probability in expanding the evolving state $\ket{\psi(t)}$ in the instantaneous spectrum.}
    \label{fig:Energy-space-loc}
\end{figure}

In this section, we mainly work with a generic $q$-local time-dependent Hamiltonian $H(t)$ defined on an $n$-site lattice. To define a $q$-local Hamiltonian here, we require that $H(t)$ admits a partition $H(t)=\sum_A h_A(t)$ into $q$-local terms, where each $h_A(t)$ is supported in a region $A$ that contains at most $q$ sites. Here, we do not have any geometric constraints on the interacting region $A$, which can be spatially nonlocal. For example, the Ising model is 2-local and also geometrically local. All-to-all 2-body couplings are 2-local but not geometrically local. The Kitaev toric code~\cite{kitaev2003fault} is 4-local and geometrically local.

Suppose a state $\ket{\psi(t)}$ is initialized at one eigenstate of $H(0)$ at $t=0$, which can also be viewed as a $\delta$-function distribution of eigenstates in energy space. Let the state evolve under $H(t)$, then $\ket{\psi(t)}$ will become a superposition of instantaneous eigenstates as long as $H(t)$ at different times does not commute with one another. Thus the evolving state's distribution in the energy space spreads. Our main results are on the localization of such spreading: The distribution is exponentially localized in a classically allowed energy window (see Fig.~\ref{fig:Energy-space-loc}) for a generic physical $H(t)$, and the width of the energy window is determined by the total variation of $H(t)$ during the evolution. We will give formal statements on this argument shortly.

Before we proceed, we first introduce the requirements of a physical $q$-local $H(t)$ in our context. Firstly, we require that the local norm $\|H(t)\|_{\textnormal{loc}}$ is bounded. Here, the local norm defined by  $\|H\|_{\text{loc}}\equiv\sup_{i}\sum_{A,A\ni i}\|h_A\|$ in the local decomposition, which roughly measures the largest local strength of the Hamiltonian acting on a site. We also define the global termwise norm $\|H(t)\|_{\textnormal{X}}=\sum_A\|h_A(t)\|$ for convenience, which measures an absolute strength of the total Hamiltonian. We summarize the required conditions for $H(t)$ as follows. 
\begin{cas}[General]
\label{case:1}
$H(t)$ is $q$-local and $\|H(t)\|_{\textnormal{loc}}\leq M$, with $\int_0^T \|\mathrm{d}H\|_{\textnormal{X}}\equiv\int_0^T \|H'(t)\|_{\textnormal{X}}\mathrm{d}t\leq\lambda n$,
\end{cas}
\noindent where $\lambda$ and $M$ are constants independent of $n$ and we call $\int_0^T \|\mathrm{d}H\|_{\textnormal{X}}$ {\it total variation\/} of $H(t)$ from $t=0$ to $T$ (see the formal definitions in Appendix~\ref{app:normsTVs}). 

We list below other conditions that are not included in the Case~\ref{case:1}, but will be useful as they have correspondence in the static scenarios.
\begin{cas}
\label{case:2}
$H(t)$ is $q$-local and can be written as $H(t)=H_C(t)+V(t)$ with $[V(t),H'(t)]=0$ and $\|H_C(t)\|_{\textnormal{loc}}\leq M$~\footnote{Here, we do not require $H_C(t)$ to be mutually commuting, which is different from Cases~\ref{case:3} and \ref{case:4}.}, with $\int_0^T \|\mathrm{d}H\|_\mathrm{X}\leq\lambda n$.
\end{cas}
\begin{cas}
\label{case:3}
$H'(t)$ is $q$-local. $H(t)$ can be partitioned into $H(t)=H_C(t)+V(t)$ with $H_C(t)$ mutually commuting and $\|H_C(t)\|_{\textnormal{loc}}\leq M$, and $[V(t),H'(t)]=0$, with $\int_0^T \|\mathrm{d}H\|_\mathrm{X}\leq\lambda n$.
\end{cas}
\begin{cas}
\label{case:4}
$H'(t)$ is quasi-$q$-local. $H(t)$ can be partitioned into $H(t)=H_C(t)+V(t)$ with $H_C(t)$ mutually commuting and $\|H_C(t)\|_{\textnormal{loc}}\leq M$, and $[V(t),H'(t)]=0$, with $\int_0^T \|\mathrm{d}H\|_{q\textnormal{-X}}\leq\lambda n$,
\end{cas}
\noindent where we define a Hamiltonian to be quasi-$q$-local if the strength of its $k$-local terms decays with $e^{-k/q}$, and the termwise $q$-norm $\|\cdot\|_{q\textnormal{-X}}$ for quasi-$q$-local Hamiltonian as the total norm with weights $e^{k/q}$ for each $k$-local term. Although intrinsic quasi-$q$-local interactions have not yet been found in reality, they arise naturally in many perturbative analyses and effective models ~\cite{abanin2017effective,abanin2017rigorous,else2017prethermal}. 

\begin{table}[t]
{\renewcommand{\arraystretch}{1.5}
\begin{tabular}{|c|c|}
        \hline
        Norms & Definitions\\[1pt]
        \hline 
        $\|H\|$  & $\sup_{\ket{\psi}\neq 0} \frac{\|H \ket{\psi}\|}{\|\ket{\psi}\|}$ \\ 
        $\|H\|_{\textnormal{loc}}$ & $\sup_{i}\sum_{A,A\ni i}\|h_A\|$ \\
        $\|H\|_{\mathrm{X}}$ & $\sum_A \|h_A\|$ \\
        $\|H\|_{q\textnormal{-X}}$ & $(q+1)\sum_{A}e^{|A|/q}\|h_A\|$\\
        \hline 
\end{tabular}}
\caption{Summary on definitions of different norms used in this work for a given decomposition of a Hamiltonian $H=\sum_A h_A$. In the definition of $\|H\|_{q\textnormal{-X}}$, $|A|$ denotes the size of the region $A$, and $q$ can be chosen from any positive real numbers. We refer the reader to Appendix~\ref{app:normsTVs} for detailed definitions and properties of different norms. We define $H$ as quasi-$q_{\star}$-local if one can find a $q_{\star}$ (usually the minimum of all possible $q$) such that $\frac{1}{n}\|H\|_{q_{\star}\textnormal{-X}}$ can be bounded by an $O(1)$ constant.}
\label{tab:norms}
\end{table}

We refer the reader to Table~\ref{tab:norms} and Appendix~\ref{app:normsTVs} for the explicit definitions and properties of different norms and total variations. With these preparations, we can present our main results, which we call \textit{energy-space localization}.

\begin{thm}[Energy-space localization, informal]
    \label{thm:Energy-local}
     We set an initial state $\ket{\psi(0)}$ as an eigenstate of $H(0)$ with energy $E_0$. If we let the state evolve according to $H(t)$ from $t=0$ to $T$, then the state at any $t$ is exponentially localized in the energy window $\mathcal{E}_0(d)\equiv[E_0-dn,E_0+dn]$ in the instantaneous spectrum of $H(t)$, where $d$ can be any $\Theta(1)$ constant satisfying $d>\lambda$. 
     
     \smallskip \noindent (I) For $H(t)$ defined in \textnormal{Case~\ref{case:1}},~\ref{case:2} and \ref{case:4}, the leakage $\epsilon(d)$ (defined as the norm of outside eigenstates) to instantaneous eigenstates with energy differences $|E(t)-E_0|\geq dn$ is bounded by
\begin{equation}
\label{eq:epsilond1f}
    \epsilon^{(1)}_{\lambda,\Delta}(d)\equiv \exp\left(-n\frac{\lambda}{\Delta}\left(\frac{d}{\lambda}-1-\ln\frac{d}{\lambda}+o(1)\right)\right),
\end{equation}
where $\Delta=\Delta_q\equiv 2qM$.

\smallskip
\noindent (II) If $H(t)$ satisfies Case~\ref{case:3}, the bound can be improved to
\begin{equation}
\label{eq:epsilond2f}
    \epsilon^{(2)}_{\lambda,\Delta}(d) \equiv \exp\left(-n\frac{\lambda}{\Delta}\left(\frac{d}{\lambda}\ln\frac{d}{\lambda}-(\frac{d}{\lambda}-1)\right)\right),
\end{equation}
where $\Delta=\Delta_q$.
\end{thm}

We present the formal statements and rigorous proofs in Appendix~\ref{app:pEnergyLoc} and provide more extensions and corollaries in Methods~A (\ref{methods:extensionForESL}). Given the exponential decay of the bound for any $d>\lambda$, we can interpret from the theorem that eigenstates outside $\mathcal{E}_0(\lambda)$ have almost no contribution to the final state $\ket{\psi(T)}$ in the thermodynamic limit (see Fig.~\ref{fig:Energy-space-loc}). On the other hand, it is unlikely to give universal bounds inside $\mathcal{E}_0(\lambda)$, as one can always choose $V(t)=E'-E_0$ to precisely pinpoint the final state at any energy $E'\in \mathcal{E}_0(\lambda)$. 

The energy window $\mathcal{E}_0(\lambda)$ is consistent with classical intuitions: The energy change of a classical $\tilde{H}(t)$ is exactly $\int_0^T\tilde{H}'(t)\mathrm{d}t$ and can be bounded by $\int_0^T|\tilde{H}'(t)|\mathrm{d}t \leq \lambda n$ \footnote{As $\tilde H(t)$ is classical here, the norm is simply the absolute value.}. Thus, the energy window $\mathcal{E}_0(\lambda)$ corresponds to the classically reachable region from the driving $H(t)$. At the thermodynamic limit ($n\rightarrow \infty$), the energy-space localization of quantum $H(t)$ is naturally reduced to classical mechanics: The contributions from classically forbidden region (outside $\mathcal{E}_0(\lambda)$) in the final state vanish exponentially.

As the leakage bound is controlled by the total variation of $H(t)$, it is invariant under any monotonic rescaling of any part of the evolution time. Thus, it is straightforward to reduce Theorem~\ref{thm:Energy-local} to static perturbations $H=H_0+V_0$, which is sketched in Methods~B (\ref{methods:staticReduce}). In this scenario, a common task is to study the properties of the eigenstates of $H$ in the {\it unperturbed} eigenbasis of $H_0$, which is often easier to handle. We consider the following general static case:

\begin{casbis}{case:2}[General]
\label{case:2'}
    $H_0$ is $q$-local and $\|H_0\|_{\textnormal{loc}}\leq M$, with $q$-local $\|V_0\|_{\textnormal{X}}\leq\lambda n$.
\end{casbis}

The Case~\ref{case:2'} can be obtained by setting $H_C(t)=H_0$, $V(t)=(1-t/T)V_0$ and let $T\rightarrow0^+$ in Case~\ref{case:2}. Other time-dependent cases can also be similarly translated into static cases as follows:

\begin{casbis}{case:3}
\label{case:3'}
    $H_0$ is mutually commuting and $\|H_0\|_{\textnormal{loc}}\leq M$, with $q$-local $\|V_0\|_{\textnormal{X}}\leq\lambda n$.
\end{casbis}
\begin{casbis}{case:4}
\label{case:4'}
    $H_0$ is mutually commuting and $\|H_0\|_{\textnormal{loc}}\leq M$, with quasi-$q$-local $\|V_0\|_{q\textnormal{-X}}\leq\lambda n$.
\end{casbis}
These are the counterparts of Cases 3 and 4 in time-dependent cases. We summarize the energy-space localization for the static cases in the theorem below.

\begin{thm}[Energy-space localization, static cases]
\label{thm:static-Energy-loc}
Let $\ket{\psi}$ be an eigenstate of the combined Hamiltonian $H=H_0+V_0$ with energy $E'$. If $H_0$ and $V_0$ satisfy any of the above cases, then in the eigenbasis of $H_0$, $\ket{\psi}$ is exponentially localized in the energy window $[E'-dn,E'+dn]$, where $d$ can be any $\Theta(1)$ constant satisfying $d>\lambda$, with leakage bounded by $\epsilon^{(1)}_{\lambda,\Delta_q}(d)$ for Case~\ref{case:2'} and Case \ref{case:4'}, and $\epsilon^{(2)}_{\lambda,\Delta_q}(d)$ for mutually commuting $H_0$ (Case \ref{case:3'}).
\end{thm}

We will show a simple proof for this theorem in Methods~B (\ref{methods:staticReduce}) by taking the evolution time in Theorem~\ref{thm:Energy-local} to infinitesimal limit.

We remark that although the energy-space localization is consistent with classical intuitions, its universality and the explicit tail bounds have not been obtained before, even in generic static cases.  Proposition 2 in the Supplementary Material of Ref.~\cite{yin2024eigenstate} obtained a similar bound to $\epsilon^{(2)}_{\lambda,\Delta_q}$ for mutually commuting $H_0$ by rewriting the Hamiltonian in a tri-diagonal form, and the result can be regarded as a special case (Case~\ref{case:3'}) in our Theorem~\ref{thm:static-Energy-loc}. However, it is not clear how to generalize the tri-diagonal method to non-commuting $H_0$ or quasi-$q$-local perturbations in Case~\ref{case:2'} and Case~\ref{case:4'}, whereas our results bypass this obstacles via reducing from time-dependent cases in Theorem~\ref{thm:Energy-local}.

Given the universality of energy-space localization, one wonders whether there are specific scenarios in which it can be used to infer stronger stability. Surprisingly, when combined with a clustered energy landscape where the closeness in energy can infer closeness in configuration space, our energy-space localization can become a powerful tool to prove numerous stabilities for such systems. This clustering property has been proven in many spin glass models~\cite{mezard2009information,gamarnik2017limits,gamarnik2021overlap,placke2024topological}, which we define formally in our context as follows:

\begin{figure}[t]
    \centering
    \includegraphics[width=0.46\textwidth]{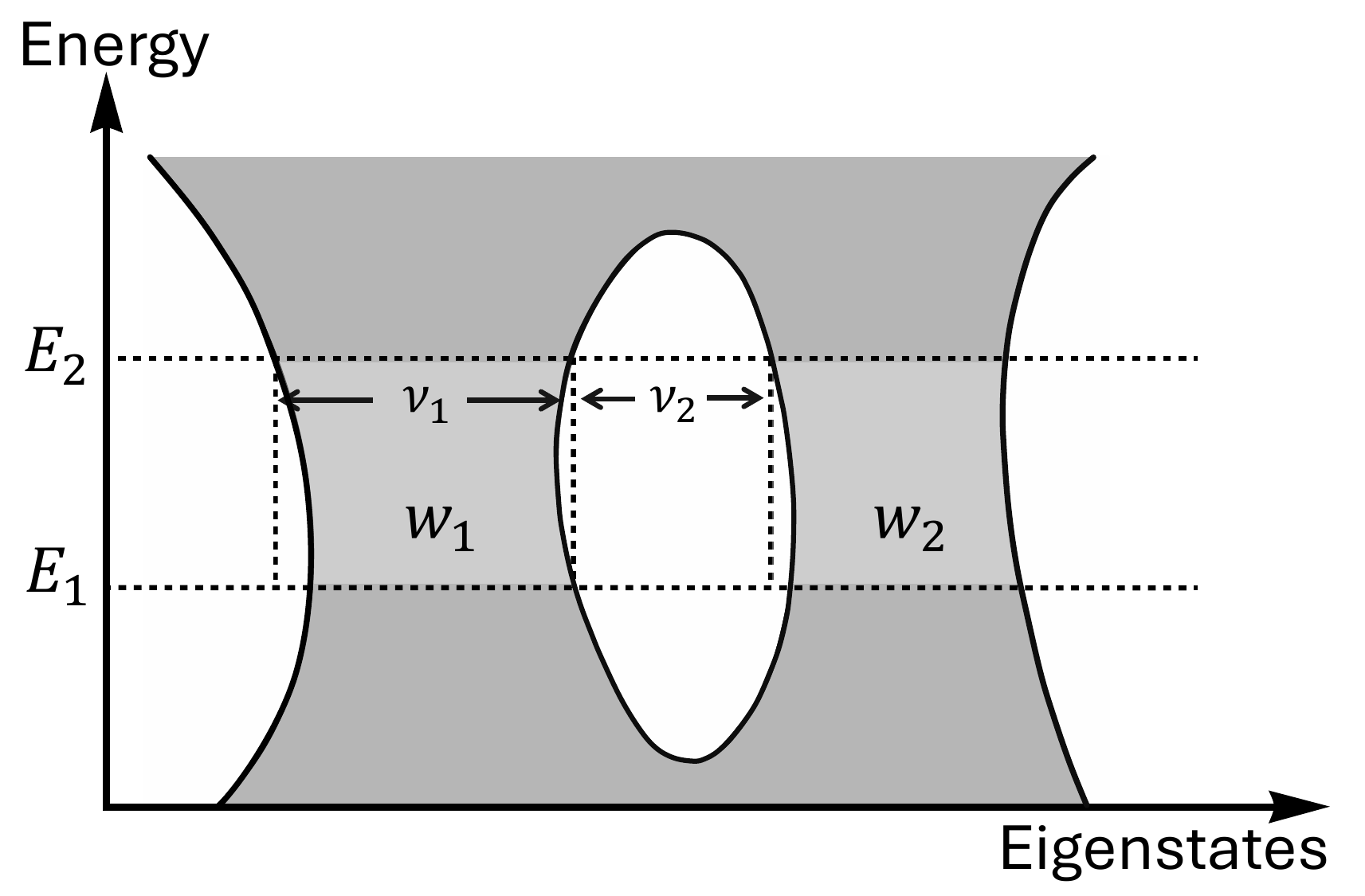}
    \caption{An illustrative example of the clustering property in the energy window $[E_1,E_2]$, where the gray area denotes all the eigenstates. We compress the high-dimensional eigenstates to one dimension for visual convenience, and the distance along the x-axis roughly represents the distance $\mathbf{D}$ between two eigenstates.}
    \label{fig:Cluster}
\end{figure}

\begin{defi}[Clustering property]
\label{def:cluster}
    A Hamiltonian $H$ has the clustering property in an energy window $[E_1,E_2]$ if all eigenstates in this energy window can be divided into clusters $\{w_j\}$ such that: For any two eigenstates $\ket{\phi}$ from $w_j$ and $\ket{\phi'}$ from $w_{j'}$, the distance satisfies $\mathbf{D}(\ket{\phi},\ket{\phi'})\leq\nu_1$ if they are within the same cluster $w_j=w_{j'}$ and $\mathbf{D}(\ket{\phi},\ket{\phi'})\geq\nu_2$ if they come from different clusters $w_j\neq w_{j'}$.
\end{defi}
\noindent Here, we define the distance $\mathbf{D}$ as follows: $\mathbf{D}(\ket{\phi},\ket{\phi'})=D$ if $k$-local operators $V$ cannot produce a nonzero $\braket{\phi|V|\phi'}$ with $k<D$ but can produce with $k=D$. For distances between $Z$-basis states in qubit systems, the above definition is equivalent to the Hamming distance of classical bits.

In this definition, the interval $[E_1,E_2]$ acts as an energy barrier for different clusters (see Fig.~\ref{fig:Cluster}). We require $\nu_2$ to be sufficiently large compared to the locality of the perturbations~\footnote{When considering $q$-local perturbations, $\nu_2>q$ is usually sufficient. For quasi-$q$-local perturbations, we require $\nu_2 \sim \Theta(n)$.  Although we mainly focus on models with $\nu_2\sim \Theta(n)$ in this work, the results can easily extend to other scalings of $\nu_2$.}. Although we have no restriction on $\nu_1$, for localization to be well-defined, typically $\nu_1$ needs to be small compared to $n$, e.g., $\nu_1<\delta n$ with $\delta\ll 1$. We will show in Sec.~\ref{sec:LDPC} that if the width $B\equiv E_2-E_1$ of the energy window is linear with $n$, the states within each cluster exhibit stability against various types of perturbations, ranging from static to time-dependent, strictly $q$-local to quasi-$q$-local, and from close to open systems (see Table~\ref{tab:LDPC-results}).

This property can appear in certain error-correcting codes and hard classical optimization problems. The stability is useful for quantum error-correcting codes, where the trapped states correspond to errors that can be easily detected and corrected to the original codeword. On the other hand, such stability can also pose an algorithmic barrier for optimization problems whose solution space exhibits the clustering property, where the evolution under weak driving cannot drive the trapped states out of local minima. We will discuss the two classes of models in Sec.~\ref{sec:LDPC} and Sec.~\ref{sec:LocClassicalOpt}, respectively.

\section{Localization as stability in LDPC codes}
\label{sec:LDPC}

\begin{table*}[t]
{\renewcommand{\arraystretch}{1.2}
\begin{tabular}{|c|c|c|c|}
        \hline 
        & Dynamical localization  & Eigenstate localization & Slow mixing of Gibbs sampler \\ 
        Perturbation type & (localization time) & (leakage)  & (mixing time) \\
        \hline 
        Static strictly $q$-local & $e^{\Omega(n)}$ [\ref{pro:DyLocalt}], c: $\infty$~\cite{yin2024eigenstate} &c: $e^{-\Omega(n)}$~\cite{yin2024eigenstate}  & $e^{\Omega(n)}$~\cite{rakovszky2024bottlenecks,gamarnik2024slow}\\
        Static quasi-$q_{\star}$-local&  $e^{\Omega(n)}$ [\ref{pro:DyLocalt}], c: $\infty$ [\ref{pro:DyLocalInf}]  &c: $e^{-\Omega(n)}$[\ref{pro:Eigen-local}] & $e^{\Omega(n)}$ [\ref{pro:SlowMixQuasi}] \\
        Time-dependent (quasi-)$q$-local &$e^{\Omega(n)}$~[\ref{pro:DyLocalt}] &- & -\\
        \hline 
\end{tabular}}
\caption{Summary on the stabilities of LDPC codes with linear soundness discussed in this work. The notation c denotes that the bound only works for cLDPC codes. Otherwise, the bound works for both cLDPC and qLDPC codes. More detailed bounds for parameters can be found in Table~\ref{tab:summary-results2} in the Appendix.}
\label{tab:LDPC-results}
\end{table*}

\begin{figure}[ht]
    \centering
    \includegraphics[width=0.46\textwidth]{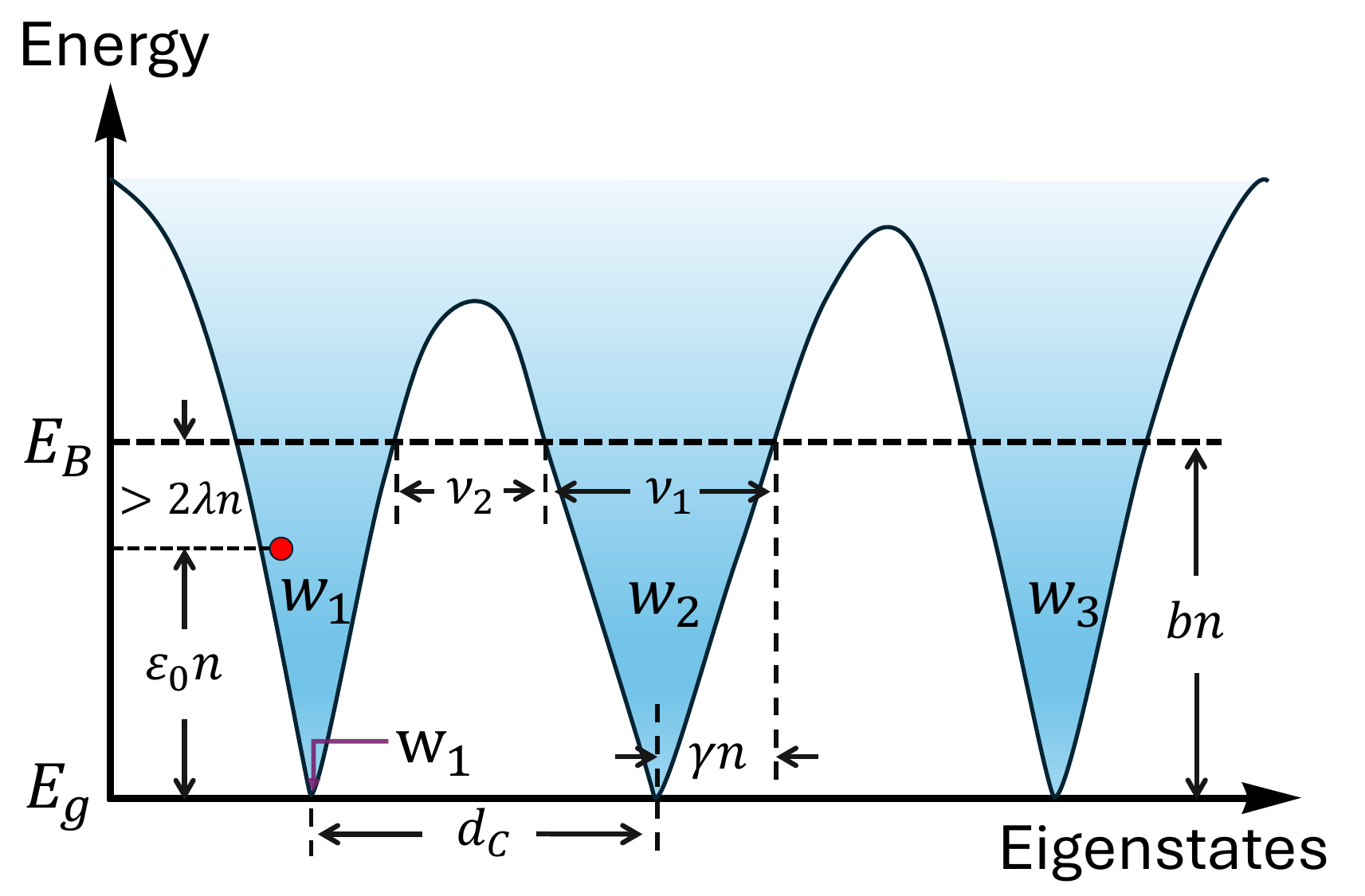}
    \caption{Illustration of the clustering property in classical/quantum LDPC codes with linear soundness, where the blue area denotes all the eigenstates. According to Proposition~\ref{pro:DyLocalt}, when time-dependent perturbations with total variation bounded by $\lambda n$ are added, any initial eigenstates (e.g., the red dot) below $E_B-2\lambda n$ are localized within their clusters for exponentially long time.}
    \label{fig:LDPC}
\end{figure}

LDPC codes are widely used error-correcting codes in both classical and quantum computing. The CSS-type quantum LDPC (qLDPC) codes can be written in the following Hamiltonian form
\begin{equation}
\label{eq:ldpcHC}
    H_C=-\left(\sum_{c^X}\prod_{i\in c^X} X_i+\sum_{c^Z}\prod_{i\in c^Z} Z_i\right),
\end{equation}
where $c^X$ and $c^Z$ are parity checks for $X$ and $Z$ operators, respectively, and they are mutually commuting. In addition, the low-density property requires that the number of checks involving each site is bounded by an $ O(1)$ constant $p_C$, so that we can set $M=p_C$ in $\|H_C\|_{\textnormal{loc}}\leq M$. The classical LDPC (cLDPC) codes can be regarded as a special case, whose $H_C$ contains only $Z$ checks.

The different ground states $\ket{\textnormal{w}}$ of the LDPC codes, usually called codewords, are degenerate and well separated in terms of distance $\mathbf{D}$, so they can robustly encode information. When errors occur, the energy of the erroneous state increases. All excited eigenstates $\ket{\mathrm{v}}_{\mathrm{w}}$ for qLDPC codes can be defined by excitations from a codeword $\ket{\mathrm{w}}$ as 
\begin{equation}    \ket{\mathrm{v}}_{\mathrm{w}}\equiv\prod_{i,v^X_i=1}X_i\prod_{j,v^{Z}_j=1} Z_j \ket{\mathrm{w}},
\end{equation}
where $\mathrm{v}\equiv\{v^X_1,..v^X_n,v^Z_1,..v^Z_n\}$ is a length-${2n}$ bit string denoting local Pauli errors. We can also define $\ket{\mathrm{v}}_{\mathrm{w}}\equiv\prod_{j,v^{Z}_j=1} Z_j \ket{\mathrm{w}}$ with $\mathrm{v}\equiv\{v^Z_1,..v^Z_n\}$ for cLDPC codes.

To efficiently recover the original state, a preferred property for error-correcting codes is the linear soundness 
\begin{equation}
    \mathbf{H} (\ket{\mathrm{v}}_{\mathrm{w}})\geq \alpha  \mathbf{D} (\ket{\mathrm{v}}_{\mathrm{w}},\ket{\mathrm{w}}), ~\text{if}~\mathbf{D} (\ket{\mathrm{v}}_{\mathrm{w}},\ket{\mathrm{w}})\leq \gamma n,
\end{equation}
where $\mathbf{H}$ counts the total number of check violations in $\ket{\mathrm{v}}_{\mathrm{w}}$.
This ensures that if the erroneous state violates only a few checks, it stays near the original codeword but still far from other codewords (see Fig.~\ref{fig:LDPC}), which is desirable for efficient decoding algorithms~\cite{DinurGoodQuantum}. In addition, the above property naturally gives rise to an energy window $[E_g,E_B]$ with $E_g$ being the ground energy and $E_B$ being the top for all clusters
\begin{equation}
    E_B\equiv E_g+bn,~b\equiv2\alpha\gamma,
\end{equation}
where we can define the energy barrier as $B\equiv E_B-E_g=bn$. In the energy window $[E_g,E_B]$, one can prove that the clustering property~\ref{def:cluster} is satisfied with
\begin{equation}
    \nu_1=2\gamma n,~\nu_2=d_C-2\gamma n,
\end{equation}
where $d_C>2\gamma n$ is a natural decodability constraint. Here, each cluster $w$ is defined by eigenstates below energy $E_B$ and within distance $\nu_1$ from a codeword $\ket{\textnormal{w}}$. This induces a linear energy barrier in the configuration space and provides stability against numerous kinds of perturbations (see Table~\ref{tab:LDPC-results}). Below, we mainly discuss two common scenarios: 1. dynamical stability under unitary evolution; 2. thermodynamic stability under quantum channels, or equivalently, static Lindbladian evolution.

\smallskip\noindent{\bf Dynamical stability.} Dynamical localization is highly nontrivial but can exist in spin glass~\cite{rademaker2020slow,PhysRevX-Metastable}. Most of the existing works discuss dynamics under static perturbations. In this section, we focus on {\it time-dependent\/} perturbations, which is a fundamentally different scenario but has easy reductions to static perturbations.

We consider a common and realistic scenario in a self-correcting quantum memory~\cite{dennis2002topological} or other quantum devices: The system is initially an eigenstate $\ket{\psi(0)}$ with energy $E_0$ of $H_C$, a good classical or quantum LDPC code with linear soundness. The state has an energy density $\varepsilon_0\equiv\frac{E_0-E_g}{n}$ above ground energy and is close to one codeword $\ket{\mathrm{w}_0}$ (see Fig.~\ref{fig:LDPC}), which represents an erroneous state with correctable errors. Then a time-dependent perturbation $V(t)$ is added and the state $\ket{\psi(t)}$ evolves under the Hamiltonian, 
\begin{equation}
    H(t)=H_C+V(t).
\end{equation}
We may ask: How long can the evolving state $\ket{\psi(t)}$ stay near the original codeword $\ket{\mathrm{w}_0}$ and remain correctable? Such stability under generic time-dependent perturbations is usually considered theoretically intractable, and we are not aware of any prior theoretical study on this question. Below, we show that by applying Theorem~\ref{thm:Energy-local}, we can rigorously prove that the erroneous state remains correctable for an exponentially long time.

\begin{pro}[Exponentially long dynamical localization under time-dependent perturbations]
\label{pro:DyLocalt}
Let $\ket{\psi(t)}$ be $|\psi(0)\rangle$ evolving under $H(t)=H_C+V(t)$, where $H_C$ is defined in Eq.~\eqref{eq:ldpcHC} and $V(t)$ is a (quasi-)$q$-local perturbation. Suppose $\ket{\psi(0)}$ is initialized at an eigenstate of $H_C$ inside a cluster $w_0$, whose energy density $\varepsilon_0$ above the ground state satisfies $\varepsilon_0<b-2\lambda$ and $\lambda$ is defined in Eq.~\eqref{eq:dyLambda-q}. Then for all Cases~\ref{case:1}-\ref{case:4}, $|\psi(t)\rangle$ almost resides within the cluster $w_0$ up to an exponentially long time $T\sim \frac{1}{\lambda} e^{\Omega(n)}$, with an exponentially small leakage $e^{-\Omega(n)}$ for all $t\leq T$.
\end{pro}

We present a sketch of the proof in Methods~ C (\ref{methods:dyLoc}) and the details of the proof in Appendix~\ref{app:pDynamicLoc}. The results demonstrate that quantum and classical error-correcting codes with a linear energy barrier can protect information against generic time-dependent perturbations for an exponentially long time, as long as the total variation induced by the perturbation is bounded by $\lambda n$. This includes two common scenarios: 1. perturbations are global but weak, such as unstable global pulses applied to every site. 2. perturbations are local but strong, such as cosmic rays~\cite{mcewen2022resolving,tan2026resilience,sriram2025rare}. This property defines a stable phase that is robust against the most general and realistic time-dependent perturbations.

The results can also be directly reduced to static perturbations by setting $V(t)$ to be fixed.
In addition, if the perturbation is weaker so that  eigenstate localization~\cite{yin2024eigenstate} can appear in a cLDPC system~\footnote{It remains open whether eigenstate localization can happen in qLDPC codes.}, the dynamical localization time can become infinite (at finite $n$). We will describe the details in Methods~D (\ref{methods:eigenLoc}) and provide rigorous proofs in Appendix~\ref{app:pEigenLoc}. In contrast to the original work~\cite{yin2024eigenstate} on eigenstate localization, our methods can straightforwardly generalize the $q$-local perturbations to quasi-$q$-local.

\smallskip\noindent{\bf Thermodynamic stability.} In reality, quantum devices inevitably interact with environments and become open systems, whose equilibrium state can be described by the Gibbs state $\rho_H=\frac{1}{\mathcal{Z}}e^{-\beta H}$. For systems weakly coupled to a large and stable environment, the evolution of the density matrix can be described by a static (i.e., time-independent) Lindbladian. This can also be equivalently described by repetitively applying a Markovian quantum channel $\mathcal{M}$, which is more common in the design of Gibbs samplers~\cite{chen2023efficient,PhysRevLett.134.140405-MixingTime,ding2025efficient}.

Suppose that the system is initialized at one codeword. If the evolving state quickly reaches thermal equilibrium, then all codewords will be mixed together and the encoded information will be lost. Thus, slow mixing indicates the thermodynamic stability of LDPC codes.

For LDPC codes with linear soundness, the mixing time is known to be exponentially long~\cite{hong2025quantum}. In addition, even if strictly $q$-local perturbations are added $H=H_C+V_0$, the slow mixing property of Gibbs samplers persists~\cite{rakovszky2024bottlenecks,gamarnik2024slow}. We now extend the results to quasi-$q$-local perturbations via Theorem~\ref{thm:static-Energy-loc}.

\begin{pro}[Robust slow mixing of Gibbs sampler with (quasi-)$q$-local perturbations]
\label{pro:SlowMixQuasi}
Let $H=H_C+V_0$, with $H_C$ defined in Eq.~\eqref{eq:ldpcHC} and $V_0$ being a (quasi-)$q$-local perturbation. If $\lambda$ (and $q_{\star}$) defined in Cases.~\ref{case:2'}-\ref{case:4'} is sufficiently small, any \mbox{(quasi-)}$q$-local Gibbs sampler $\mathcal{M}$ whose steady state is $e^{-\beta H}$ has an exponentially long mixing time $e^{\Omega(n)}$.
\end{pro}

We show the proof in Appendix~\ref{app:pGibbs} by invoking the quantum bottleneck theorem developed in~\cite{rakovszky2024bottlenecks}. A key step in this proof is to upper bound the low-energy leakage in high-energy eigenstates, which can be exponentially bounded by Theorem~\ref{thm:static-Energy-loc}.

\section{Localization as an algorithmic barrier in hard optimization problems}
\label{sec:LocClassicalOpt}

Classical optimization problems are widely studied in computer science, physics, and many other areas. Among them, many are classified as NP-hard in the worst cases. Although some NP-hard problems admit efficient approximate algorithms~\cite{goemans1995improved,PhysRevLett.35.1792-SKSpinGlass,montanari2025optimization}, some problems in certain parameter regions, however, are difficult to find approximate solutions even in average cases. These include $p$-spin glasses~\cite{chen2019suboptimality,gamarnik2021overlap,Anschuetz2024combinatorialnlts}, max $k$-SAT~\cite{doi:10.1126/science.1073287-parisi,achlioptas2011solution,ProofkSat,bresler2022algorithmic,Anschuetz2024combinatorialnlts} and maximum independent sets~\cite{coja2015independent,gamarnik2017limits}.

Such a hardness of approximation is closely related to the solution space geometry~\cite{prl-clusteringSolRSP,gamarnik2021OGP-pnas} of these problems (see Fig.~\ref{fig:OGP}), which is similar to that of error-correcting codes: States with energy below a certain threshold are well-separated into clusters. A way to describe such a feature is the overlap gap property (OGP)~\cite{gamarnik2021OGP-pnas}: Suppose that the optimal value for an optimization problem $\mathcal{L}(\mathbf{z})$ is $E_g=\min_{\mathbf{z}}\mathcal{L}(\mathbf{z})$, and let $\mathbf{z}_1$ and $\mathbf{z}_2$ be two near-optimal configurations such that $\mathcal{L}(\mathbf{z}_{1/2})\leq E_g +\tilde{B}$. The OGP requires that $\mathbf{D}(\mathbf{z}_{1},\mathbf{z}_{2})$ be either less than $\tilde{\nu}_1$ or more than $\tilde{\nu}_2$, where $0\leq \tilde{\nu}_1<\tilde{\nu}_2$. For most models~\cite{coja2015independent,chen2019suboptimality,Anschuetz2024combinatorialnlts}, $\tilde{\nu}_1$ will approach 0 near the ground states, so that one can achieve $\tilde{\nu}_2>2\tilde{\nu}_1$ by choosing a proper $B<\tilde{B}$. In this case, one can verify that the clustering property \ref{def:cluster} is always satisfied with $\nu_1=\tilde{\nu}_1$ and $\nu_2=\tilde{\nu}_2$ in the energy window $[E_g,E_B]$ with $E_B\equiv E_g+B$.

\begin{figure}[ht]
    \centering
    \includegraphics[width=0.46\textwidth]{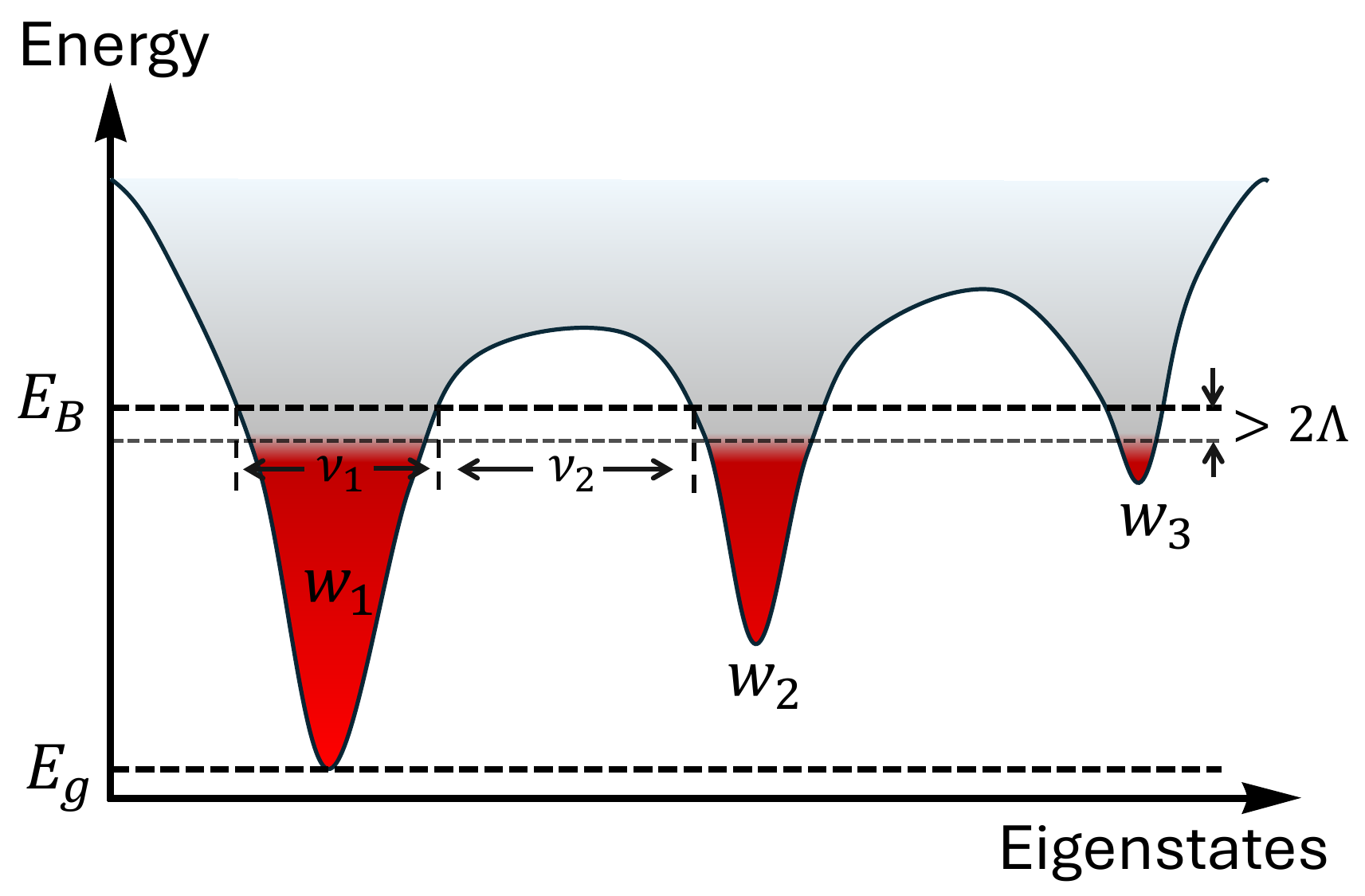}
    \caption{Illustration of the clustering property in hard optimization problems, where all eigenstates are simply $Z$-basis states. If the total variance of the algorithmic driving is $\Lambda$, then states localized below $E_B-2 \Lambda$ (red zones) cannot escape from their clusters. In most cases, the number of the shallower local minima, like $w_3$, is much larger than the deeper minima, like $w_1$ and $w_2$.}
    \label{fig:OGP}
\end{figure}

\smallskip\noindent{\bf Ineffectiveness of Hamiltonian-based quantum algorithms.} Quantum algorithms that try to solve these hard optimization problems have been extensively studied~\cite{doi:10.1126/science.1057726-AQC,farhi2014quantum,knysh2016zero,benjamin2017measurement,PRX-QOptSpinGlass,PRXQuantum.5.030348-SolBolSATQAOA,yu2021quantum,ebadi2022quantum,PRL-ScalingAdvantageQA,king2025beyond} in both theoretical and experimental works. In this section, we focus on a class of quantum algorithms that use Hamiltonian evolution for solution search and discuss the effects of a clustered solution space on these algorithms. For a classical (unconstrained) binary optimization problem $\min_\mathbf{z} \mathcal{L}(\mathbf{z})$, we can always replace the classical binary variables $\textnormal{z}$ with $\frac{1-Z}{2}$ and define the corresponding Pauli $Z$-basis qubit Hamiltonian $H_{\mathcal{L}}$, whose ground state encodes the optimal solution of $\mathcal{L}$. The algorithmic Hamiltonian we consider is defined by
\begin{equation}
\label{eq:HdefOGP}
    H(t)=H_{\mathcal{L}}+V(t),
\end{equation}
where $V(t)$ is a $q$-local Hamiltonian aiming to drive the evolving state to near-ground states of $H_{\mathcal{L}}$, and often vanishes at the end of the algorithm. Many generic quantum algorithms contain this type of evolution, such as quantum (diabatic) annealing~\cite{PhysRevE.58.5355-QAnnealing,crosson2021prospects}, quantum adiabatic algorithm~\cite{farhi2000quantum,doi:10.1126/science.1057726-AQC,RevModPhys.90.015002-AQCreview} and quantum Hamiltonian descent~\cite{leng2023quantum}.

The general performance analysis for these algorithms has not been established. Our results on dynamical localization partially fill the gap in the region where changes in $V(t)$ are sufficiently small. In this case, any near-optimal solution is stuck in each cluster and cannot be improved by any algorithmic evolution $H(t)$ with a total variation less than $\Lambda$ up to a superpolynomial time. This can be summarized as:

\begin{pro}[Freezing of solutions]
\label{pro:freezing}
If the total variation $\Lambda$ (defined in Eq.~\eqref{eq:dyLambda-q}) of the algorithmic Hamiltonian $H(t)$ (defined in Eq.~\eqref{eq:HdefOGP}) in a time interval is below $B/2$ and $\|H(t)\|_{\textnormal{loc}}\leq M$, then any $Z$-basis state initially inside a cluster $w_0$ with energy below $E_B-2\Lambda$ is localized inside $w_0$ up to $T\sim \frac{1}{\Lambda}e^{\Omega(\Lambda/M)}$.
\end{pro}

The results can be directly inferred from the dynamical localization~\ref{pro:DyLocalt} and rule out the success of those quantum algorithms without sufficiently large variations to find the near-optimal results. Here, we can choose the problem of finding the (near) ground state energy of a $q$-spin glass~\cite{Anschuetz2024combinatorialnlts} on a random regular $p$-degree $q$-uniform hypergraph $G$ via the adiabatic algorithm as an example, whose problem Hamiltonian can be written as  
\begin{equation}
    H_{\mathcal{L}}\equiv\sum_{\braket{i_1,i_2,...i_q}\in  G} J_{i_1,i_2,...i_q} Z_{i_1}Z_{i_2}...Z_{i_q},
\end{equation}
where $q\geq 4$ is the size of an hyperedge and each $J_{i_1,i_2,...i_q}$ can be chosen independently and randomly from $\{-1,1\}$. It was demonstrated in~\cite{Anschuetz2024combinatorialnlts} that the model has the clustering property in a linear energy window $[E_g, E_g+B]$ near the ground states. The corresponding adiabatic algorithm can be given via the following time-dependent Hamiltonian, 
\begin{equation}
\label{eq:adiaEvo}
    H(t)=s(t) H_{\mathcal{L}}+\left(1-s(t)\right)H_M,
\end{equation}
where $H_M$ is a coupling Hamiltonian whose ground state can be easily prepared, e.g., $H_M=\sum_j X_j$, and $s(t)$ is a monotonic function increasing from 0 to 1. As $\|H_{\mathcal{L}}\|_{\textnormal{loc}}= p$ (the degree of the hypergraph), the parameters $M$ and $\Lambda,B\sim \Theta(n)$ are still well defined in Theorem~\ref{thm:Energy-local} and Proposition~\ref{pro:DyLocalt}, where the total variation $\Lambda$ for the evolution from some finite $s>0$ to $1$ can be defined as $\Lambda=\frac{1-s}{s}\|H_M\|_{\textnormal{X}}$. On the other hand, when $s$ is close to 0, the corresponding total variation of the adiabatically driving is typically not small, whose effects may not fall into the scope of our theorems.

In~\cite{doi:10.1073/pnas.1002116107-LocMakeAdiaFail}, it was demonstrated from perturbative arguments that at some vanishing (when $n\rightarrow\infty$) $\tilde{\lambda}$ the $H(t)$ will have an exponentially small gap between ground states and excited states in the worst cases, which means adiabatic algorithms can take an exponentially long time to find the exact ground state of a worst-case $H_{\mathcal{L}}$. Since finding the exact minimum is unlikely, if we reduce our goal from exact ground states to low-energy states with energy much below $E_B$, can evolution with polynomial time produce good results? Unfortunately, our results give a negative answer: For $s(t)\geq s_{\star} $ and $\lambda_{\star}\equiv\frac{1-s_{\star}}{s_{\star}}\ll 1$ is small but still $\Theta(1)$, which is far before the vanishing $\tilde{\lambda}$ in~\cite{doi:10.1073/pnas.1002116107-LocMakeAdiaFail}, the system already fall into the region where any solutions below $E_B-2\lambda_\star n$ are frozen inside their clusters and cannot directly tunnel between each other (see Fig.~\ref{fig:OGP}). In this case, not just the exact optimum in a worst-case $H_{\mathcal{L}}$, even approximate optimal solutions in a typical-case $H_{\mathcal{L}}$ become hard to find.

Thus, if the state before $s_\star$ does not contain enough weights in near-optimal clusters, they cannot be improved by the evolution from $s_\star$ to 1 within polynomial time. On the other hand, if we do find some near-optimal solutions before $s_\star$, dynamical localization can protect good solutions from escaping. This means that even if we rapidly drive $s(t)$ from $s_\star$ to 1, those good solutions still remain close to their local minima and can be efficiently recovered by classical local search methods.

We leave open the possibilities whether those solutions in near-optimal (deep) clusters can have significant weights in the evolving states before $s_{\star}$, and add some intuitive remarks here. Firstly, due to the locality of the coupling $H_M$, the evolving state cannot fully explore the low-energy solution subspace of $H_{\mathcal{L}}$ without reaching high-energy solutions. Within those high-energy solutions slightly above $E_B$, typically there are exponentially more of them connected to shallow clusters than those connected to deep clusters~\cite{gamarnik2021overlap}. Thus, we expect that it is generally hard for both classical and quantum algorithms to significantly explore inside the deep clusters to find good solutions. On the other hand, given the flexibility in choosing $H_M$ and evolution paths, it is also hard to rigorously rule out all possible polynomial-time adiabatic algorithms to find good solutions. Even if we consider a weaker task to rule out the possibilities to find the exact optimum, this remains extremely hard. If one can prove no polynomial-time adiabatic algorithms are capable to find the exact optimum, given their equivalence to circuit-based quantum algorithms~\cite{aharonov2008adiabatic,yu2018exact}, it is equivalent~\footnote{Although finding the exact ground state energy (up to certain precision) is not directly in QMA by definition, it can be polynomially reduced to the decision version of the local Hamiltonian problems, which is a canonical QMA problem.} to show that a problem can be in QMA~\cite{kempe2006complexity} but not in BQP~\cite{bernstein1993quantum}, which would be a groundbreaking result in quantum complexity theory.

Nevertheless, regardless of its performance before $s_{\star}$, if our goal is to find classically inaccessible solutions much below $E_B$, the late-time evolution of $s(t)$ from $s_{\star}$ to 1 is ineffective, in the sense that it will neither improve nor worsen the output states in polynomial time. For such a goal, unless we are given an exponentially long runtime, we should accelerate the late-time evolution regardless of the possible exponentially small gap within this time interval, or we can simply stop the evolution at $s_{\star}$, which will not affect much the quality of the final output states. 

\section{Experimental verification}

In this section, we briefly discuss how to experimentally detect the energy-space localization and its consequences. Directly probing the exponentially small leakage outside the energy window $\mathcal{E}_0(d)$ is usually infeasible in practice. Thus, we need the unusual stability inferred from energy-space localization in systems with energy barriers to detect and confirm it. One approach is to test the stability of Hamiltonian built from LDPC codes or hard optimization problems under time-dependent perturbations.

For cLDPC codes specifically, the experimental procedure is simpler. We can prepare the initial state to be any codeword for simplicity, which is in computational basis, then add a generic time-dependent perturbation with bounded total variation. After a long time, we can measure in the $Z$ basis and check whether we obtain: 1. The number of violated checks is bounded; 2. The output bit string is close to the initial state in terms of average Hamming distance. Similar experiments can also be performed on models built from hard optimization problems, where one can prepare the initial state to be an approximate classical solution close to optimum.

For qLDPC codes, detecting the localization is more complicated, as the state distance here is hard to directly measure. We can initialize the system at a logical $Z$ state for simplicity, and evolve the perturbed system for a long time. We repeat this with identical perturbations for multiple runs and obtain many copies of the erroneous states after long-time perturbation. Then, (1) we  measure the stabilizers to see the syndromes, (2) we apply the decoder to correct them, and  (3) we measure the logical $Z$ operators to see if it remains the original codeword. Our energy-space localization predicts that the number of syndromes will saturate to $\lesssim\lambda n$ even after a very long time, and the corrected logical $Z$ values remain the same as initial states with high probability, which one can implement and verify with sufficient statistics.

To distinguish this new stability from other stable phases like MBL and prethermalization, the generic time-dependency of the perturbations is essential. To fully rule out the effects from prethermalization or other localization mechanisms, one can implement many time-dependent perturbations and evolve the system for at least $T\gg O(n)$ time and average over the results. We expect that the effects from energy-localization can still survive while other mechanisms cannot.

We remark that existing constructions of good LDPC codes and hard optimization problems all require long-range connectivity of qubits, which is not easy to realize via native local couplings of qubits. However, nowadays such long-range couplings can be experimentally implemented in quantum platform of trapped ions~\cite{cirac1995quantum,monroe1995demonstration,bruzewicz2019trapped} and neutral atom systems~\cite{bluvstein2024logical,bluvstein2026fault}. For models built from hard optimization problems, another way to circumvent long-range interactions is to encode them in equivalent 2D problems~\cite{nguyen2023quantum}. Therefore, the dynamics of the models discussed can in principle be simulated on these experimental platform, but the intrinsic noise inside current quantum computers limit the evolution time of those simulations.

\section{Conclusion and Discussion}
\label{sec:discussion}

In this work, we have introduced and proved the energy-space localization as a universal property of time-dependent $q$-local Hamiltonians (Theorem~\ref{thm:Energy-local}). The static energy-space localization stated in Theorem~\ref{thm:static-Energy-loc} can be regarded as a special limit of the dynamical case in Theorem~\ref{thm:Energy-local}. These theorems, serving as both fundamental theoretical results and versatile mathematical tools, have enabled us to demonstrate wide applications in proving the stability of systems with extensive energy barriers, such as in certain LDPC codes and hard optimization problems.

The energy-space localization also allows us to study generic quantum evolution in the lens of energy changes. In the thermodynamic limit, the exponentially decaying contributions from eigenstates with $D>\Lambda$ naturally reduce to the classical constraint due to energy conservation, demonstrating a provable quantum-classical correspondence. The property also holds ubiquitously in qubit and fermionic systems, as long as their local norms are bounded, except for bosonic systems, where the local norms are unbounded.

Using energy-space localization, we have shown that systems with extensive energy barriers exhibit extremely stable phases, even under time-dependent perturbations. Among the various stabilities demonstrated for LDPC codes (as summarized in Table~\ref{tab:LDPC-results}), the exponentially long dynamical localization against generic time-dependent perturbations stands out for its theoretical novelty and practical implications. This result not only establishes proof of stability in a previously considered intractable time-dependent setting but also provides a theoretical foundation on how error-correcting codes protect information against realistic time-varying noise. We have also extended other previously known stabilities of LDPC codes to quasi-$q$-local perturbations. On the other hand, for hard optimization problems, we proved that the evolving state under an algorithmic driving Hamiltonian becomes trapped inside local minima for an exponentially long time if the total variation of the driving is not sufficiently large. In summary, our theoretical framework rigorously connects the intrinsic property of quantum evolution with the stability of error-correcting codes and the algorithmic hardness of optimization problems from the perspective of the energy space.

We have also discussed how to experimentally probe the unusual localization. Key difficulties for these experiments are the requirements for long-range qubit connectivity and long-time unitary evolution. Although long-range connectivity can indeed be realized in many quantum platforms, the long-time evolution is hard to simulate due to the noise. 

Our results also suggest several promising future research directions. Firstly, since our leakage bounds apply to generic systems, they can be further improved in specific systems, such as mean-field models (e.g., SYK models~\cite{PhysRevLett.70.3339}) and in specific circumstances, such as short-time evolution and the low-energy subspace. Secondly, the bounds are closely related to the instantaneous higher moments of a Hamiltonian and its expectations of the nested commutators, which also appear in the study of Lieb-Robinson bounds~\cite{PhysRevX.10.031009}, operator growth~\cite{PRR-OpGrowDisorder,PhysRevX-UnivOpGrow}, and Krylov subspace methods~\cite{nandy2025quantum}, where many well-established tools in these areas might be combined with our energy-space localization to obtain better bounds. Our theoretical framework may also help establish deeper connections among them, e.g. in quench dynamics.

\medskip \noindent {\bf Acknowledgments}.
The authors thank David Gamarnik for insightful discussions on the relation between OGP and the clustering property. This work was partly supported by the U.S. Department of Energy,
Office of Science, National Quantum Information Science
Research Centers, Co-design Center for Quantum Advantage
(C2QA) under Contract No. DE-SC0012704, in particular, on
the consequences of localization on the LDPC codes and quantum algorithms, and by the National Science Foundation under Grant No. PHY
2310614, in particular, on the part of the energy-space localization.

\bibliography{tMBL}

\newpage

\section{Methods}

In this section, we provide some brief but intuitive descriptions of the techniques used in this work and refer the readers to the Appendices for detailed and rigorous proofs. In Sec.~\ref{methods:extensionForESL}, we discuss some simple extensions and corollaries from energy-space localization in more generalized scenarios. In Sec.~\ref{methods:staticReduce}, we briefly describe how to reduce the time-dependent energy-space localization (Theorem~\ref{thm:Energy-local}) to static cases (Theorem~\ref{thm:static-Energy-loc}). In Sec.~\ref{methods:dyLoc}, we give an intuitive descriptions of the techniques used in proving dynamical localization of LDPC codes. In Sec.~\ref{methods:eigenLoc}, we briefly show that if the perturbed eigenstates remain localized under static perturbations, which only appears in cLDPC codes (and classical hard optimization problems), the dynamical localization time can be infinitely long. In Sec.~\ref{methods:sym}, we illustrate with a simple non-interacting model that can exhibit all the stabilities discussed in this work to show that they arise purely due to the energy barriers. This is made possible by restricting the dynamics in a clustered subspace via symmetry induced from a hard optimization problem.

\subsection{Extensions for energy-space localization}
\label{methods:extensionForESL}

The key step (see Appendix~\ref{app:pEnergyLoc}) to prove the energy-space localization (Theorem~\ref{thm:Energy-local}) is to bound the growth of arbitrary $2k$-th moments $\braket{\psi(t)|(H(t)-E_0)^{2k}|\psi(t)}$. The bounds are closely related to the growth of the nested commutators $\mathrm{ad}^m_{H(t)}(H'(t))$, where we denote $\mathrm{ad}_H (V)\equiv[H,V]$ and $\mathrm{ad}^{m}_H (V)\equiv[H,\mathrm{ad}^{m-1}_H(V)]$ for brevity. Thus, our results apply to any $H(t)$ whose $\mathrm{ad}^m_{H(t)}(H'(t))$ follows similar scalings, which can be summarized in the following corollary: 

\begin{Cor}[Bounds for energy-space localization from nested commutators, informal]
For any $H(t)$ with $\int\|\mathrm{d}H\|_X \leq \lambda n$, if the norm of the nested commutators $\left\|\mathrm{ad}_{H(t)}^m\left(\frac{H'(t)}{\|H'(t)\|_{\textnormal{X}}}\right)\right\|$ can be uniformly bounded by a series $F_m\sim O(m!\Delta^m)$ that grows at most factorially with $m$, the leakage outside the energy window $\mathcal{E}_0(d)$ can be bounded by $\epsilon^{(1)}_{\lambda,\Delta}(d)$. If $F_m\sim O(\Delta^m)$ grows at most exponentially with $m$, the bound can be improved to $\epsilon^{(2)}_{\lambda,\Delta}(d)$.
\end{Cor}
The Corollary can be inferred from the proofs in Appendix~\ref{app:pEnergyLoc}. Intuitively, since commuting modifications of $H(t)$ will not mix different eigenstates in the spectrum of the unmodified Hamiltonian, the contribution of the weights outside the energy window can only come from the non-commuting terms, which can be quantified by the nested commutators $\mathrm{ad}^m_{H(t)}(H'(t))$. As the nested commutators grow at most factorially~\cite{PhysRevX-UnivOpGrow,PRR-EulideanOpGrow,PRR-OpGrowDisorder} with $m$ for $q$-local Hamiltonians, our results can be applied to any $q$-local $H(t)$ with a bounded local norm.

We can then extend the energy-space localization to further cases (i.e., Case~\ref{case:2}-\ref{case:4} in the main text), where the scaling of the nested commutators remains bounded above by a factorial growth (see Appendix~\ref{app:proofCommute}).

Although our result requires the initial state to be an eigenstate of $H(0)$, it can be easily extended to general initial states $\ket{\Psi(0)}$, which can always be expanded $\ket{\Psi(0)}=\sum_j c_j \ket{\psi_j(0)}$ in the eigenbasis of $H(0)$. Taking Case~\ref{case:1} as an example, the final state is simply $\ket{\Psi(T)}=\sum_j c_j \ket{\psi_j(T)}$, where each $\ket{\psi_j(T)}$ is the final state evolving from $\ket{\psi_j(0)}$ via $H(t)$. Suppose that the energy window defined for each $\ket{\psi_j(0)}$ is $\mathcal{E}_j(d)\equiv[E_j-dn,E_j+dn]$, then we can bound the final weight $\epsilon_E$ for any energy-$E$ level by
\begin{equation}
    \epsilon_E\leq \inf_{d>\lambda} \left(\sum_{E \in \mathcal{E}_j(d)} |c_j| +\sum_{E \notin \mathcal{E}_j(d)}| c_j| \epsilon^{(1)}_{\lambda,\Delta_q}(d)\right).
\end{equation}

In addition, our method does not rely on specific scalings of the total variation or local norm, the results can easily be extended to general scalings.

\begin{Cor}[Bounds for general scalings, informal]
If $\Lambda$ and $M$ of $H(t)$ from $\int\|\mathrm{d}H\|_X \leq \Lambda$ and $\|H(t)\|_{\textnormal{loc}}\leq M$ both follow general scalings with $n$, then one can choose an energy window $[E_0-D,E_0+D]$ with $D\gtrsim \Lambda$, such that the leakage outside the window can be bounded at the order of $e^{- \Omega(\Lambda/M)}$.
\end{Cor}
\noindent We present a brief proof of this Corollary in Appendix~\ref{app:generalScalingP}. We remark that Theorem~\ref{thm:Energy-local} is its special case where $\Lambda=\lambda n$, $D=dn$ and $M$ is $O(1)$.

\subsection{Static reduction for energy-space localization}
\label{methods:staticReduce}

The energy-space localization from Theorem~\ref{thm:Energy-local} can be easily reduced to static perturbations. To do the reduction, we can consider $|\psi(0)\rangle$ being an eigenstate of $H=H_0+V_0$ that follows the evolution from $t=0$ to $T$ under the time-dependent Hamiltonian
\begin{equation}
\label{eq:tTostatic}
    H(t)=H_0+(1-\frac{t}{T})V_0.
\end{equation}
Obviously, $H(T)=H_0$. We then take the $T\rightarrow 0^{+}$ limit. From the Schr\"odinger equation, we know that a finite quench of the Hamiltonian will not instantly change the wavefunction, so the state $\ket{\psi(0^+)}=\ket{\psi(0)}$ remains the eigenstate of $H_0+V_0$, but is now exponentially localized in the energy space of $H_0$ by Theorem~\ref{thm:Energy-local}. Then $\epsilon^{(1)}$ and $\epsilon^{(2)}$  in Theorem~\ref{thm:Energy-local} explicitly bound the leakage outside the energy window in the eigenbasis of $H_0$. Since here $[V(t),H'(t)]=0$ is always satisfied by construction (as $V(t)=(1-t/T)V_0$), we no longer need the general Case~\ref{case:1} but only need Cases~\ref{case:2},~\ref{case:3}, and~\ref{case:4}. In this way, the Case~\ref{case:2}-\ref{case:4} can be translated into Case~\ref{case:2'}-\ref{case:4'}, respectively.

\subsection{Sketch of the proof of dynamical localization of LDPC codes}
\label{methods:dyLoc}

Here we give an intuitive proof sketch for dynamical localization (Proposition~\ref{pro:DyLocalt}) of LDPC codes and leave the full rigorous one to Appendix~\ref{app:pDynamicLoc}. For strictly $q$-local $V(t)$, since its matrix elements for eigenstates of $H_C$ (the Hamiltonian of an LDPC code) from different clusters are exactly zero, the evolving state cannot tunnel to other clusters without reaching a higher-energy eigenstate of $H_C$. Thus, if we can show that the high-energy weights $\epsilon_{>}$ in the evolving state $\ket{\psi(t)}$ in the eigenbasis of $H_C$ are exponentially small at any time within $[0,T]$, then the whole leakage can be upper-bounded roughly by $\lambda n T \epsilon_{>}$. 

To bound the high-energy weights in the eigenbasis of $H_C$ at time $t_1$, we can invoke Theorem~\ref{thm:Energy-local} for an extended evolution for $\ket{\tilde{\psi}(t)}$ under
\begin{equation}
    H_{\textnormal{ext}}(t)=\begin{cases}
        H_C &~\textnormal{for} ~t=0^{-},\\
        H_C+V(t) &~\textnormal{for} ~0\leq t\leq t_1,\\
        H_C &~\textnormal{for} ~t=t_1^{+},
    \end{cases}
\end{equation}
where we set $\ket{\tilde{\psi}(0^-)}=\ket{\psi(0)}$. We also define that the $H_{\textnormal{ext}}(t)$ at $t=0^-,0$ and $t_1,t_1^+$ are linearly connected (similar to Eq.~\eqref{eq:tTostatic}) and we let the $\ket{\tilde{\psi}(t)}$ evolve from $0^{-}$ to $t_1^{+}$. Here, we should use the total variation of $H_{\textnormal{ext}}(t)$ instead of $H(t)$ and define $\lambda$ for $q$-local perturbation in Cases.~\ref{case:1}-\ref{case:3} as 
\begin{equation}
\label{eq:dyLambda-q}
    \|V(0)\|_{\textnormal{X}}+\int_0^T  \|V'(t)\|_{\textnormal{X}}\leq \lambda n\equiv\Lambda.
\end{equation}
For quasi-$q_{\star}$-local $V(t)$ in Case~\ref{case:4}, we define $\lambda$ as
\begin{equation}
    \|V(0)\|_{q\textnormal{-X}}+\int_0^T  \|V'(t)\|_{q\textnormal{-X}}\leq \lambda n\equiv \Lambda.
\end{equation}

From Schr\"odinger's equation, we know that $\ket{\tilde{\psi}(t_1^+)}$ from $H_{\textnormal{ext}}(t)$ is the same as $\ket{\psi(t_1)}$ from $H(t)$. Thus, we can bound the high-energy leakage of $\ket{\psi(t_1)}$ in the eigenbasis of $H_C$ at any time $t_1$, at the cost of adding a factor 2 to the perturbation strength.

Therefore, as long as $\lambda<({b-\varepsilon_0})/{2}$, the leakage to high-energy states can be bounded by $\epsilon^{(1)}_{2\lambda,\Delta_q}(b-\varepsilon_0)$ (and can be improved to $\epsilon^{(2)}_{2\lambda,\Delta_q}(b-\varepsilon_0)$ if $V(t)$ satisfies the condition in Case~\ref{case:3}), which is exponentially small and is denoted by $ e^{-\xi n}$. Thus, we can choose $T\lesssim \frac{1}{\lambda n} e^{\xi_1 n}$ with any $\xi_1<\xi$, and bound the whole leakage by  $e^{-(\xi-\xi_1)n}$.

\subsection{Eigenstate localization and infinite-time dynamical localization}
\label{methods:eigenLoc}

In the main text, we show that classical and quantum LDPC codes with linear energy barriers can remain localized for an exponentially long time against time-dependent perturbations, which can also be applied to static perturbations. Here, we show that the localization can be infinitely long if the static perturbation strength is weak enough, due to eigenstate localization~\cite{yin2024eigenstate}.

This property was recently discovered~\cite{yin2024eigenstate} in certain cLDPC codes. By invoking Theorem~\ref{thm:static-Energy-loc}, we reproduce the results and generalize them to including quasi-$q$-local perturbations. Specifically, we consider the following Hamiltonian
\begin{equation}
\label{eq:HdefEigenLoc}
    H=H_C+V_0+H_{\textnormal{d}},
\end{equation}
where $H_C$ is a cLDPC Hamiltonian with linear soundness defined above, $V_0$ is a (quasi-)$q$-local perturbation with $\|V\|_{(q\textnormal{-)X}}\leq \lambda n$, and $H_{\textnormal{d}}$ is a vanishingly small detuning term and we can choose $H_{\textnormal{d}}\equiv \frac{1}{n^2}\sum_{j}h_j^z Z_j$ with $O(1)$ randomly chosen $h_j^z$. We claim that

\begin{pro}[Eigenstate localization with (quasi-)$q$-local perturbations, informal]
\label{pro:Eigen-local}
For the $H$ defined above~(\ref{eq:HdefEigenLoc}), if $\lambda$ (also $q_\star$ for quasi-$q_{\star}$-local) is sufficiently small, each perturbed eigenstate $\ket{\psi}$ of $H$ with energy below $E_B$ is exponentially localized near one codeword $\ket{\mathrm{w}}$, with leakage roughly bounded by $\epsilon_{>}/\delta_W$ defined below.
\end{pro}

We show explicit bounds with the proof in Appendix~\ref{app:pEigenLoc}. Here $\epsilon_{>}$ is roughly the leakage to higher-energy ($>E_B$) eigenstates up to an energy scale, which can be exponentially bounded by Theorem~\ref{thm:static-Energy-loc}. The $\delta_W$ is roughly the minimum detuning from $H_{\textnormal{d}}$ between different clusters. Thus, the key step of the proof is to ensure that the detuning $\delta_W$ is much larger than the leakage $\epsilon_{>}$. It turns out that for cLDPC codes, $\delta_W$ can be lower bounded by $e^{-2n}$ with high probability, while for qLDPC codes, it is unlikely to lower bound $\delta_W$. This is because the ground state degeneracy of good qLDPC codes is stable~\cite{de2025low,yin2025low} against any local perturbations.

The eigenstate localization can also affect the dynamical localization time. Although the exponentially long localization time in Proposition~\ref{pro:DyLocalt} still works here, it can be improved to infinitely long for the static Hamiltonian~\eqref{eq:HdefEigenLoc} due to eigenstate localization.

\begin{pro}[Infinitely long dynamical localization with static (quasi-)$q$-local perturbation]
\label{pro:DyLocalInf}
Let $\ket{\psi(t)}$ be initialized in one cluster $w_0$ and evolve under the static Hamiltonian: $H=H_C+V_0+H_d$. If $\lambda$ (and $q_\star$ in quasi-$q_{\star}$-local perturbations) is sufficiently small, the state resides almost within $w_0$ to infinitely long time, with a leakage bounded by $e^{-\Omega(n)}$.
\end{pro}

We show the proof in Appendix~\ref{app:DyLocInfinite}. Intuitively, if the perturbation is sufficiently weak, we can prove that the majority of $\ket{\psi(0)}$ comes from the localized eigenstates inside $w_0$. Since weights on eigenstates are conserved at any time of the evolution, the leakage outside $w_0$ can be bounded for all time.

\subsection{Symmetry-protected localization from constraint optimization problems }
\label{methods:sym}

To demonstrate that all the stabilities discussed in Sec.~\ref {sec:LDPC} solely arise from the clustered energy landscape, we give an example that a trivial non-interacting Hamiltonian can have those stabilities if the perturbations respect some symmetries so that the evolution is restricted to some subspaces defined from hard constraint optimization problems.

Specifically, we can consider the maximum independent set (MIS) of a random $p$-regular graph $G$: finding the largest subset of vertices that does not contain any edges between them. The quantum Hamiltonian for the cost function $\mathcal{L}$ is a simple non-interacting Hamiltonian counting the vertex number
\begin{equation}
    H_{\textnormal{V}}\equiv\sum_j Z_j.
\end{equation}
We can also define a $H_{\textnormal{E}}(G)$ for the graph $G$ to count the edge number of a $Z$-basis states
\begin{equation}
    H_{\textnormal{E}}(G)\equiv\sum_{\braket{ij}\in G} (1-Z_i)(1-Z_j).
\end{equation}
Then the MIS corresponds to the ground state of $H_{\textnormal{V}}$ in the zero-subspace of $H_{\textnormal{E}}(G)$. The model also has the clustering property~\cite{coja2015independent,gamarnik2017limits} in a linear energy window $[E_g,E_g+B]$.

Therefore, the non-interacting $H_{\textnormal{V}}$ should enjoy all the stability properties discussed in Sec.~\ref{sec:LDPC}, if the perturbation $V$ leaves the ground state subspace of $H_{\textnormal{E}}$ invariant (this requirement is weaker than $[V,H_{\textnormal{E}}]=0$)~\footnote{One can also choose choose $H_{\mathcal{L}}=H_{\textnormal{E}}(G)$ and set $V$ to respect the symmetry that conserves the total vertex number (i.e., $[V,H_{\textnormal{V}}]=0$). Such a symmetry preserving particle number is more common, and one can verify that its energy landscape exhibits clustering property in a linear energy window if choosing a vertex number near MIS size, thus it is also stable against all $q$-local symmetry-respecting $V$.  }. One choice of such a perturbation $V$ can be
\begin{equation}
    V(t)=\lambda(t)\sum_{i}\prod_{j,\braket{ij}\in G} P_j^zX_i,
\end{equation}
where $P_j^z=\frac{1+Z_j}{2}$ is the projector to state $\ket{0}$ for qubit $j$. Similar interactions also appear in the PXP model~\cite{bernien2017probing}, which exhibits many-body quantum scars~\cite{bernien2017probing,serbyn2021quantum} when $G$ is some bipartite graph different from our cases. In such a type of $V$, if the initial state is in the ground subspace of $H_{\textnormal{E}}$ and has sufficiently low energy in $H_{\textnormal{V}}$, it is robust against all kinds of perturbations discussed in Sec.~\ref{sec:LDPC}.

\clearpage


\onecolumngrid

\newpage

\appendix

\begin{center}
    {\large \textbf{Supplementary Information}}
\end{center}

\makeatother
\appendixcontents

\renewcommand{\thefigure}{S\arabic{figure}}
\renewcommand{\thetable}{S\arabic{table}}
\setcounter{equation}{0}
\setcounter{figure}{0}
\setcounter{table}{0}

\setcounter{page}{1}
\renewcommand{\thepage}{S\arabic{page}}

\begin{center}
    {\textbf{Overview of the Supplementary Information}}
\end{center}

In this supplemental material, we will show formal statements of theorems, propositions, and notations mentioned or used in the main text, as well as their detailed proofs. Appendix~\ref{app:normsTVs} gives all the definitions and properties of the norms and total variations used in this work. 

In Appendix~\ref{app:pEnergyLoc}, we give a formal statement of energy-space localization (Theorem~\ref{thm:Energy-local}), and show the detailed proofs for all cases defined in the main text. The main technique is to construct a recurrence relation for the arbitrary $k$-th central moment of energy, so that their bounds can be obtained recursively. In bounding the growth of the $k$-th central moment, we relate it to the growth of nested commutators of a $q$-local Hamiltonian, whose bounds are obtained with more detail in Appendix~\ref{app:proofCommute}.

In Appendix~\ref{app:pDynamicLoc}, we prove that for systems with linear energy barriers, initial states inside one cluster are localized for an exponentially long time under generic time-dependent perturbations. Such an exponentially long-time dynamical localization only requires the linear energy barrier in the clustered solution space, regardless of whether the perturbed eigenstates are localized or not. If the perturbed eigenstates are indeed localized, we also discuss the possibility of infinite-time dynamical localization, which can only appear under static perturbations. This infinite-time dynamical localization only works when the system exhibits the eigenstate localization under static perturbations, which was originally introduced in~\cite{yin2024eigenstate} for static, strictly $q$-lcoal perturbations, and will be discussed and extended to quasi-$q$-local in Appendix~\ref{app:pEigenLoc}.

In Appendix~\ref{app:pEigenLoc}, we prove that for systems with linear energy barriers under static (quasi)-$q$-local perturbations, the perturbed eigenstates are localized inside the original clusters. Such eigenstate localization under static strictly $q$-local perturbation has already been obtained in~\cite{yin2024eigenstate}. A key technique in~\cite{yin2024eigenstate} is rewriting the perturbed Hamiltonian in a tri-diagonal form to bound the high-energy contribution in perturbed eigenstates. However, this method is hard to extend to quasi-$q$-local cases, as quasi-$q$-local terms cannot be written in tri-diagonal forms in the unperturbed basis. To resolve this, we replace the method with our Theorem~\ref{thm:static-Energy-loc} to bound the high-energy leakage, which can easily extend to quasi-$q$-local cases, and adapt the rest of procedures in~\cite{yin2024eigenstate} to generalize the eigenstate localization to quasi-$q$-local perturbations.

In Appendix~\ref{app:pGibbs}, we prove that for systems with linear energy barriers under static (quasi)-$q$-local perturbations, their corresponding Gibbs samplers require an exponentially long time to reach equilibrium. Such slow-mixing was also obtained in~\cite{rakovszky2024bottlenecks,gamarnik2024slow} when the perturbation is strictly $q$-local. The similar tri-diagonal technique to~\cite{yin2024eigenstate} is used in~\cite{rakovszky2024bottlenecks} to bound the high-energy and low-energy contributions in perturbed eigenstates. As previously mentioned, this method is hard to generalize to quasi-$q$-local cases. To resolve this, we replace their method with our Theorem~\ref{thm:static-Energy-loc} and adapt the quantum bottleneck theorem developed in~\cite{rakovszky2024bottlenecks} to generalize the results to quasi-$q$-local perturbations. 

In Appendix~\ref{app:proofCommute}, we provide some simple proofs to bound the growth of nested commutators in different cases, whose growth is closely related to the tail bound in energy-space localization proved in Appendix~\ref{app:pEnergyLoc}. The growth of nested commutators is well-known in studies on operator growth~\cite{PhysRevX-UnivOpGrow,PRR-OpGrowDisorder}, Floquet-Magnus theory~\cite{kuwahara2016floquet}, and ETH~\cite{gong2022bounds}. We briefly derive them in Appendix~\ref{app:proofCommute} for the sake of technical completeness and to maintain consistency with our definitions and notations.

\appsection{Definition of different norms and total variations}
\label{app:normsTVs}

In this work, we define various norms and total variations for mathematical convenience. Here, we summarize them with their definitions and properties.

We denote $\|H\|$ as the operator norm of a Hamiltonian $H$, i.e., 
\begin{equation}
    \|H\|\equiv\sup_{\ket{\psi}\neq 0} \frac{\|H \ket{\psi}\|}{\|\ket{\psi}\|}.
\end{equation}
Upon this, we define the termwise total norm $\|\cdot\|_{\text{X}}$ as
\begin{equation}
    \|H\|_{\textnormal{X}}\equiv\sum_{A}\|h_A\|, \quad \mbox{for} \, H=\sum_A h_A,
\end{equation} 
where we decompose $H$ into local terms and each $h_A$ is supported on a region $A$. Note that one may have different ways of decomposing the total $H$, e.g., coarse-graining, so the termwise norm is not generally unique. Any reasonable decomposition should suffice. Moreover, for qubit Hamiltonians, we can always choose decomposition into Pauli operators to ensure uniqueness. With the decomposition fixed, we can define the termwise total variation $\mathbf{V}_{\textnormal{X}}$ for a time-dependent $H(t)$ from $t=0$ to $T$ as
\begin{equation}
    \mathbf{V}_{\textnormal{X}}(H(t);[0,T])\equiv\int_{0}^T\|H'(t)\|_{\textnormal{X}}\mathrm{d}t\equiv \int_{0}^{T} \|\mathrm{d}H\|_{\textnormal{X}},
\end{equation}
where $T>0$. We can rewrite $H(t)$ as $H(t)=H(0)+V(t)$. An inequality we will frequently use in this work is
\begin{equation}
    \|V(t)\|\leq \|V(t)\|_{\textnormal{X}}=\left\|\int_{0}^t\mathrm{d}H\right\|_{\textnormal{X}}\leq \int_{0}^{t} \|\mathrm{d}H\|_{\textnormal{X}}.
\end{equation}
These are norms for global properties of a Hamiltonian. We also need a local norm to depict the local strength of a Hamiltonian. For a Hamiltonian $H$ with decomposition $H=\sum_{A}h_A$ on sites $i$'s, we define the local norm $\|\cdot\|_{\text{loc}}$ as
\begin{equation}
    \|H\|_{\text{loc}}\equiv\sup_{i}\sum_{A,A\ni i}\|h_A\|,
\end{equation} 
where the supremum is over all sites $i$. It is obvious that for a fixed decomposition $H=\sum_{A}h_A$, we have
\begin{equation}
    \|H\|\leq\|H\|_{\textnormal{X}}\leq n \|H\|_{\text{loc}},
\end{equation}
where $n$ is the total number of sites.

The above definition is sufficient for strictly $q$-local Hamiltonians. However, for quasi-$q$-local Hamiltonians, we need an additional parameter $q$ to describe its decaying speed with the number of sites involved in the interactions. We define a Hamiltonian $H$ to be quasi-$q$-local if it has a decomposition $H=\sum_Ah_A$, where the norm of each $k$-local term roughly decays exponentially $e^{-k/q}$ with the locality. To give a precise definition, we define the termwise $q$-norm as
\begin{equation}
    {\left\|H\right\|}_{q\textnormal{-X}}\equiv \sum_{A}(q+1)e^{|A|/q}\|h_A\|,
\end{equation}
where $q>0$ and $|A|$ counts the number of sites of the support $A$. Similar definitions also appear in the study on prethermalization  \cite{abanin2017effective,abanin2017rigorous,else2017prethermal}. We add an additional $q+1$ factor for mathematical convenience. We define $H$ as quasi-$q_{\star}$-local if one can find a $q_{\star}$ (usually the minimum of all possible $q$) such that $\frac{1}{n}\|H\|_{q_{\star}\textnormal{-X}}$ can be bounded by an $O(1)$ constant. It is easy to verify that given a fixed decomposition, we have
\begin{equation}
    \|H\|\leq\|H\|_{\textnormal{X}}\leq \|H\|_{q\textnormal{-X}}.
\end{equation}
In addition, given the decomposition $H=\sum_Ah_A$, if $\|H\|_{q\textnormal{-X}}\leq \lambda n$, where $\lambda$ is $O(1)$, we have
\begin{equation}
    \sum_{A,|A|=k}\|h_A\|\leq \frac{\lambda n}{q+1} e^{-k/q}.
\end{equation}
This indicates that the termwise norm of $k$-local components of a quasi-$q$-local $H$ decays exponentially with $k$. Similarly, we define the total variation $\mathbf{V}_{q\textnormal{-X}}$ of quasi-$q$-local $H(t)$ as
\begin{equation}
    \mathbf{V}_{q\textnormal{-X}}\left(H(t);[0,T]\right)\equiv \int_{0}^T \left\|\frac{\partial H(t)}{\partial t}\right\|_{q\textnormal{-X}}\mathrm{d}t\equiv \int_0^T \|\mathrm{d}H\|_{q\textnormal{-X}}
\end{equation}
Similarly, if we define $V(t)=H(t)-H(0)$, we have the following inequality
\begin{equation}
    \|V(t)\|\leq \|V(t)\|_{\textnormal{X}} \leq\| V(t)\|_{q\textnormal{-X}} \leq \int_{0}^{t} \|\mathrm{d}H\|_{q\textnormal{-X}}.
\end{equation}

\appsection{Proof of the energy-space localization}
\label{app:pEnergyLoc}

In this section, we are going to rigorously prove the Theorem~\ref{thm:Energy-local} for Cases~\ref{case:1}-\ref{case:4}. The procedures for all four cases are similar, and a key step is to bound the nested commutator $\mathrm{ad}^m_{H(t)}(H'(t))$. We use a factorially growing series to bound the commutators For Cases~\ref{case:1}, \ref{case:2} and \ref{case:4}, and an exponential bound for Cases~\ref{case:3}. The proofs for same-scaling commutators are similar, so we put the proofs for the two different scalings in the Appendix.~\ref{app:EnergyLocP} and \ref{app:EnergyLocPcom}, respectively.

For clarification, we rewrite Theorem~\ref{thm:Energy-local} (for Case~\ref{case:1} as an example) with more details. Other cases follow a similar formulation.

\begin{thmapp}[Energy-space localization, for Case~\ref{case:1} as an example]
For a time-dependent $q$-local (no geometrical constraints here) system $H(t)$ defined on a $n$-site lattice, we assume that its total variation of termwise norm in $t\ \in[0,T]$ is bounded by 
\begin{equation}
    \int_0^T\left\|\frac{\partial H(t)}{\partial t}\right\|_{\textnormal{X}}\mathrm{d}t \leq \lambda n.
\end{equation} 
We assume that $H(t)$ has a finite {\it local} norm $\|H(t)\|_{\textnormal{loc}}\leq M$ for any $t$, i.e., for a partition $H(t)=\sum_A h_A(t)$, and for any $t$ and site $i$
\begin{equation}
    \sum_{A,A\ni i}\|h_A(t)\|\leq M.
\end{equation} 
This means that all terms that have (partial) support at site $i$ have a finite combined norm.
Let a state $|\psi(t)\rangle$, initialized in any of the eigenstates of $H(0)$ with energy $E_0$
\begin{equation}
    H(0)\ket{\psi(0)}=E_0\ket{\psi(0)},
\end{equation}
evolve according to $H(t)$, i.e.,
\begin{equation}
    i \frac{\partial}{\partial t} \ket{\psi(t)}=H(t)\ket{\psi(t)}.
\end{equation}
Then the evolving state $\ket{\psi(t)}$ at any $t$ can be expanded in the eigenbasis of $H(t)$:
\begin{equation}
    \ket{\psi(t)}=\sum_j a_j (t) \ket{\phi_j(t)},
\end{equation}
where $\phi_j(t)$'s are the instantaneous eigenstates of $H(t)$ with energy $E_j(t)$. 

We claim that $\ket{\psi(T)}$ is exponentially localized in the energy window $\mathcal{E}_0(d)\equiv[E_0-dn,E_0+dn]$ in the eigenbasis of $H(T)$, with leakage $\epsilon_T(d)$ bounded by
\begin{equation}
    \epsilon_T(d)\equiv\sqrt{\sum_{j,E_j(T)\notin \mathcal{E}_0(d)} |a_j(T)|^2} \leq e^{-Cn},
\end{equation}
where $d>\lambda$ is a parameter that sets the leakage width, and $C$ is defined by
\begin{equation}
    C\equiv\frac{\lambda}{\Delta_{q}}\left(\frac{d}{\lambda}-1-\ln\frac{d}{\lambda}\right)+o(1)
\end{equation} 
with $d/\lambda>1$, $\Delta_q\equiv2qM$ and $o(1)$ being a small positive number that will vanish if $n\rightarrow\infty$. Using the definition in Eq.~\eqref{eq:epsilond1f}, the bound for $\epsilon_T(d)$ can be written as
\begin{equation}
    \epsilon_T(d)\leq \epsilon^{(1)}_{\lambda,\Delta_q}(d).
\end{equation}
\end{thmapp}

The formulations for other cases can easily be obtained by replacing the constraints on $H(t)$ accordingly, where the leakage can be bounded by $\epsilon_T(d)\leq \epsilon^{(1)}_{\lambda,\Delta_q}(d)$ for Cases \ref{case:1}, \ref{case:2} and \ref{case:4}, and $\epsilon_T(d)\leq \epsilon^{(2)}_{\lambda,\Delta_q}(d)$ for Case~\ref{case:3}. The functions $\epsilon^{(1)}_{\lambda,\Delta_q}(d)$ and $\epsilon^{(1)}_{\lambda,\Delta_q}(d)$ are explicitly stated in the main text (in Theorem~\ref{thm:Energy-local}). We show the proofs and derive these two functions below.

\appsubsection{Main proof for Cases \ref{case:1}, \ref{case:2}, \ref{case:4}} 
\label{app:EnergyLocP}

For Cases.~\ref{case:1} and \ref{case:2}, we define the  total variation {\it per site} at time $t$ with $0\leq t\leq T$ as 
\begin{equation}
\label{eq:lambda_t-app}
    \lambda(t)\equiv\frac{1}{n}\int_{0}^{t}\left\|\frac{\partial H(\tau)}{\partial \tau}\right\|_{\textnormal{X}} \mathrm{d}\tau.
\end{equation}
Then it is obvious that $\lambda(t)\geq 0$ and $\lambda'(t)\geq 0$. For brevity, we define $\lambda_t\equiv\lambda(t)\leq\lambda(T)\equiv\lambda_T=\lambda$. Additionally, according to the definition, we have the following.
\begin{equation}
\label{eq:lambda'_t-app}
    \lambda'(t)n=\left\|\frac{\partial H(t)}{\partial t}\right\|_{\textnormal{X}}.
\end{equation}
Similarly, for Case~\ref{case:4}, we can define 
\begin{equation}
\label{eq:lambda_t-qx-app}
    \lambda(t)\equiv\frac{1}{n}\int_{0}^{t}\left\|\frac{\partial H(\tau)}{\partial \tau}\right\|_{q\textnormal{-X}} \mathrm{d}\tau
\end{equation} 
and obtain 
\begin{equation}
\label{eq:lambda'_t-qx-app}
    \lambda'(t)n=\left\|\frac{\partial H(t)}{\partial t}\right\|_{q\textnormal{-X}}.
\end{equation}

Starting with $H(0)\ket{\psi(0)}=E_0\ket{\psi(0)}$, we consider the growth of the expectation of $\big(H(t)-E_0\big)^{2k}$ on $\ket{\psi(t)}$. We define the shifted Hamiltonian $\tilde{H}(t)\equiv H(t)-E_0$ and the shifted eigen-energies $\tilde{E}_j(t)\equiv E_{j}(t)-E_0$ for convenience. Denote the following
\begin{equation}
    g_k(t)\equiv\left\|\tilde{H}(t)^k\ket{\psi(t)}\right\|.
\end{equation}
It is obvious that $g_k(0)=0$ and $g_k(t)^2$ gives exactly the expectation $\braket{\tilde{H}(t)^{2k}}$, i.e.,
\begin{equation}
   G_{2k}(t)\equiv \braket{\psi(t)|\tilde{H}(t)^{2k}|\psi(t)}=g_k(t)^2.
\end{equation}
With instantaneous eigenstates $\ket{\phi_j(t)}$ of $\tilde{H}(t)$ with energy $\tilde{E}_j(t)$, we can expand $\ket{\psi(t)}$ as follows:
\begin{equation}
    \ket{\psi(t)}=\sum_j a_j(t)\ket{\phi_j(t)}.
\end{equation}
We can also interpret $\ket{\psi(t)}$ as giving rise to a classical probability distribution of energies
\begin{equation}
    \mathbf{P}(\hat{E}_t=E)=\sum_{j,\tilde{E}_j(t)=E}|a_j(t)|^2,
\end{equation}
where we use the $\hat{}$ to emphasize that $\hat{E}_t$ is a classical random variable. By definition, the expectation $G_{2k} (t)$ in the instantaneous eigenbasis is the same as the $2k$-moment of $\hat{E}_t$,
\begin{equation}
    G_{2k}(t)=\mathbf{E}(\hat{E}_t^{2k})=\sum_{j}|a_j(t)|^2 \tilde{E}_j(t)^{2k},
\end{equation}
Thus, the far-energy contribution from eigenstates with $|\tilde{E}_j(t)|>dn$ can be upper bounded by Markov's inequality:
\begin{equation}
 \label{eqn:Pbound}
    \sum_{j,|\tilde{E}_j(t)|\geq dn} |a_j(t)|^2 =\mathbf{P}(|\hat{E}_t|\geq dn) \leq \min_k\frac{G_{2k}(t)}{(dn)^{2k}}.
\end{equation}
Thus, by definitions we have
\begin{equation}
\label{eq:epsilon<gk-appA1}
    \epsilon_t(d)\leq \min_{k} \frac{g_k(t)}{(dn)^k}.
\end{equation}

\smallskip
Next, we are going to bound $g_k(t)$ for arbitrary $k$. Denoting $\braket{\cdot}_t\equiv\braket{\psi(t)|\cdot|\psi(t)}$, using the obvious fact that $\tilde{H}(t)^k$ commutes with $H(t)$, we have 
\begin{equation}
    \frac{d}{dt}\Braket{\tilde{H}(t)^{2k}}_t=\Braket{\frac{\partial}{\partial t}\left( \tilde{H}(t)^{2k}\right)}_t.
\end{equation}
For $k=1$, we have
\begin{equation}
\begin{aligned}
    \frac{\partial}{\partial t}G_2(t)&= \braket{\psi(t)|\left(\tilde{H}(t)\frac{\partial H(t)}{\partial t}+\frac{\partial H(t)}{\partial t}\tilde{H}(t)\right)|\psi(t)}\\
    &\leq 2 g_1(t)\left\|\frac{\partial H(t)}{\partial t}\right\|\leq 2 g_1(t)\lambda'(t)n.
\end{aligned}
\end{equation}
Using $\frac{\partial}{\partial t}G_2(t)=2g_1(t)g_1'(t)$, we have
\begin{equation}
    g_1'(t)\leq \lambda'(t)n.
\end{equation}
Since $g_1(0)=0$ and $\lambda(0)=0$, we arrive at 
\begin{equation}
    0\leq g_1(t)\leq\lambda(t)n.
\end{equation}
In the above analysis, we implicitly assume $g_1(t)>0$ for $t>0$ so that we can safely eliminate one $g_1(t)$ on both sides. This can be justified as follows: By definition we have $g_k(t)\geq 0$ and $g_k(0)=0$. If there is one additional zero point $g_1(t_1)=0$ at some $t_1$, it indicates that the evolving state $\ket{\psi(t_1)}$ is exactly the instantaneous eigenstate of $H(t_1)$ with energy $E_0$. Thus, we can bound $g_1(t)$ by integrating $\lambda'(t)n$ from $0$ to $t$ for $0<t<t_1$ and from $t_1$ to $t$ for $t>t_1$, and both of them are bounded by the integral from $0$ to $t$ (i.e., $\lambda(t)n$). Thus, the above bound holds regardless of any additional zero points of $g_1(t)$ in $[0,T]$. Therefore, we can neglect the potential zero points and eliminate the $g_1(t)$ on both sides without affecting the final bounds. This argument holds for any $g_k(t)$: As $g_k(t_1)=0$ indicates that $\ket{\psi(t_1)}$ is an instantaneous eigenstate of $H(t_1)$ for any $k$, we can set $t_1$ as our new initial point and the bound obtained can only be smaller than that from $0$ to $T$.

Next, we consider $k=2$. Similarly, we have
\begin{equation}
\begin{aligned}
    2g_2(t)g_2'(t)&=\Braket{\tilde{H}(t)^3\frac{\partial H(t)}{\partial t}+\tilde{H}(t)^2\frac{\partial H(t)}{\partial t}\tilde{H}(t)+\tilde{H}(t)\frac{\partial H(t)}{\partial t}\tilde{H}(t)^2+\frac{\partial H(t)}{\partial t}\tilde{H}(t)^3}_t\\
    &=\Braket{2\tilde{H}(t)^2\frac{\partial H(t)}{\partial t}\tilde{H}(t)+2\tilde{H}(t)\frac{\partial H(t)}{\partial t}\tilde{H}(t)^2+ \tilde{H}(t)^2\left[\tilde{H}(t),\frac{\partial H(t)}{\partial t}\right]+\left[\frac{\partial H(t)}{\partial t},\tilde{H}(t)\right]\tilde{H}(t)^2}_t\\
    &\leq 4 \left\|\frac{\partial H(t)}{\partial t}\right\|g_1(t)g_2(t)+2g_2(t) \left\|\left[\tilde{H}(t),\frac{\partial H(t)}{\partial t}\right]\right\|.
\end{aligned}
\label{eq:g2}
\end{equation}
Here, we need to bound the commutator $\mathrm{ad}_{H(t)}(H'(t))$. We recall ${\rm ad}_A(B)\equiv [A,B]$. In fact, we can show  that (see Appendix~\ref{app:proofCommute}) for Cases~\ref{case:1} and~\ref{case:2}:
\begin{equation}
    \left\|\mathrm{ad}_{\tilde{H}(t)}^m \left(\frac{\partial H(t)}{\partial t}\right)\right\|\leq m!\Delta_q^m \left\|\frac{\partial H(t)}{\partial t}\right\|_{\textnormal{X}},
\end{equation}
and for Case~\ref{case:4}:
\begin{equation}
    \left\|\mathrm{ad}_{\tilde{H}(t)}^m \left(\frac{\partial H(t)}{\partial t}\right)\right\|\leq m!\Delta_q^m \left\|\frac{\partial H(t)}{\partial t}\right\|_{q\textnormal{-X}}.
\end{equation}
In all cases, the RHS is $m!\Delta_q^m \lambda'(t)n$ by definition. Taking $m=1$ and applying the bound to Eq.~(\ref{eq:g2}), we have
\begin{equation}
    g_2'(t)\leq 2 \lambda(t)\lambda'(t)n^2+\Delta_q\lambda'(t)n.
\end{equation}
Thus, by integration and $g_2(0)=0$, we obtain the upper bound
\begin{equation}
    g_2(t)\leq \lambda(t)^2n^2+\Delta_q\lambda(t)n.
\end{equation}
The procedure can continue with induction. At the $k$-th step, we have $\frac{\partial}{\partial t}G_{2k}(t)=\Braket{\tilde{H}(t)^k\frac{\partial \tilde{H}(t)^k}{\partial t}+\frac{\partial\tilde{H}(t)^k}{\partial t}\tilde{H}(t)^k}_t$ and
\begin{equation}
    \frac{\partial}{\partial t}\tilde{H}(t)^k=\sum_{m=0}^{k-1} \tilde{H}(t)^{k-m-1} \frac{\partial H(t)}{\partial t} \tilde{H}(t)^{m}.
\end{equation}
Using $VH=(H-\mathrm{ad}_H)\circ V$, we get $V H^m=(H-\mathrm{ad}_H)^m\circ V$, so that
\begin{equation}
\begin{aligned}
    \frac{\partial}{\partial t}\tilde{H}(t)^k&=\sum_{m=0}^{k-1} \tilde{H}(t)^{k-m-1} (\tilde{H}(t)-\mathrm{ad}_{\tilde{H}(t)})^m\circ \frac{\partial H(t)}{\partial t} \\
    &=\frac{\tilde{H}(t)^k-(\tilde{H}(t)-\mathrm{ad}_{\tilde{H}(t)})^k}{\mathrm{ad}_{\tilde{H}(t)}}\circ \frac{\partial H(t)}{\partial t}\\
    &=\sum_{m=1}^k(-1)^{m-1} {k \choose m}\tilde{H}(t)^{k-m}\mathrm{ad}_{\tilde{H}(t)}^{m-1}\left(\frac{\partial H(t)}{\partial t}\right),
\end{aligned}
\end{equation}
where we have used the fact that $\tilde{H}(t)$ commutes with $\mathrm{ad}_{\tilde{H}(t)}$ when treated as linear operators, so that we can manipulate them in the above equations just like classical variables.
Therefore, we have obtained
\begin{equation}
\begin{aligned}
    2 g_k(t)g_k'(t)&=\frac{\partial}{\partial t}G_{2k}(t)=\Braket{\tilde{H}(t)^k\frac{\partial \tilde{H}(t)^k}{\partial t}+\frac{\partial\tilde{H}(t)^k}{\partial t}\tilde{H}(t)^k}_t\\
    &\leq 2 g_{k}(t)\sum_{m=1}^k {k \choose m}g_{k-m}(t)(m-1)!\Delta_q^{m-1}\lambda'(t) n,
\end{aligned}
\end{equation}
which gives
\begin{equation}
    g_k'(t)\leq \sum_{m=1}^k {k \choose m}g_{k-m}(t)(m-1)!\Delta_q^{m-1}\lambda'(t) n.
\end{equation}
Dividing both sides by $\Delta_q^k$, we have\begin{equation}
    \frac{g_k'(t)}{\Delta_q^k}\leq \sum_{m=1}^k {k \choose m}\left(\frac{g_{k-m}(t)}{\Delta_q^{k-m}}\right)(m-1)!\frac{\lambda'(t) n}{\Delta_q}.
\end{equation}
For convenience of dealing with the right-hand side, we define a degree-$k$ polynomial $f_k^{(1)}(x)$ for comparison with $g_k$, with $f_0^{(1)}(x)\equiv 1$, $f_1^{(1)}(x)\equiv x$ and the recurrence relation
\begin{equation}
\label{eq:f1recur-app}
    f_k^{(1)}(x)\equiv\int_{0}^x\mathrm{d}s\sum_{m=1}^k {k \choose m}(m-1)!f_{k-m}^{(1)}(s).
\end{equation}
One can prove by induction to show that 
\begin{equation}
    g_k(t)\leq \Delta^k_{q}f_k^{(1)}\left(\frac{\lambda_t n}{\Delta_q}\right)
\end{equation}
One can further verify that (see Appendix~\ref{app:proofF1k}) 
\begin{equation}
    f_k^{(1)}(x)=\prod_{j=0}^{k-1}(x+j).
\end{equation}
By setting $x=\frac{\lambda_t n}{\Delta_q}$, we have 
\begin{equation}
    g_k(t)\leq \prod_{j=0}^{k-1}(\lambda_t n+j\Delta_q).
\end{equation}
Using the Markov inequality at $k=k_d\equiv\lceil \frac{d-\lambda_t}{\Delta_q}n\rceil$ (i.e., setting $k=k_d$ in Eq.~\eqref{eq:epsilon<gk-appA1}), we have  
\begin{equation}
    \epsilon_t(d) \leq \left(\frac{\Delta_q}{dn}\right)^{k_d}\frac{\Gamma\left(\frac{\lambda_t}{\Delta_q}n+k_d\right)}{\Gamma\left(\frac{\lambda_t}{\Delta_q}n\right)}, \label{eq:epsilon_t}
\end{equation}
where $\Gamma(x)$ is the gamma function with $\Gamma(x+1)=x\Gamma(x)$.
At the large $n$ limit, we can use Stirling's formula and obtain the bound
\begin{equation}
    \epsilon_t(d)\leq \exp\left(-n\frac{\lambda_t}{\Delta_q}\left(\frac{d}{\lambda_t}-1-\ln\frac{d}{\lambda_t}\right)-o(n)\right),
\end{equation}
where the leading order of the $-o(n)$ correction is $\frac{1}{2}\ln \frac{\lambda_t}{d}<0$. The above bound holds for any $t$ and decays exponentially for any $d>\lambda_t$.

\appsubsection{Improved bounds for Case~\ref{case:3}}
\label{app:EnergyLocPcom}

In Case~\ref{case:3}, we assume that the $H(t)$ admits a partition into a mutually commuting core $H_{C}(t)$ and the remainder $V(t)$ 
\begin{equation}
    H(t)=H_{C}(t)+V(t),
\end{equation}
such that $[V(t),H'(t)]=0$ for all $t$ and all terms in the $H_{C}(t)=\sum_{X} h_X(t)$ are mutually commuting with each other. In this case, we can use the same definition for $\lambda(t)$ and $\lambda'(t)$ as \eqref{eq:lambda_t-app} and \eqref{eq:lambda'_t-app}. We do not require that the $k$-locality of $H_C (t)$ is finite, but only require that $H'(t)$ is $q$-local with finite $q$ and $\|H_C(t)\|_{\text{loc}}\leq M$. An example of this case is $H_C(t)=\sum_{j}\sigma_j^z \sigma^z_{j+1}$ (simply time independent) and $V(t)=\frac{\lambda t}{T}\sum_j \sigma^x_j$. 

Following similar procedures to those in Appendix~\ref{app:EnergyLocP}, the key different step is to bound the scaling of the nested commutators $\mathrm{ad}_{\tilde{H}(t)}^m (H'(t))$. One can verify that the nested commutators in this case can be bounded (see Appendix~\ref{app:proofCommute}) by
\begin{equation}
\left\|\mathrm{ad}_{\tilde{H}(t)}^m \left(\frac{\partial H(t)}{\partial t}\right)\right\|\leq \Delta_q^m\left\|\frac{\partial H(t)}{\partial t}\right\|_{\textnormal{X}},
\end{equation}
where $\Delta_q \equiv 2qM$. Note that compared to the previous three cases, the upper bound here does not contain the factorial $m!$, and hence we expect a tighter bound on the leakage. Following the same steps, we can obtain the recursive bounds for $g_k(t)$
\begin{equation}
    g_k'(t)\leq \sum_{m=1}^k {k \choose m}g_{k-m}(t)\Delta_q^{m-1}\lambda'(t) n.
\end{equation}
Similarly, we can define another polynomial $f_k^{(2)}(x)$ for comparison with $g_k(t)$. In this case, the recurrence relation for $f_k^{(2)}(x)$ becomes (c.f. Eq.~\eqref{eq:f1recur-app})
\begin{equation}
\label{eq:f2recur-app}
f_k^{(2)}(x)\equiv\int_{0}^x\mathrm{d}s\sum_{m=1}^k {k \choose m}f_{k-m}^{(2)}(s),
\end{equation}
where we set $f_0^{(2)}(x)=1$ and $f_1^{(2)}(x)=x$ and we can verify that
\begin{equation}
    g_k(t)\leq \Delta_q^kf^{(2)}_k\left(\frac{\lambda_t n}{\Delta_q}\right).
\end{equation} 
One can use the generating function approach (or alternatively use the approach in Appendix.~\ref{app:proofF2k}) to  obtain the solution of the \eqref{eq:f2recur-app}
\begin{equation}
    f_{k}^{(2)}(x)=\sum_{j=0}^k \left\{k \atop j\right\}x^j,
\end{equation}
where $\left\{k \atop j\right\}$ is the Stirling number of the second kind. We can set $x=\frac{\lambda_t n}{\Delta_q}$ and use Markov's inequality to obtain
\begin{equation}
    \epsilon_t(d)^2\leq\min_k \frac{g_{k}(t)^{2}}{(dn)^{2k}} \leq\min_k \frac{f_{k}^{(2)}(x)^2}{(dn/\Delta_q)^{2k}}.
\end{equation}
Note that if $\hat{Y}$ is a random variable from the mean-$x$ Poisson distribution, its $k$th moment is known to be exactly $f^{(2)}_k(x)$, i.e., $\mathbf{E}(\hat{Y}^k)=f_k^{(2)}(x)$~\cite{riordan1937moment}. In addition, its moment-generating function
\begin{equation}
    \mathbf{E}(e^{\tau \hat{Y}})=\sum_{k=0}^\infty \frac{\tau^k}{k!}\mathbf{E}(\hat{Y}^k)=e^{x(e^\tau-1)}
\end{equation} 
also exists for any $\tau$. A known result~\cite{philips1995moment} on bounding the best $k$-moment bound by the Chernoff bound gives the following relation:
\begin{equation}
    \min_k \frac{f_{k}^{(2)}(x)}{(dn/\Delta_q)^k} \leq \inf_{\tau\geq 0} \mathbf{E}(e^{\tau \hat{Y}})e^{-\frac{dn}{\Delta_q}\tau}.
\end{equation}
Then we can obtain an explicit upper bound for $d>\lambda_t$:
\begin{equation}
\begin{aligned}
\epsilon_t(d)&\leq \inf_{\tau\geq 0} \exp\left(x(e^{\tau}-1)-\frac{dn}{\Delta_q}\tau\right)\\
&=\exp\left(-n\frac{\lambda_t}{\Delta_q}\left(\frac{d}{\lambda_t}\ln \frac{d}{\lambda_t}-(\frac{d}{\lambda_t}-1)\right)\right)  .
\end{aligned} 
\label{eq:epsilon_t2}
\end{equation}

\appsubsection{Bounds for perturbations of general scalings}
\label{app:generalScalingP}

In previous sections we discussed cases with linear-scaling perturbations and energy windows, with bounded local norms. Here, we roughly estimate the leakage bounds in general settings where all parameters can scale differently with $n$.

In general, we can still define
\begin{equation}
    \|H(t)\|_{\textnormal{loc}}\leq M,
\end{equation}
where $M$ is a constant that may scale with $n$ instead of an $O(1)$ constant previously, but still independent of $t$.

For a strictly $q$-local $H(t)$, we can define
\begin{equation}
    \Lambda_t\equiv \int_0^t \left\|\frac{\partial H(\tau)}{\partial \tau}\right\|_{\textnormal{X}}\mathrm{d}\tau,
\end{equation}
where $\Lambda_t$ can have generic scaling of $n$ apart from the previous linear scaling. (Note we use the symbol $\Lambda$ here, and in the previous cases, $\Lambda_t= \lambda(t)n$.) Since $\|H(t)\|\leq nM$, we require that 
\begin{equation}
    \Lambda_t \lesssim O(nM),
\end{equation}
otherwise the change of the $H(t)$ will be too large to be called perturbations. One can also define $\Lambda_t \equiv\int_0^t \|\mathrm{d}H\|_{q\textnormal{-X}}$ for quasi-$q$-local perturbations.

Similarly, we set the initial state $\ket{\psi(0)}$ as one eigenstate of $H(0)$ with energy $E_0$, and let the state $\ket{\psi(t)}$ evolve under $H(t)$. At time $t$, we define the eigenstates of $H(t)$ as $\ket{\phi_j(t)}$ with energy $E_j(t)$, and expand the evolving state in the instantaneous eigenbasis
\begin{equation}
    \ket{\psi(t)}=\sum_j a_j(t)\ket{\phi_j(t)}.
\end{equation}
We set the energy window $D>\Lambda_t$ (usually the same scaling as $\Lambda_t$) accordingly and define the leakage $\Xi_t(D)$ as 
\begin{equation}
    \Xi _t(D)\equiv\sqrt{\sum_{j,|E_j(t)-E_0|\geq D}|a_j(t)|^2}. 
\end{equation}
Following similar procedures, we can bound the cases whose nested commutators $\mathrm{ad}_{H(t)}^m(H'(t))$ increase factorially $O(m!\Delta_q^m)$ by $\Xi^{(1)}_t(D)$
\begin{equation}
\Xi^{(1)}_{\Lambda_t,\Delta_q}(D)\equiv\left(\frac{\Delta_q}{D}\right)^{k_D}\frac{\Gamma\left(\frac{\Lambda_t}{\Delta_q}+k_D\right)}{\Gamma\left(\frac{\Lambda_t}{\Delta_q}\right)},
\end{equation}
where $\Delta_q\equiv2qM$ and $k_D\equiv\left\lceil\frac{D-\Lambda_t}{\Delta_q}\right\rceil$ is well defined only for $D>\Lambda_t$. This can be compared to Eq.~\eqref{eq:epsilon_t} in the previous case.

Similarly, for $H(t)$ whose $\mathrm{ad}_{H(t)}^m(H'(t))$ scales exponentially $O(\Delta_q^m)$, the leakage can be bounded by 
\begin{equation}
    \Xi^{(2)}_{\Lambda_t,\Delta_q}(D)\equiv \exp\left(-\frac{\Lambda_t}{\Delta_q}\left(\frac{D}{\Lambda_t}\ln\frac{D}{\Lambda_t}-\big(\frac{D}{\Lambda_t}-1\big)\right)\right),
\end{equation}
which can be compared to Eq.~\eqref{eq:epsilon_t2}.

If we set $D/\Lambda_t>1$ to be a $\Theta(1)$ constant independent of $t$ and $\Lambda_t$ is large, one can verify that both bounds scale the same as
\begin{equation}
    \Xi_{t}(D)\lesssim \exp\left(-\Omega\big(\frac{\Lambda_t}{2qM}\big)\right).
\end{equation}
Specifically, for $H(t)$ with a $\Theta(1)$ local norm $M$ and a general scaling $\Lambda_t \sim \Theta(f(n))$, the bound is roughly $e^{-\Omega(f(n))}$.

\appsubsection{Properties of $f_k^{(1)}$}
\label{app:proofF1k}
In this section, we prove that the $f^{(1)}_k(x)$'s we found 
\begin{equation}\label{eq:f1nx-pro}
    f_n^{(1)}(x)=\prod_{m=0}^{n-1}(x+m)=\sum_{j=0}^{n}\left[{n\atop j}\right]x^j
\end{equation}
indeed satisfies the recurrence relation \eqref{eq:f1recur-app}, where $\left[{n\atop j}\right]$ is the unsigned Stirling number of the first kind by definition. The  $\left[{n\atop j}\right]$ also has a combinatorial interpretation: $\left[{n\atop j}\right]$ counts the number of permutations when dividing the $n$ elements into $j$ disjoint cycles. It is easy to verify that $f_0^{(1)}(x)$ and $f_1^{(1)}(x)$ satisfy the relation \eqref{eq:f1recur-app}. Suppose that Eq.~\eqref{eq:f1nx-pro} holds for all $n<k$. At the $k$-th step, we can obtain $f_k^{(1)}(x)$ from the recurrence relation
\begin{equation}
\begin{aligned}
    f_{k}^{(1)}&(x)=\sum_{m=1}^k{k\choose m} (m-1)!\sum_{j=0}^{k-m}\left[{k-m\atop j}\right] \frac{x^{j+1}}{j+1}\\
    &=\sum_{j=0}^{k-1}\frac{x^{j+1}}{j+1}\sum_{m=1}^{k-j}{k\choose m} (m-1)! \left[{k-m\atop j}\right].
\end{aligned}
\end{equation}
We are going to prove the following identity combinatorially:
\begin{equation}
    \frac{1}{j+1}\sum_{m=1}^{k}{k\choose m} (m-1)!\left[{k-m\atop j}\right] = \left[{k\atop j+1}\right].
\end{equation}
The RHS can be easily interpreted as the number of permutations in dividing $k$ elements into $j+1$ disjoint cycles. The LHS can be interpreted as follows: 1. Choose a cycle $C$ of $m$ elements from the total $k$ elements, which gives ${k \choose m}$ choices. The cycle $C$ has $(m-1)!$ permutations. 2. Dividing the rest of elements into $j$ cycles, which gives $\left[{k-m\atop j}\right]$ permutations. 3. Sum over all possible $m$, we get the number of permutations in dividing $k$ elements into $j+1$ disjoint cycles. As any of the $j+1$ cycles can be labeled as $C$, we need an additional factor $1/(j+1)$ for multiple counting. Therefore, the identity is proved. Using $\left[{k\atop 0}\right]=0$ and $\left[{i\atop j}\right]=0$ for any $i<j$, we obtain
\begin{equation}
    f_k^{(1)}(x)=\sum_{j=0}^{k-1}\left[{k\atop j+1}\right]x^{j+1}=\sum_{j=0}^{k}\left[{k\atop j}\right]x^{j},
\end{equation}
which is exactly Eq.~\eqref{eq:f1nx-pro}.

\appsubsection{Properties of $f_k^{(2)}$}
\label{app:proofF2k}
In this section, we want to confirm that
\begin{equation}\label{eq:f2nx-pro}
    f_n^{(2)}(x)=\sum_{j=0}^{n}\left\{{n\atop j}\right\}x^j
\end{equation}
satisfies the recurrence relation \eqref{eq:f2recur-app}, where $\left\{{n\atop j}\right\}$ is the Stirling number of the second kind. The $\left\{{n\atop j}\right\}$ has a combinatorial interpretation: $\left\{{n\atop j}\right\}$ counts the number of ways to partition $n$ elements into $j$ unlabeled non-empty sets. It is easy to verify that $f_0^{(2)}(x)=1$ and $f_1^{(2)}(x)=x$ satisfy the relation \eqref{eq:f2recur-app}. Suppose that Eq.~\eqref{eq:f2nx-pro} holds for all $n<k$. In the $k$-th step, we can obtain $f_k$ from the recurrence relation.
\begin{equation}
\begin{aligned}
    f_{k}^{(2)}&(x)=\sum_{m=1}^k{k\choose m} \sum_{j=0}^{k-m}\left\{{k-m\atop j}\right\} \frac{x^{j+1}}{j+1}\\
    &=\sum_{j=0}^{k-1}\frac{x^{j+1}}{j+1}\sum_{m=1}^{k-j}{k\choose m}  \left\{{k-m\atop j}\right\}.
\end{aligned}
\end{equation}
We are going to prove the following identity combinatorially:
\begin{equation}
    \frac{1}{j+1}\sum_{m=1}^k{k\choose m}\left\{{k-m\atop j}\right\} = \left\{{k\atop j+1}\right\}.
\end{equation}
The RHS can be easily interpreted as the number of partitions in dividing $k$ elements into $j+1$ non-empty sets. The LHS can be interpreted as follows: 1. Choose a set $C$ containing $m$ elements from the total $k$ elements, which gives ${k \choose m}$ choices. 2. Divide the rest of the elements into $j$ non-empty sets, which gives $\left\{{k-m\atop j}\right\}$ ways. 3. Sum over all possible $m$, we get the number of partitions in dividing $k$ elements into $j+1$ nonempty sets. As any of the $j+1$ sets can be labeled as $C$, we need an additional factor $1/(j+1)$ for multiple counting. Therefore, the identity is proved. Using $\left\{{k\atop 0}\right\}=0$ and $\left\{{i\atop j}\right\}=0$ for any $i<j$, we obtain
\begin{equation}
    f_k^{(2)}(x)=\sum_{j=0}^{k-1}\left\{{k\atop j+1}\right\}x^{j+1}=\sum_{j=0}^{k}\left\{{k\atop j}\right\}x^{j},
\end{equation}
which is exactly Eq.~\eqref{eq:f2nx-pro}.

\begin{table*}[t]
\begin{tabular}{|c|c|c|c|c|}
        \hline
        & \multicolumn{2}{c|}{Dynamical localization}  & Eigenstate & Slow mixing of  \\ 
        \cline{2-3}
        Perturbation type  & Exponential time & Infinite time &  localization & Gibbs sampler \\
        \hline 
        Static strictly $q$-local  & $\lambda$ \eqref{eq:lambdaExp-q}, $\xi$ \eqref{eq:xi-DyLoc-q-app} &$\lambda$ \eqref{eq:lambdaForInf-app}, $\xi$ \eqref{eq:xi-InfLoc-q-app} &$\lambda$ \eqref{eq:lambdaForEigen-app}, $\xi$  \eqref{eq:xi-eigenLoc-q-app}  & $\lambda$ \eqref{eq:lambda-Gibbs-q-app}, $\xi$ \eqref{eq:xi-Gibbs-q-app}\\
        Static quasi-$q_{\star}$-local  &  $\lambda$ \eqref{eq:lambdaForExp-app}, $\xi$ \eqref{eq:xi-DyLoc-Quasiq-app} & $\lambda,q_\star$ \eqref{eq:Lambda-InfLoc-Quasiq-app}, $\xi$ \eqref{eq:xi-InfLoc-QuasiQ-app}  & $\lambda,q_\star$ \eqref{eq:Lambda-eigenLoc-Quasiq-app}, $\xi$ \eqref{eq:xi-eigenLoc-Quasiq-app} &  $\lambda,q_\star$ \eqref{eq:lambda-Gibbs-QuasiQ-app}, $\xi$ \eqref{eq:xi-Gibbs-QuasiQ-app} \\
        Time-dependent (quasi-)$q$-local  & \eqref{eq:lambdaExp-q}, \eqref{eq:xi-DyLoc-q-app}; \eqref{eq:lambdaForExp-app}, \eqref{eq:xi-DyLoc-Quasiq-app}&- &- & -\\
        \hline 
\end{tabular}
\caption{Summary of requirements for perturbation strength parameters $\lambda$ \& $q_\star$ and leakage bounds $e^{-\xi n}$ for different stabilities. For strictly $q$-local perturbations, the strength is depicted by $\lambda$. For quasi-$q_{\star}$-local perturbations, the strength is depicted by both $\lambda$ and $q_{\star}$. In closed systems, these parameters are mainly determined by the energy barrier height $bn$ and the initial energy density $\varepsilon_0$. In open systems, they are also determined by the inverse temperature $\beta$.}
\label{tab:summary-results2}
\end{table*}

\appsection{Proof of the dynamical localization}
\label{app:pDynamicLoc}

We consider a state $\ket{\psi(t)}$ initialized at eigenstate of $H_C$ inside $w_0$ with energy $E_0=\varepsilon_0 n+E_g$, evolving under $H(t)=H_C+V(t)$. We assume $H_C$ is a good c/qLDPC code with linear soundness and satisfies the property discussed in Sec.~\ref{sec:LDPC}. An essential property from the low-density requirement is that $H_C$ is at most $q_C$-local and 
\begin{equation}
    \|H_C\|_{\textnormal{loc}}\leq p_C,
\end{equation}
where $p_C$ is the maximum number of checks involving one single qubit and $q_C$ is the maximum qubit number in one check.

We are going to prove that $\ket{\psi(t)}$ remains inside $w_0$ up to an exponentially/infinitely long time under different conditions, with the leakage being exponentially small.
The dynamical localization is defined in the eigenbasis of $H_C$. For eigenstates $\ket{\phi_j}$ of $H_C$, we denote $P_w\equiv\sum_{\phi_j\in w} \ket{\phi_j}\bra{\phi_j}$ as the projector to the cluster $w$ and $P_W=\sum_{w}P_w$ as the projector to \textit{all} clusters. But there are also high-energy states not in any cluster, and we denote the projector to these high-energy states as $P_{>}\equiv1-P_{W}$, i.e.,
\begin{equation}
    P_{>}=1-\sum_{w}P_w.
\end{equation} 
Since all eigenstates of $H_C$ below $E_B\equiv E_g+bn$ are divided into clusters $w$'s, $P_{>}$ projects a state to its components with energies higher than $E_B$.

As different clusters are globally separated, a state cannot transit to other clusters directly without first overcoming the energy barriers. Thus, an essential step in proving dynamical localization is to bound the total weight of high-energy eigenstates of $H_C$ at arbitrary time $t$. To bound the high-energy weights at time $t_1$, we can invoke Theorem~\ref{thm:Energy-local} for an extended evolution under the following segments of the Hamiltonian
\begin{equation}
\label{eq:DyExt-app}
    H_{\textnormal{ext}}(t)=\begin{cases}
        H_C &~\textnormal{for} ~t=0^{-},\\
        H_C+V(t) &~\textnormal{for} ~0\leq t\leq t_1,\\
        H_C &~\textnormal{for} ~t=t_1^{+},
    \end{cases}
\end{equation}
where $H_{\textnormal{ext}}(t)$ at $t=0^-,0$ and $t_1,t_1^+$ are linearly connected (similar to Eq.~\eqref{eq:tTostatic}) and we let the system evolve from $0^{-}$ to $t_1^{+}$. From Schr\"odinger's equation, we know that $\ket{\psi(t_1^+)}=\ket{\psi(t_1)}$, and thus we can bound the high-energy leakage of $\ket{\psi(t_1)}$ in the eigenbasis of $H_C$ at any time $t_1$. This results in the following lemma.

\begin{lem}
\label{lem:P>-app}
We consider the $\ket{\psi(t)}$ evolving under $H(t)=H_C+V(t)$ defined above with strictly $q$-local $V(t)$, where we assume $q\geq q_C$ for simplicity. For Case~\ref{case:1} and Case~\ref{case:2}, we can bound the high-energy contribution in the eigenbasis of $H_C$ by
\begin{equation}
    \|P_{>}\ket{\psi(t)}\|\leq \epsilon^{(1)}_{2\lambda_t,\Delta_q}(b-\varepsilon_0),
\end{equation}
For Case~\ref{case:3}, we have
\begin{equation}
    \|P_{>}\ket{\psi(t)}\|\leq \epsilon^{(2)}_{2\lambda_t,\Delta_q}(b-\varepsilon_0),
\end{equation}
where $\lambda_t\equiv\frac{1}{n}\int_{0^-}^{t}\left\|H_{\text{ext}}'(\tau)\right\|_{\textnormal{X}} \mathrm{d}\tau$ and $\Delta_q=2qM$. If $V(t)$ is quasi-$q_{\star}$-local and $H(t)$ satisfies Case~\ref{case:4}, we have
\begin{equation}
    \|P_{>}\ket{\psi(t)}\|\leq \epsilon^{(1)}_{2\lambda_t,\Delta_{q_{\star}}}(b-\varepsilon_0),
\end{equation}
where $\lambda_t\equiv\frac{1}{n}\int_{0^-}^{t}\left\| H_{\text{ext}}'(\tau)\right\|_{q_{\star}\textnormal{-X}} \mathrm{d}\tau$ and $\Delta_{q_{\star}}=2q_{\star}M$. In all cases we only require $\lambda_t<(b-\varepsilon_0)/2$, or equivalently, $\varepsilon_0<b-2\lambda_t$.
\end{lem}

\begin{proof}
At a given time $t=t_1$, we consider $\ket{\psi(t_1^+)}$ from the extended evolution $H_{\text{ext}}(t)$ in \eqref{eq:DyExt-app}, which produces exactly the same state as $\ket{\psi(t_1)}$ under $H(t)=H_C+V(t)$ from $t=0$ to $t_1$. For the evolution from $t_1$ to $t_1^+$,  the Hamiltonian is effectively $H_{\textnormal{ext}}(t)=H_C+(1-\frac{t-t_1}{\tau})V(t_1)$ for $t\in[t_1,t_1+\tau]$ with $\tau\rightarrow 0^+$. The evolution of the state and the Hamiltonian is then
\begin{equation}
\begin{aligned}
    H_C\xrightarrow{t\in[0^-,0]} &H(0)\xrightarrow{~t\in [0,t_1]~} H(t_1)\xrightarrow{t\in [t_1,t_1^+]}H_C,\\
    \ket{\psi(0)}\xrightarrow{t\in[0^-,0]} &\ket{\psi(0)}\xrightarrow{t\in [0,t_1]} \ket{\psi(t_1)}\xrightarrow{t\in [t_1,t_1^+]}\ket{\psi(t_1)}.
\end{aligned}
\end{equation} 
We then use Theorem~\ref{thm:Energy-local} to bound the high-energy components in the instantaneous eigenbasis of $H_{\textnormal{ext}}(t_1^{+})$ from from $t=0^-$ to $t_1^{+}$. By construction, now the instantaneous Hamiltonian $H_{\textnormal{ext}}(t_1^{+})$ is exactly the original Hamiltonian
\begin{equation}
    H_{\textnormal{ext}}(t_1^{+})=H_C.
\end{equation}
Therefore, the high-energy component $\|P_{>}\ket{\psi(t_1)}\|$ in the eigenbasis of $H_C$ at $t_1$ can be directly bounded by Theorem~\ref{thm:Energy-local} after the extension.

With such extension of evolution, the total variation here increases by the additional virtual evolution from $H(t_1)=H_C+V(t_1)$ back to $H_C$ from $t_1$ to $t_1^{+}$. The contribution to the total variation is exactly $\|V(t_1)\|_{\textnormal{X}}$. We can use
\begin{equation}
   \int_{t_1}^{t_1^+} \left\|\mathrm{d}H_{\textnormal{ext}}\right\|_{\mathrm{X}}=\|V(t_1)\|_{\textnormal{X}}=\left\|\int^{t_1}_{0^-} \mathrm{d}H_{\textnormal{ext}}\right\|_{\mathrm{X}}\leq \int^{t_1}_{0^-} \|\mathrm{d}H_{\textnormal{ext}}\|_{\mathrm{X}}\equiv\lambda_{t_1} n
\end{equation}
to bound the variation by $\lambda_{t_1} n$. Note that we have re-defined the total variation as $\int^{t_1}_{0^-} \|\mathrm{d}H_{\textnormal{ext}}\|_{\mathrm{X}}$ instead of $\int^{t_1}_{0} \|\mathrm{d}H\|_{\mathrm{X}}$. Thus, we can bound the total variation for the extended evolution by $2\lambda_{t_1} n$, at a cost of a factor of two, 
\begin{equation}
    \int_{0^-}^{t_1^+} \|H_{\textnormal{ext}}'(t)\|_{\textnormal{X}} \mathrm{d}t\leq 2\lambda_{t_1} n.
\end{equation}
Since the procedure works for arbitrary $t_1$, by setting 
\begin{equation}
d=b-\varepsilon_0,~\lambda=2\lambda_t,~\Delta=\Delta_q\equiv 2qM,
\end{equation}
in Theorem~\ref{thm:Energy-local} for Case~\ref{case:1} - \ref{case:4}, we obtain the all the bounds respectively.
\end{proof}

We remark that for Cases~\ref{case:2} - \ref{case:4}, the condition $[H'(t),V(t)]=0$ is reduced to $[V'(t),V(t)]=0$, which can be satisfied for perturbations such as $V(t)=f(t)V_0$. In the following, we show the detailed proof for the localization time by using the above lemma that we just proved.

\appsubsection{Exponentially long time dynamical localization}
\label{app:DyLocExpProof}

In this section, we consider the evolution under the following time-dependent Hamiltonian:
\begin{equation}
    H(t)=H_C+V(t),
\end{equation}
where $V(t)$ is a time-dependent perturbation. We require that $V(t)$ satisfies $\int_{0^-}^T \|H_{\textnormal{ext}}'(t)\|_{\textnormal{X}}\leq \lambda n$ for strictly $q$-local cases. This is equivalent to
\begin{equation}
    \|V(0)\|_{\textnormal{X}}+\int_0^T  \|V'(t)\|_{\textnormal{X}}\leq \lambda n.
\end{equation}
Similarly, for quasi-$q_{\star}$-local $V(t)$ in Case~\ref{case:4}, we require that 
\begin{equation}
    \|V(0)\|_{q\textnormal{-X}}+\int_0^T  \|V'(t)\|_{q\textnormal{-X}}\leq \lambda n.
\end{equation}

\subsubsection{Strictly $q$-local perturbation}
Here, we are going to prove the dynamical localization for Case~\ref{case:3}, and the results for Case~\ref{case:1} and \ref{case:2} can be obtained by replacing $\epsilon^{(2)}$ with $\epsilon^{(1)}$.

In condition $q<\nu_2$, we have $P_wV(t)P_{w'}=0$ for $w\neq w'$, i.e., there are no transition elements between two different clusters. Consider the expectation of $P_{w_0}$ for an initial eigenstate with energy $\varepsilon_0 n+E_g$ inside the cluster $w_0$. Its derivative is given by
\begin{equation}
    \frac{\mathrm{d}}{\mathrm{d}t}\braket{\psi(t)|P_{w_0}|\psi(t)}=i\Braket{[H(t),P_{w_0}]}_t.
\end{equation}
Note that
\begin{equation}
\begin{aligned}
    [H(t),P_{w_0}]&=[V(t),P_{w_0}]\\
    &=(1-P_{w_0})V(t)P_{w_0}-P_{w_0}V(t)(1-P_{w_0})\\
    &=P_{>}V(t)P_{w_0}-P_{w_0}V(t)P_{>},
\end{aligned}
\end{equation}
where we have used $P_{w'}V(t)P_w=0$ for $w\neq w'$ and $P_{>}=1-\sum_{w}P_w$ in the second and third lines above. We have already proved in Lemma.~\ref{lem:P>-app} that $\|P_{>}\ket{\psi(t)}\|$ is exponentially small
\begin{equation}
    \|P_{>}\ket{\psi(t)}\|\leq \epsilon^{(2)}_{2\lambda_t,\Delta_q}(b-\varepsilon_0).
\end{equation}
Thus, we can bound the derivative
\begin{equation}
    \left|\frac{\mathrm{d}}{\mathrm{d}t}\braket{P_{w_0}}_t\right|\leq 2\|P_{>}\ket{\psi(t)}\|\cdot\|V(t)P_{w_0}|\psi(t)\rangle\|\leq 2\lambda_t n \epsilon^{(2)}_{2\lambda_t,\Delta_q}(b-\varepsilon_0).
\end{equation}
Using $\braket{P_{w_0}}_{t=0}=1$ and $\lambda_t\leq \lambda_T\equiv\lambda$, we can obtain
\begin{equation}
    \braket{1-P_{w_0}}_T\leq 2\lambda nT \epsilon^{(2)}_{2\lambda,\Delta_q}(b-\varepsilon_0).
\end{equation}
Since $\epsilon^{(2)}_{2\lambda,\Delta_q}(b-\varepsilon_0)\leq e^{-\xi n}$, where
\begin{equation}
\label{eq:xi-DyLoc-q-app}
    \xi=\frac{2\lambda}{\Delta_q}\left(\frac{b-\varepsilon_0}{2\lambda}\ln\frac{b-\varepsilon_0}{2\lambda}-\big(\frac{b-\varepsilon_0}{2\lambda}-1\big)\right),
\end{equation}
the leakage outside $w_0$ is exponentially small as long as $T\leq \frac{1}{2\lambda n}  e^{\xi_1 n}$ with any $\xi_1<\xi$. 

Therefore, any initial eigenstate of $H_C$ localized in $w_0$, with energy $E_g+\varepsilon_0 n$ with $\varepsilon_0<b$, remains localized inside $w_0$ under $q$-local time-dependent perturbation $V(t)$ up to exponentially long time $T<\frac{1}{2\lambda n}e^{\xi_1 n}$, as long as the perturbation strength satisfies
\begin{equation}
\label{eq:lambdaExp-q}
    \lambda <\frac{b-\varepsilon_0}{2}\equiv\lambda_{\textnormal{Exp}},
\end{equation}
with leakage outside $w_0$ bounded by $e^{-(\xi-\xi_1)n}$. Here, one can choose any $\xi_1$ such that $0<\xi_1<\xi$ to fulfill the exponentially long dynamical localization.

\subsubsection{Quasi-$q_{\star}$-local perturbation}
Here, we are going to prove the dynamical localization for Case~\ref{case:4}. As a quasi-$q$-local perturbation contains terms that can be very non-local, thus the trick of zero matrix elements between two different clusters does not work. We have to proceed with a slight modification. For $H(t)=H_C+V(t)$ with quasi-$q_\star$-local $V$, we define 
\begin{equation}
    H(t)=H_C+V_L(t)+V_R(t),
\end{equation}
where $V(t)$ is divided into the local part $V_L(t)$ and the remainder $V_R(t)$. $V_L(t)$ includes all terms in $V(t)$ with the $k$-locality satisfying $k<\nu_2$ with $\nu_2\sim O(n)$, so that they do not introduce direct matrix elements between different clusters, and the remainder $V_R(t)\equiv V(t)-V_L(t)$ contains all the non-local terms in $V(t)$. Then the local part $V_L(t)$ satisfies
\begin{equation}
    P_{w'}V_L(t)P_w=0, ~\text{for }~w\neq w'.
\end{equation}
By $\int_{0^{-}}^t\|H'(\tau)\|_{q_\star\textnormal{-X}} \mathrm{d} \tau\leq \lambda_t n$, we have
\begin{equation}
    \|V_R(t)\|\leq \lambda_t n e^{-\nu_2/q_{\star}}.
\end{equation}
Similarly, we have
\begin{equation}
    \frac{\mathrm{d}}{\mathrm{d}t}\braket{\psi(t)|P_{w_0}|\psi(t)}=i\Braket{[H(t),P_{w_0}]}_t.
\end{equation}
We can rewrite the r.h.s. by noting that
\begin{equation}
    [H(t),P_{w_0}]=[V_L(t)+V_R(t),P_{w_0}]    =P_{>}V_L(t)P_{w_0}-P_{w_0}V_L(t)P_{>}+[V_R(t),P_{w_0}].
\end{equation}
Using Lemma~\ref{lem:P>-app}, we can bound the derivative by
\begin{equation}
    \left|\frac{\mathrm{d}}{\mathrm{d}t}\braket{P_{w_0}}_t\right|\leq 2\lambda_t n \left(\epsilon^{(1)}_{2\lambda_t,\Delta_{q_{\star}}}(b-\varepsilon_0)+ e^{-\nu_2/q_{\star}}\right).
\end{equation}
Using $\braket{P_{w_0}}_{t=0}=1$ and $\lambda_t\leq \lambda_T\equiv\lambda$, we obtain
\begin{equation}
    \braket{1-P_{w_0}}_T\leq 2\lambda nT \left(\epsilon^{(1)}_{2\lambda,\Delta_{q_{\star}}}(b-\varepsilon_0)+ e^{-\nu_2/q_{\star}}\right).
\end{equation}
Here $\nu_2=d_C-2\gamma n\sim O(n)$ by definition. Since both $\epsilon^{(1)}_{2\lambda,\Delta_{q_{\star}}}(b-\varepsilon_0) $ and $e^{-\nu_2/q_{\star}}$ decay exponentially with $n$, we can define
\begin{equation}
\label{eq:xi-DyLoc-Quasiq-app}
    \xi=\min\left(\frac{2\lambda}{\Delta_{q_{\star}}}\Big(\frac{b-\varepsilon_0}{2\lambda}-1-\ln\frac{b-\varepsilon_0}{2\lambda}\Big),\frac{\nu_2}{q_{\star}n}\right).
\end{equation}
Then the leakage outside $w_0$ is exponentially small as long as $T\leq \frac{1}{2\lambda n} e^{\xi_1 n}$ with any $\xi_1<\xi$.

Therefore, any initial eigenstate of $H_C$ localized in $w_0$, with energy $E_g+\varepsilon_0 n$ with $\varepsilon_0<b$, remains localized inside $w_0$ under quasi-$q_{\star}$-local time-dependent perturbation $V(t)$ up to exponentially long time $T<e^{\xi_1 n}$, as long as the perturbation strength satisfies
\begin{equation}
\label{eq:lambdaForExp-app}
    \lambda <\frac{b-\varepsilon_0}{2}\equiv\lambda_{\textnormal{Exp}},
\end{equation}
with leakage outside $w_0$ bounded by $e^{-(\xi-\xi_1)n}$. Here, one can choose any $\xi_1$ such that $0<\xi_1<\xi$. We note that the requirement for $\lambda$ is the same as that for $q$-local perturbation, but the requirement for $\xi_1$ is different.

\appsubsection{Infinitely long time dynamical localization in cLDPC codes}
\label{app:DyLocInfinite}

In this section, we are going to show that the dynamical localization can be infinitely long in some classical models with linear energy barriers. 

We consider the following static Hamiltonian (similarly defined in~\cite{yin2024eigenstate})
\begin{equation}
    H=H_C+V_0+H_{\textnormal{d}},
\end{equation}
where $H_C$ is a good cLDPC code with linear soundness and $H_{\textnormal{d}}$ is a small diagonal detuning defined by
\begin{equation}
    H_{\textnormal{d}}\equiv \frac{1}{n^2}\sum_{j}h_j^z Z_j,
\end{equation}
where $h_j^z$ can be independently randomly chosen from the Gaussian distribution with mean 0 and variance 1. We note that $H_{\textnormal{d}}$ only contributes $O(\frac{1}{n})$ to the total variation of $H$ from $H_C$, so it is negligible when considering leakage bounds in the energy space.

While the setup of $H$ here is only a special case in the exponentially-long-time dynamical localization discussed in Appendix~\ref{app:DyLocExpProof}, it remains nontrivial to prove that the localization time can be truly infinite at finite $n$.

In Appendix~\ref{app:pEigenLoc}, we prove that if $\lambda\sim O(1)$ is small enough (also $q_\star\sim O(1)$ is small enough for quasi-$q_\star$-local perturbations), the Hamiltonian $H$ exhibits eigenstate localization for $E'<E_g+bn$ (see explicit bounds for $\lambda$ in Eq.~\eqref{eq:lambdaBoundsEigenLocQ-app} and \eqref{eq:lambdaBoundsEigenLocQuasiQ-app}). Here, the bound for $\lambda$ is a function of $E_B-E'$ and will vanish if $E_B-E'\rightarrow 0$ 

Suppose we have an initial state $\ket{\psi(t)}$ initialized at eigenstate of $H_C$ inside $w_0$ with energy $E_0=\varepsilon_0 n+E_g$, evolving under $H=H_C+V_0+H_{\textnormal{d}}$. From Theorem~\ref{thm:static-Energy-loc} we know that $\ket{\psi(0)}$ roughly is exponentially localized in the energy window $[E_0-\lambda n,E_0+\lambda n]$ in the spectrum of $H$. From Appendix~\ref{app:pEigenLoc}, we know that any eigenstate of $H$ below $E_B\equiv E_g+bn$ is localized in one single cluster if $\lambda$ is sufficiently small. Thus, as long as the majority of $\ket{\psi(0)}$ comes from the localized eigenstates inside $w_0$, we can bound the leakage for all time. We show the detailed proof for this argument below.

\subsubsection{Strictly $q$-local perturbation}
In this subsection, we assume $V_0+H_{\mathrm{d}}$ is strictly $q$-local and 
\begin{equation}
    \|V_0+H_{\mathrm{d}}\|_{\textnormal{X}}\leq\lambda n.
\end{equation}
To rigorously demonstrate our arguments at the beginning of the section, we can expand the $\ket{\psi(0)}$, an eigenstate of $H_C$ with energy $E_g+\varepsilon_0 n$, in the eigenbasis of $H$
\begin{equation}
\label{eq:psi0=3termsInfiniteLoc-app}
\begin{aligned}
    \ket{\psi(0)}&=\sum_{j} a_{\omega_0,j}\ket{\phi_{\omega_0,j}}+\sum_{\omega\neq\omega_0 ,j}a_{\omega,j}\ket{\phi_{\omega,j}}+\ket{\psi_R}\\
    &\equiv\ket{\psi_{w_0}}+\ket{\psi_{\neq w_0}}+\ket{\psi_R},
\end{aligned}
\end{equation}
where the first two terms contain all localized eigenstates (see Appendix~\ref{app:pEigenLoc} for details) with energy below $E_g+(\varepsilon_0+d)n$ (with $d>\lambda$) of $H$ and $\ket{\psi_R}$ is the remainder. The notion $\ket{\phi_{\omega_0,j}}$ denotes that it is localized inside $w_0$ with exponentially small leakage, which is proved in Appendix~\ref{app:pEigenLoc}. The eigenstate localization condition requires that $\varepsilon_0+d<b$. From Theorem~\ref{thm:static-Energy-loc} we know that the leakage is 
\begin{equation}
    \|\ket{\psi_R}\|\leq \epsilon^{(2)}_{\lambda,\Delta_q}(d).
\end{equation}
From Eq.~\eqref{eq:EigenLocQLeak-app} we know that for energy below $E_g+(\varepsilon_0+d)n$, all eigenstates of $H$ is localized when $\delta_W$ is sufficiently large, which is satisfied with high probability $1-e^{-2(1-\ln 2)n}$ and is shown in Appendix~\ref{app:pEigenLoc}. The leakage outside $w_0$ can be bounded by
\begin{equation}
    \|(1-P_w)\ket{\phi_{w,j}}\|\leq e^{-\xi(\varepsilon_0+d,\lambda) n},
\end{equation}
where we can set $\varepsilon'=\varepsilon_0+d$ and 
\begin{equation}
    \xi(\varepsilon',\lambda)\equiv\frac{b-\varepsilon'}{\Delta_q}\ln\frac{b-\varepsilon'}{\lambda}-\frac{b-\varepsilon'-\lambda}{\Delta_q}-2.
\end{equation}
Thus, the contribution to $w_0$ in $\ket{\psi_{\neq w_0}}$ of Eq.~\eqref{eq:psi0=3termsInfiniteLoc-app} can be bounded by 
\begin{equation}
    \left\|P_{w_0}\sum_{\omega\neq\omega_0 ,j}a_{\omega,j}\ket{\phi_{\omega,j}}\right\| \leq 2^{n/2}e^{-\xi(\varepsilon',\lambda) n}.
\end{equation}
Thus, contribution of the $\ket{\psi_{w_0}}$ of Eq.~\eqref{eq:psi0=3termsInfiniteLoc-app} cannot be too small
\begin{equation}
\begin{aligned}
   & \|\ket{\psi_{\omega_0}}\|\geq \|P_{\omega_0}\ket{\psi_{\omega_0}}\|\\
    &\geq \|P_{\omega_0}\ket{\psi(0)}\|-\|P_{\omega_0}\ket{\psi_{\neq w_0}}\|-\|P_{\omega_0}\ket{\psi_{R}}\|\\
    &\geq 1-e^{-(\xi-\frac{1}{2}\ln 2)n}-\epsilon^{(2)}_{\lambda,\Delta_q}(d)\equiv1-\epsilon_{\neq\omega_0},
\end{aligned}
\end{equation}
where we have defined
\begin{equation}
    \epsilon_{\neq\omega_0}\equiv e^{-(\xi-\frac{1}{2}\ln 2)n}+\epsilon^{(2)}_{\lambda,\Delta_q}(d).
\end{equation}
We can also bound $\|\ket{\psi_{\neq \omega_0}}\|$ by
\begin{equation}
\begin{aligned}
    \|\ket{\psi_{\neq \omega_0}}\|\leq\sqrt{1-\|\ket{\psi_{\omega_0}}\|^2}\leq\sqrt{2\epsilon_{\neq\omega_0}}.    
\end{aligned}
\end{equation}
Since the evolution under $H$ will not change the weights on its eigenstates, $\|\ket{\psi_{\omega_0}}\|$ and $\|\ket{\psi_{\neq \omega_0}}\|$ are conserved at all times. Then we can bound the leakage of $\ket{\psi(t)}$ outside $w_0$ at any $t$ by 
\begin{equation}
\begin{aligned}
   &\| (1-P_{w_0})\ket{\psi(t)}\|=\| (1-P_{w_0})e^{-iHt}\ket{\psi(0)}\|\\
   &\leq \|(1-P_{w_0})\ket{\psi_{\omega_0}}\| + \|\ket{\psi_{\neq\omega_0}}\| +\|\ket{\psi_R}\|\\
    &\leq 2\sqrt{2\epsilon_{\neq\omega_0}} + \epsilon^{(2)}_{\lambda,\Delta_q}(d).
\end{aligned}
\end{equation}
For any $d>\lambda$, the $\epsilon^{(2)}_{\lambda,\Delta_q}(d)$ is exponentially small. Thus, we only need to ensure that $e^{-(\xi-\frac{1}{2}\ln 2)n}$ is exponentially small, i.e., $\xi>\frac{1}{2}\ln2$. For convenience, we define $\xi_{\infty}\equiv(\xi-\frac{1}{2}\ln 2)/2$ and set $d=(b-\varepsilon_0)/2>\lambda$, and obtain
\begin{equation}
\label{eq:xi-InfLoc-q-app}
\xi_\infty=\frac{b-\varepsilon_0}{4\Delta_q}\ln\frac{b-\varepsilon_0}{2\lambda}-\frac{b-\varepsilon_0-2\lambda}{4\Delta_q}-1-\frac{1}{4}\ln 2.
\end{equation}
One can verify that if
\begin{equation}
\label{eq:lambdaForInf-app}
    \lambda<\frac{b-\varepsilon_0}{2e}e^{-\frac{2e\Delta_q}{b-\varepsilon_0}}\equiv\lambda_{\infty},
\end{equation}
we can ensure that $\xi_{\infty}>0$. 

Thus, if $\lambda$ satisfies the above condition, any initial eigenstate of $H_C$ localized in $w_0$, with energy $E_g+\varepsilon_0 n$ with $\varepsilon_0<b$, remains localized inside $w_0$ up to infinite time. The leakage outside $w_0$ can be bounded by
\begin{equation}
    \| (1-P_{w_0})\ket{\psi(t)}\| \leq e^{-\xi_{\infty}n}.
\end{equation}

If we can compare the requirement for $\lambda$ in Eq.~\eqref{eq:lambdaForInf-app} with exponentially long dynamical localization \eqref{eq:lambdaForExp-app} and eigenstate localization \eqref{eq:lambdaForEigen-app}, we can find that
\begin{equation}
    \lambda_{\infty}<\lambda_{\textnormal{Eigen}}<\lambda_{\textnormal{Exp}}.
\end{equation}
Thus, the requirement for infinite-time dynamical localization is the most stringent.

\subsubsection{Quasi-$q_{\star}$-local perturbation}
In this subsection, we assume $V_0+H_{\mathrm{d}}$ is quasi-$q_{\star}$-local and 
\begin{equation}
    \|V_0+H_{\mathrm{d}}\|_{q_{\star}\textnormal{-X}}\leq\lambda n.
\end{equation}
Similar to Eq.~\eqref{eq:psi0=3termsInfiniteLoc-app}, we can divide the initial state into three parts
\begin{equation}
\label{eq:psi0=3termsInfiniteLocQausi-app}
    \ket{\psi(0)}\equiv\ket{\psi_{w_0}}+\ket{\psi_{\neq w_0}}+\ket{\psi_R},
\end{equation}
where the first two terms contain all localized eigenstates with energy below $E_g+(\varepsilon_0+d)n$ (with $d>\lambda$) of $H$ and $\ket{\psi_R}$ is the remainder. The eigenstate localization condition requires that $\varepsilon_0+d<b$. From Theorem~\ref{thm:static-Energy-loc} we know that the leakage is 
\begin{equation}
    \|\ket{\psi_R}\|\leq \epsilon^{(1)}_{\lambda,\Delta_{q_\star}}(d).
\end{equation}
From Eq.~\eqref{eq:EigenLocQuasiQLeak-app} we know that for energy below $E_g+(\varepsilon_0+d)n$, all eigenstates of $H$ is localized with high probability $1-e^{-2(1-\ln 2)n}$, with leakage
\begin{equation}
    \|(1-P_w)\ket{\phi_{w,j}}\|\leq e^{-\xi n},
\end{equation}
where we can set $\varepsilon'=\varepsilon_0+d$ and
\begin{equation}
    \xi\equiv\min\left(\frac{b-\varepsilon'-\lambda}{\Delta_{q_\star}}-\frac{\lambda}{\Delta_{q_\star}}\ln\frac{b-\varepsilon'}{\lambda},\frac{\nu_2}{n q_{\star}}\right)-2.
\end{equation}
Similarly, we can bound the first two terms in Eq.~\eqref{eq:psi0=3termsInfiniteLocQausi-app} by
\begin{equation}
    \left\|P_{w_0}\sum_{\omega\neq\omega_0 ,j}a_{\omega,j}\ket{\phi_{\omega,j}}\right\| \leq 2^{n/2}e^{-\xi n},
\end{equation}
and
\begin{equation}
    \|\ket{\psi_{\omega_0}}\|\geq \|P_{\omega_0}\ket{\psi_{\omega_0}}\|
    \geq 1-e^{-(\xi-\frac{1}{2}\ln 2)n}-\epsilon^{(1)}_{\lambda,\Delta_{q_\star}}(d)\equiv1-\epsilon_{\neq\omega_0}.
\end{equation}
Then we can bound the leakage of $\ket{\psi(t)}$ outside $w_0$ at any $t$ by 
\begin{equation}
   \| (1-P_{w_0})\ket{\psi(t)}\|\leq 2\sqrt{2\epsilon_{\neq\omega_0}} + \epsilon^{(1)}_{\lambda,\Delta_{q_\star}}(d).
\end{equation}
Similarly, we define $\xi_{\infty}\equiv(\xi-\frac{1}{2}\ln 2)/2$ and set $d=(b-\varepsilon_0)/2>\lambda$, and only need to ensure that $\xi_{\infty}>0$, where
\begin{equation}
\label{eq:xi-InfLoc-QuasiQ-app}
    \xi_\infty\equiv\min\left(\frac{b-\varepsilon_0-2\lambda}{4\Delta_{q_\star}}-\frac{\lambda}{2\Delta_{q_\star}}\ln\frac{b-\varepsilon_0}{2\lambda},\frac{\nu_2}{2n q_{\star}}\right)\\
    -1-\frac{1}{4} \ln 2.    
\end{equation}

One can verify that if
\begin{equation}
\label{eq:Lambda-InfLoc-Quasiq-app}
    \lambda<\frac{b-\varepsilon_0}{2e},~q_{\star}\leq\min\left(\frac{(b-\varepsilon_0)(1-2/e)}{4eM},\frac{\nu_2}{e n}\right),
\end{equation}
then we can ensure that $\xi_{\infty}>0$. 

Thus, if $\lambda$ and $q_{\star}$ satisfy the above condition, any initial eigenstate of $H_C$ localized in $w_0$, with energy $E_g+\varepsilon_0 n$ with $\varepsilon_0<b$, remains localized inside $w_0$ to infinite time. The leakage outside $w_0$ can be bounded by
\begin{equation}
    \| (1-P_{w_0})\ket{\psi(t)}\| \leq e^{-\xi_{\infty}n}.
\end{equation}

\appsection{Proof of the eigenstate localization}
\label{app:pEigenLoc}

In this section, we consider the following Hamiltonian (similarly defined in~\cite{yin2024eigenstate})
\begin{equation}
    H=H_C+V_0+H_{\textnormal{d}},
\end{equation}
where $H_C$ is a good cLDPC code and $\|V_0\|\sim \lambda n$ is an arbitrary (quasi-)$q$-local perturbation. $H_{\textnormal{d}}$ is a small diagonal detuning defined by
\begin{equation}
    H_{\textnormal{d}}\equiv \frac{1}{n^2}\sum_{j}h_j^z Z_j,
\end{equation}
where $h_j^z$ can be independently randomly chosen from the Gaussian distribution with mean 0 and variance 1. We note that $H_{\textnormal{d}}$ only contributes $O(\frac{1}{n})$ to the total variation of $H$ from $H_C$, so it is negligible when considering leakage bounds in the energy space. In addition, $H_{\textnormal{d}}$ will be infinitesimal in the thermodynamic limit. However, the polynomially small $H_{\textnormal{d}}$ is sufficient to lift the degeneracies in $H_C$ to gaps above an exponentially small threshold. We remark that the essential ingredient here is not the specific form of $H_{\textnormal{d}}$, but the fact that the degeneracies in cLDPC codes can be easily lifted by local detuning. In contrast, the ground state degeneracies of good qLDPC codes cannot be lifted (with maximum separating $\sim\lambda^{O(n)}$) by any local perturbations~\cite{de2025low,yin2025low}, similar to the property in topological order. We will see shortly that lifting the degeneracies is essential for eigenstate localization.

We are going to prove that the perturbed eigenstate $\ket{\psi}$ of $H$ with energy $E'\equiv\varepsilon' n+E_g$ (with $\varepsilon'<b$) is exponentially localized in one cluster $w_{\star}$. For strictly $q$-local $V_0$, we require that $\|V_0\|_{\textnormal{X}}\leq\lambda n$ and
\begin{equation}
\label{eq:lambdaBoundsEigenLocQ-app}
    \lambda<\frac{b-\varepsilon'}{e}\exp\left(-\frac{2\Delta_q}{b-\varepsilon'}\right).
\end{equation}
For quasi-$q_{\star}$-local $V_0$, we require that $\|V_0\|_{q_{\star}\textnormal{-X}}\leq\lambda n$ and
\begin{equation}
\label{eq:lambdaBoundsEigenLocQuasiQ-app}
    \lambda<\frac{b-\varepsilon'}{e},~q_{\star}\leq\min\left(\frac{(b-\varepsilon')(1-2/e)}{4M},\frac{\nu_2}{2 n}\right).
\end{equation}
The leakages under these conditions are determined by Eq.~\eqref{eq:EigenLocQLeak-app} and \eqref{eq:EigenLocQuasiQLeak-app}.

The eigenstate localization for strictly $q$-local perturbation in good cLDPC codes was already proved in~\cite{yin2024eigenstate}. Here, we follow a similar procedure to prove it, but use a slightly tighter bound obtained from our Theorem.~\ref{thm:static-Energy-loc}. We remark that the method used in~\cite{yin2024eigenstate} is hard to extend to quasi-$q$-local perturbations, whereas our Theorem~\ref{thm:static-Energy-loc} can allow us to bound the leakage for quasi-$q$-local perturbations as well.

\appsubsection{Eigenstate localization in cLDPC codes with strictly $q$-local perturbations}

For the perturbed eigenstate $\ket{\psi}$ with energy $E'\equiv\varepsilon' n+E_g$, its energy difference from the top of $W$ is $ (b-\varepsilon' )n$. Thus, we can choose the energy window to be $[E'-dn,E'+dn]$, with $d=b-\varepsilon'$. Note that here the $b$ and $\varepsilon'$ are defined in different eigenbases, i.e., the eigenbases of $H_C$ and $H$, respectively. From Theorem~\ref{thm:static-Energy-loc}, we know that the leakage outside all the clusters is bounded by (noting $P_>=1-P_W$)
\begin{equation}
\|P_{>}\ket{\psi}\|\leq\epsilon^{(2)}_{\lambda,\Delta_q}(b-\varepsilon'),
\end{equation}
which is exponentially small. Then we divide the perturbed eigenstate $\ket{\psi}$ in two parts
\begin{equation}
    \ket{\psi}=\ket{\eta}+\ket{\epsilon},
\end{equation}
where $\ket{\eta}\equiv P_W \ket{\psi}$ is in some clusters and $\ket{\epsilon}\equiv P_{>}\ket{\psi}$ is outside all clusters. Applying the projector $P_W$ to $H\ket{\psi}=E'\ket{\psi}$ we have 
\begin{equation}
    P_WHP_W\ket{\psi}+ P_WH(1-P_W)\ket{\psi}=E'\ket{\eta}. 
\end{equation}
Defining the truncated Hamiltonian as $H_W\equiv P_WHP_W$, we obtain 
\begin{equation}
 (H_W-E')\ket{\eta}=P_W V_0 \ket{\epsilon},
\end{equation} 
where we have used $P_WH_C(1-P_W)=0$ and $P_WH_{\mathrm{d}}(1-P_W)=0$. Thus, $\ket{\eta}$ and $E'$ are an approximate eigensolution of $H_W$, as $||\epsilon\rangle||$ is exponentially small. For simplicity, we denote the RHS norm after normalizing $\ket{\eta}$ as $\epsilon_{>}\equiv\frac{\|P_WV_0\ket{\epsilon}\|}{\sqrt{1-\|\ket{\epsilon}\|^2}}$. Using $\|V_0\|\leq \lambda n$, we can bound $\epsilon_{>}$ by
\begin{equation}
    \epsilon_{>}\leq 2\lambda n \epsilon^{(2)}_{\lambda,\Delta_q}(b-\varepsilon'),
\end{equation}
where we assume $\epsilon^{(2)}_{\lambda,\Delta_q}(b-\varepsilon')<3/4$ so that $\frac{1}{\sqrt{1-\|\ket{\epsilon}\|}}< 2$. 

Note that $H$ cannot directly connect different clusters $w$'s due to the large distances between them, so $P_w H P_{w'}=0$ for $w\neq w'$. Thus, $H_W$ is automatically block diagonalized, with each block corresponding to each cluster $w$. Therefore, the eigenstates of $H_W$ all reside within a single cluster $w$, and all $P_w$'s commute with $H_W$. We can denote $H_W$'s $m$-th eigenstate localized in cluster $w$ as $\ket{\phi_{w,m}}$ with energy $E_{w,m}$. To show that the eigenstates of the whole Hamiltonian $H$ also reside within a single cluster, we need an additional lemma.

\begin{lem}
\label{lem:1-P_V<e/d-e-app}
For any Hamiltonian $H$, if one can find an approximate eigensolution $(H-E')\ket{\Psi}=\ket{\epsilon}$, where $\ket{\Psi}$ is normalized and $\|\ket{\epsilon}\|=\epsilon<1$, then there is at least one exact energy $E_{\star}$ of $H$ satisfying 
\begin{equation}
    |E_{\star}-E'|\leq \epsilon.
\end{equation}
Let $\mathbb{V}$ be some subspace constructed from eigenstates of $H$ that contains $\ket{E_{\star}}$. For any $\ket{E_i}\in \mathbb{V}$ and $\ket{E_j}\notin \mathbb{V}$, we define $\delta_\mathbb{V}\equiv\min_{i,j}|E_j-E_i|$. If $\delta_\mathbb{V}>\epsilon$, then $\ket{\Psi}$ resides almost in $\mathbb{V}$, with leakage
\begin{equation}
    \|(1-P_\mathbb{V})\ket{\Psi}\|\leq \frac{\epsilon}{\delta_\mathbb{V}-\epsilon}.
\end{equation}
    
\end{lem}
We show the proof in Appendix.~\ref{app:proofB1}. Thus, there is at least one $\ket{\phi_{w_\star,j_\star}}$ whose energy is close to $E'$. By choosing $\mathbb{V}=w_\star$, we can obtain 
\begin{equation}
    \|(1-P_{w_{\star}})\ket{\eta}\|\leq \frac{\epsilon_{>}}{\delta_{w_\star}-\epsilon_{>}}\|\ket{\eta}\|,
\end{equation}
where $\delta_{w_\star}$ is the minimum spectral gap between states in cluster $w_{\star}$ and other clusters. In general, we can define
\begin{equation}
    \delta_W\equiv \min_{w_1\neq w_2} |E_{w_1,m_1}-E_{w_2,m_2}|.
\end{equation}
Since $\epsilon_{>}$ is exponentially small, as long as $\delta_W$ decays slower than $\epsilon_{>}$, the leakage of $\ket{\eta}$ outside the $w_{\star}$ can still be exponentially small. Therefore, we need the decaying speed of the gap between different clusters to be lower bounded, which can be stated as the following lemma:

\begin{lem}[Limited level attraction between the clusters~\cite{yin2024eigenstate}]
\label{lem:LLAbetween-app}
If the small diagonal detuning $H_{\textnormal{d}}$ is added, for any $0<\xi_2<\xi_1-2\ln 2$, we have
\begin{equation}
    \mathbf{P}\left(\delta_W < e^{-\xi_1 n}\right)<e^{- \xi_2 n}.
\end{equation}
\end{lem}
We show the proof in Appendix.~\ref{app:proofB2}. We can choose $\xi_1=2$ here for convenience. Thus, $\delta_W$ is at least $e^{-2 n}$ with high probability $1-e^{-2(1-\ln 2)n}$. 

Thus, with high probability $1-e^{-2(1-\ln 2)n}$, the leakage outside $w_\star$ is bounded by 
\begin{equation}
\label{eq:EigenLocQLeak-app}
    \|(1-P_{w_{\star}})\ket{\psi}\|\leq \exp\left(-\left(\xi-o(1)\right) n\right),
\end{equation}
where 
\begin{equation}
\label{eq:xi-eigenLoc-q-app}
    \xi\equiv\frac{b-\varepsilon'}{\Delta_q}\ln\frac{b-\varepsilon'}{\lambda}-\frac{b-\varepsilon'-\lambda}{\Delta_q}-2.
\end{equation}
We can verify that when 
\begin{equation}
\label{eq:lambdaForEigen-app}
    \lambda<\frac{b-\varepsilon'}{e}e^{-\frac{2\Delta_q}{b-\varepsilon'}}\equiv\lambda_{\textnormal{Eigen}},
\end{equation}
we can ensure that $\xi>0$.

\appsubsection{Eigenstate localization in cLDPC codes with quasi-$q$-local perturbations}

For a quasi-$q_{\star}$-local perturbation $V_0$, we can follow similar procedures. We consider an eigenstate $\ket{\psi}$ of 
\begin{equation}
    H=H_C+V_0+H_{\textnormal{d}},
\end{equation}
with energy $E'=\varepsilon'n+E_g$. From Theorem~\ref{thm:static-Energy-loc}, we know that the leakage outside all the clusters is bounded by
\begin{equation}
\|P_{>}\ket{\psi}\|\leq\epsilon^{(1)}_{\lambda,\Delta_{q_{\star}}}(b-\varepsilon'),
\end{equation}
which is exponentially small. Then we divide the quasi-$q_{\star}$-local $V_0$ to the local part $V_L$ and the remainder $V_R$, where $V_L$ includes all terms in $V_0$ with $k$-locality satisfying $k<\nu_2$, and $V_R\equiv V_0-V_L$. The local part $V_L$ satisfies
\begin{equation}
    P_{w'}V_L(t)P_w=0, ~\text{for }~w\neq w',
\end{equation}
and the remainder $V_R$ satisfies
\begin{equation}
    \|V_R\|\leq \lambda n e^{-\nu_2/q_{\star}}.
\end{equation}

Similarly, we can define the truncated Hamiltonian $H_W$ as 
\begin{equation}
    H_W\equiv P_W (H_C +V_L+ H_{\textnormal{d}})P_W,
\end{equation}
and obtain
\begin{equation}
    (H_W-E')\ket{\eta}=P_W V_0 \ket{\epsilon}-P_W V_RP_W\ket{\psi},
\end{equation}
where $\ket{\eta}\equiv P_W \ket{\psi}$ and $\ket{\epsilon}\equiv P_{>}\ket{\psi}$. We define 
\begin{equation}
    \epsilon_{>}\equiv \frac{\|P_W V_0 \ket{\epsilon}-P_W V_RP_W\ket{\psi}\|}{\|\ket{\eta}\|},
\end{equation} 
As the norm of $\ket{\epsilon}$ is exponentially small, we assume it is less than $3/4$, so that $\|\ket{\eta}\|\geq 1/2$. Then we can bound $\epsilon_{>}$ by
\begin{equation}
    \epsilon_{>}\leq 2\lambda n \left(\epsilon^{(1)}_{\lambda,\Delta_{q_{\star}}}(b-\varepsilon')+e^{-\nu_2/q_{\star}}\right).
\end{equation}

Next, we want to use Lemma~\ref{lem:1-P_V<e/d-e-app} to bound the leakage outside a cluster $w$ of $\ket{\psi}$. It is crucial to check whether the eigenstates of the truncated Hamiltonian $H_W$ remain localized in each cluster. Although $V_L$ now includes larger-$k$-local terms, $k$ in $V_L$ is strictly less than $\nu_2$ by definition. Therefore, $H_W$ still contains no interaction between different clusters. Each eigenstate of $H_W$ remains within one single cluster $w$, and all $P_w$'s commute with $H_W$. Thus, the Lemma~\ref{lem:1-P_V<e/d-e-app} still works here. 

The next task is to bound $\delta_W$, which can be defined similarly
\begin{equation}
    \delta_W\equiv \min_{w_1\neq w_2} |E_{w_1,m_1}-E_{w_2,m_2}|.
\end{equation}
One can verify that the Lemma~\ref{lem:LLAbetween-app} also works as long as the eigenstates of $H_W$ remain clustered in each $w$. Thus, We can choose $\xi_1=2$ here for convenience. Thus, $\delta_W$ is at least $e^{-2 n}$ with high probability $1-e^{-2(1-\ln 2)n}$. 
Therefore, with high probability $1-e^{-2(1-\ln 2)n}$, the leakage outside $w_\star$ is bounded by
\begin{equation}
\label{eq:EigenLocQuasiQLeak-app}
    \|(1-P_{w_{\star}})\ket{\psi}\|\leq \exp\left(-\xi n\right),
\end{equation}
where
\begin{equation}
\label{eq:xi-eigenLoc-Quasiq-app}
    \xi\equiv\min\left(\frac{b-\varepsilon'-\lambda}{\Delta_{q_\star}}-\frac{\lambda}{\Delta_{q_\star}}\ln\frac{b-\varepsilon'}{\lambda},\frac{\nu_2}{n q_{\star}}\right)-2.
\end{equation}
We can choose $\lambda<\frac{b-\varepsilon'}{e}$ and adjust $q_{\star}$ to ensure that both terms are positive. This gives an upper bound for $q_{\star}$
\begin{equation}
\label{eq:Lambda-eigenLoc-Quasiq-app}
   \lambda<\frac{b-\varepsilon'}{e}, \quad q_{\star}\leq\min\left(\frac{(b-\varepsilon')(1-2/e)}{4M},\frac{\nu_2}{2 n}\right).
\end{equation}

\appsubsection{Discussions on Eigenstate localization in qLDPC codes}
To investigate the eigenstate localization in qLDPC codes, we can still use Theorem~\ref{thm:static-Energy-loc} to bound the $\|P_{>}\ket{\psi}\|$ exponentially. 

However, an issue arises when we want to apply Lemma~\ref{lem:1-P_V<e/d-e-app} to bound the $\|(1-P_w)\ket{\psi}\|$, where we need to lower bound $\delta_W$ and ensure that it is much larger than the $\|P_{>}\ket{\psi}\|$. For strictly $q$-local $V_0$, we can upper bound $\|P_{>}\ket{\psi}\|$ by $\epsilon^{(2)}_{\lambda,\Delta_q}(b-\varepsilon')$. On the other hand, it is recently proved~\cite{yin2025low,de2025low} that the ground state degeneracy of good qLDPC codes can only be lifted with maximum energy difference bounded by $\lambda e^{-\Theta(n)}$~\cite{yin2025low}, which has similar scaling with $\epsilon^{(2)}_{\lambda,\Delta_q}$. In addition, this only gives an upper bound of $\delta_W$ but what we need is its lower bound. Since the upper bound is already exponentially small, it is unlikely to rigorously obtain a lower bound that is much larger than $\epsilon^{(2)}_{\lambda,\Delta_q}$.

While the ground state degeneracy cannot be lifted by local operators, we can consider detuning $H_{\mathrm{d}}$ including terms like
\begin{equation}
    H_{\mathrm{d}}=\frac{1}{n^2}\sum_{j}h_j\bar{Z}_j,
\end{equation}
where $\bar{Z}_j$ is the logical $Z$ operator for $j$-th logical qubit. Although such detuning term requires $O(n)$-body interaction and is unrealistic, we can consider it as a spontaneous symmetry breaking: $H_{\textnormal{d}}$ will be infinitesimal in the thermodynamic limit but is sufficient to lift all the degeneracies. Thus, although the eigenstate localization for qLDPC may not exist in finite systems, its dynamical stability can be understood from the SSB perspective.

\appsubsection{Proof of Lemma \thesection.1}
\label{app:proofB1}

\smallskip\noindent{\bf Lemma \thesection.1} For any Hamiltonian $H$, if one can find an approximate eigensolution $(H-E')\ket{\Psi}=\ket{\epsilon}$, where $\ket{\Psi}$ is normalized and $\|\ket{\epsilon}\|=\epsilon<1$, then there is at least one exact energy $E_{\star}$ of $H$ satisfying 
\begin{equation}
    |E_{\star}-E'|\leq \epsilon.
\end{equation}
Let $\mathbb{V}$ be some subspace constructed from eigenstates of $H$ that contains $\ket{E_{\star}}$. Suppose that for any $\ket{E_i}\in \mathbb{V}$ and $\ket{E_j}\notin \mathbb{V}$, we define $\delta_\mathbb{V}\equiv\min_{i,j}|E_j-E_i|$. If $\delta_\mathbb{V}>\epsilon$, then $\ket{\Psi}$ resides almost in $\mathbb{V}$, with leakage
\begin{equation}
    \|(1-P_\mathbb{V})\ket{\Psi}\|\leq \frac{\epsilon}{\delta_\mathbb{V}-\epsilon}.
\end{equation}

\begin{proof}
We can expand $\ket{\Psi}$ in the eigenbasis of $H$,
\begin{equation}
    \ket{\Psi}=\sum_ja_j\ket{E_j}.
\end{equation}
From $(H-E')\ket{\Psi}=\ket{\epsilon}$, we can obtain
\begin{equation}
    \ket{\epsilon}=\sum_ja_j(E_j-E')\ket{E_j}
\end{equation}
Taking the norm square of both sides, we obtain
\begin{equation}
\label{eq:epsilon=sum_jaj-app}
    \epsilon^2=\sum_j |a_j|^2(E_j-E')^2
\end{equation}
Since $\sum_j|a_j|^2=1$, there is at least one $E_\star$ in all $E_j$'s satisfying 
\begin{equation}
    |E_\star-E'|\leq \epsilon,
\end{equation}
otherwise the RHS of Eq.~\eqref{eq:epsilon=sum_jaj-app} is always greater than $\epsilon^2$. 

Then we can choose a subspace $\mathbb{V}\ni \ket{E_\star}  $ and rewrite the expansion for $\ket{\Psi}$
\begin{equation}
    \ket{\Psi}=\sum_{j} a_j^{\text{in}}\ket{E^{\text{in}}_j}+\sum_ja_j^{\text{out}}\ket{E^{\text{out}}_j}
\end{equation}
where $\ket{E^{\text{in}}_j}\in \mathbb{V}$ and $\ket{E^{\text{out}}_j}\notin \mathbb{V}$. Using $(H-E')\ket{\Psi}=\ket{\epsilon}$, we obtain
\begin{equation}
    \ket{\epsilon}=\sum_{j} a_j^{\text{in}}(E^{\text{in}}_j-E')\ket{E^{\text{in}}_j}
    +\sum_ja_j^{\text{out}}(E^{\text{out}}_j-E')\ket{E^{\text{out}}_j}.    
\end{equation}
By definition, we have $|E^{\text{out}}_j-E_\star|\geq \delta_\mathbb{V}$. With $|E_\star-E'|\leq\epsilon$, we can get $|E^{\text{out}}_j-E'|\geq \delta_\mathbb{V}-\epsilon$.
Thus, we obtain
\begin{equation}
\epsilon^2\geq\sum_{j}|a_j^{\text{out}}|^2(E^{\text{out}}_j-E')^2\geq(\delta_{\mathbb{V}}-\epsilon)^2 \sum_{j}|a_j^{\text{out}}|^2.
\end{equation}
Therefore, the leakage can be bounded by 
\begin{equation}
    \|(1-P_\mathbb{V})\ket{\Psi}\|=\sqrt{\sum_{j}|a_j^{\text{out}}|^2}\leq \frac{\epsilon}{\delta_\mathbb{V}-\epsilon}.
\end{equation}
\end{proof}

\appsubsection{Proof of Lemma \thesection.2}
\label{app:proofB2}
\smallskip\noindent{\bf Lemma \thesection.2}. If the small diagonal detuning $H_\mathrm{d}$ is added
\begin{equation}
    H_{\textnormal{d}}\equiv \frac{1}{n^2}\sum_{j}h_j^z Z_j,
\end{equation}
for any $0<\xi_2<\xi_1-\ln 4$, we have
\begin{equation}
    \mathbf{P}\left(\delta_W < e^{-\xi_1 n}\right)<e^{- \xi_2 n}.
\end{equation}
Here, $\delta_W$ is defined in the eigenbasis of 
\begin{equation}
    H_W\equiv P_W(H_C+V_0+H_{\textnormal{d}})P_W,
\end{equation}
whose eigenstates are all exactly localized in a specific cluster $w$ and can be labeled as $\ket{\phi_{w,m}}$ with energy $E_{w,m}$. Then we can define
\begin{equation}
    \delta_W \equiv \min_{w_1\neq w_2} |E_{w_1,m_1}-E_{w_2,m_2}|.
\end{equation}

\begin{proof}
For a specific cluster $w$ with codeword $\ket{\textnormal{w}}$, we can write it as a bitstring $\textnormal{w}\equiv \{z^w_1, z^w_2...z^w_n\}$, where each $z^w_j$ is in $\{0,1\}$. Consider the following variable
\begin{equation}
    h^w_j\equiv\frac{1}{\sqrt{n}}\sum_{k=1}^n(-1)^{z^w_k}e^{-i\frac{2\pi}{n}jk}h^z_k.
\end{equation}
One can find an orthogonal transformation to connect the random variables $h^z_i$ to $h^w_i$, so that all $h^w_i$'s also follow independent Gaussian distributions with zero mean and unit variance. The $H_{\textnormal{d}}$ become
\begin{equation}
    H_{\textnormal{d}}=\frac{1}{n^{5/2}}\sum_{j=1}^{n} h^w_j \sum_{k=1}^n e^{i\frac{2\pi}{n}jk} (-1)^{z^w_k} Z_k.
\end{equation}
We can regard $H_W=H_C+V_0+H_{\textnormal{d}}$ as a function of all $\{h^w_j\}$. Using the Hellmann-Feynman theorem, for arbitrary energy $E_{w,m}$ of $H_W$, we have 
\begin{equation}
    \frac{\partial E_{w,m}}{\partial h^w_j}=\Braket{\phi_{w,m}|\frac{\partial H_W}{\partial h^w_j}|\phi_{w,m}}.
\end{equation}
Take $j=1$ as an example, we obtain
\begin{equation}
    \frac{\partial E_{w,m}}{\partial h^w_1}=\frac{1}{n^{5/2}}\Braket{\phi_{w,m}|\sum_{k=1}^n (-1)^{z^w_k}Z_k|\phi_{w,m}}.
\end{equation}
Note that for a $Z$-basis state $\ket{z}$,
\begin{equation}
    \Braket{z|\sum_{k=1}^n (-1)^{z^w_k}Z_k|z}=n-2\cdot\mathbf{D} (\textnormal{w},z),
\end{equation}
where $\mathbf{D}$ here is reduced to the Hamming distance. Since all $Z$-basis states inside $w$ is at most $\gamma n$ far from $\textnormal{w}$, we obtain
\begin{equation}
    \frac{\partial E_{w,m}}{\partial h^w_1}\geq\frac{n-2\gamma n}{n^{5/2}}.
\end{equation}
Similarly, for an arbitrary $w'\neq w$, all $Z$-basis states inside $w'$ is at least $d_C-\gamma n$ far from $\textnormal{w}$, we have
\begin{equation}
\begin{aligned}
    \frac{\partial E_{w',m'}}{\partial h^w_1}&=\frac{1}{n^{5/2}}\braket{\phi_{w',m'}|\sum_{k=1}^n (-1)^{z^w_k}Z_k|\phi_{w',m'}}\\
    &\leq\frac{n-2(d_C-\gamma n)}{n^{5/2}}.
\end{aligned}
\end{equation}
Therefore, we obtain
\begin{equation}
    \frac{\partial}{\partial h^w_1}(E_{w,m}-E_{w',m'})\geq \frac{2(d_C-2\gamma n)}{n^{5/2}}=\frac{2\nu_2}{n^{5/2}}.
\end{equation}
Now we can fix all the other $h^w_j$ with $j\neq 1$ and regard $E_{w,m}-E_{w',m'}$ as a function of $h^w_1$, which increases monotonically. Thus, the interval length of $h^w_1$ for $|E_{w,j}-E_{w',j'}|<e^{-\xi_1 n}$ is at most $\frac{n^{5/2}}{\nu_2}e^{-\xi_1 n}$. Taking into account all possible pairs for $E_{w,m}$ and $E_{w',m'}$, the total interval length is at most $4^n\frac{n^{5/2}}{\nu_2}e^{-\xi_1 n}$. Since $h^w_1$ is chosen from a Gaussian with zero mean and unit variance, we have
\begin{equation}
    \mathbf{P}(h^w_1)=\frac{1}{\sqrt{2\pi}}\exp\left(-\frac{(h^w_1)^2}{2}\right).
\end{equation}
Then the probability that $h_1^w$ is within those intervals is bounded by 
\begin{equation}
    \mathbf{P}\left(\delta_W < e^{-\xi_1 n}\right)\leq \frac{n^{5/2}}{\nu_2\sqrt{2\pi}}e^{-(\xi_1-\ln 4) n}<e^{- \xi_2 n},
\end{equation}
where we can choose $\xi_2$ in the range $0<\xi_2<\xi_1-\ln 4$.

\end{proof}

\appsection{Proof of the slow mixing of Gibbs samplers}
\label{app:pGibbs}

In this section, we investigate the efficiency of reaching thermal equilibrium of the c/qLDPC codes via a local Gibbs sampler, especially in low-temperature regions. A general Gibbs sampler can be any quantum channel $\mathcal{M}_H$ for a Hamiltonian $H$ whose steady state is $\rho_0=\frac{1}{\mathcal{Z}}e^{-\beta H}$, where $\mathcal{Z}\equiv\operatorname{tr} e^{-\beta H}$ is the partition function. In principle, by applying $\mathcal{M}$ on an arbitrary initial state $t$ times, and if $t$ is sufficiently large, the resultant density matrix can be arbitrarily close to the steady state of $\mathcal{M}$, which is the Gibbs state of $H$ by definition. The minimal $t$ required to reach sufficiently close to the steady state is called the mixing time.

We can always write $\mathcal{M}$ in Kraus operators
\begin{equation}
    \mathcal{M}[\rho]=\sum_iK_i \rho K_i^\dagger.
\end{equation}
The locality constraint on $\mathcal{M}$ requires that all $ K_i$'s are strictly $k$-local where $k$ is an $O(1)$ constant. Although we mainly focus on strictly $k$-local $\mathcal{M}$ here, the result can easily be extended to quasi-$k$-local channels~\cite{rakovszky2024bottlenecks}.

We are going to prove that the Gibbs samplers for c/qLDPC codes are generally slow, even if (quasi-)$q$-local perturbations are added to $H_C$. In particular, we set $\mathcal{M}$ as a local Gibbs sampler for
\begin{equation}
    H=H_C+V_0.
\end{equation}
We will show shortly that if an initial state is prepared in one cluster $w$, it takes an exponentially long time to converge to the steady state, which spreads in all clusters. The proof relies on the following quantum bottleneck theorem proved in~\cite{rakovszky2024bottlenecks}.

\begin{lem}[Quantum bottleneck theorem~\cite{rakovszky2024bottlenecks}]
For a quantum channel $\mathcal{M}[\rho]=\sum_iK_i \rho K_i^\dagger$ whose steady state is $\rho_0$, suppose we can divide the Hilbert space into $\mathbb{H}=\mathbb{A}\oplus\mathbb{B}_1 \oplus\mathbb{B}_2\oplus\mathbb{C}$, with corresponding projectors $P_{\mathbb{A},\mathbb{B}_1,\mathbb{B}_2,\mathbb{C}}$, and $\mathcal{M}$ satisfies
\begin{equation}
\left(P_{\mathbb{B}_2}+P_{\mathbb{C}}\right) K_i P_{\mathbb{A}}=\left(P_{\mathbb{A}}+P_{\mathbb{B}_1}\right) K_i P_{\mathbb{C}}=0 .
\end{equation}
Then we can find an approximate steady state $\rho_{\mathbb{A}}\equiv\frac{P_{\mathbb{A}}\rho_0 P_{\mathbb{A}}}{\operatorname{tr}\left(P_{\mathbb{A}} \rho_0\right)  }$ with error bounded by
\begin{equation}
\left\|\mathcal{M}\left[\rho_{\mathbb{A}}\right]-\rho_{\mathbb{A}}\right\|_1 \leq 10 \frac{\sqrt{\operatorname{tr}(P_{\mathbb{B}} \rho_0)}}{\operatorname{tr}\left(P_{\mathbb{A}} \rho_0\right)}\equiv \epsilon_{\mathcal{M}}(\mathbb{A}),
\end{equation}
where $P_{\mathbb{B}}=P_{\mathbb{B}_1}+P_{\mathbb{B}_2}$.
\end{lem}

From the Lemma, one can verify that if $\epsilon_{\mathcal{M}}(\mathbb{A})$ is exponentially small, an initial state localized in $\mathbb{A}$ will be stuck there for an exponentially long time. 

Thus, the next step is to show that $\epsilon_{\mathcal{M}}(\mathbb{A})$ is indeed exponentially small even if (quasi-)$q$-local perturbations are added. The procedures for strictly $q$-local perturbations are similar to those in~\cite{rakovszky2024bottlenecks}, except that we bound the high-energy contribution---the actual bottleneck of the Gibbs sampler $\mathcal{M}$---from Theorem.~\ref{thm:static-Energy-loc}. In addition, we can also bound the bottleneck ratio for quasi-$q$-local perturbations from Theorem~\ref{thm:static-Energy-loc}, which is beyond the methods in \cite{rakovszky2024bottlenecks}.

\appsubsection{Strictly $q$-local perturbations}

We can then choose $\mathbb{A}=w$, $\mathbb{C}$ being all the other clusters, and $\mathbb{B}$ being the high energy contribution with $\mathbb{B}_1$ contain all eigenstates of $H_C$ that can be connected by local channel $\mathcal{M}$ with eigenstates in $w$. This gives 
\begin{equation}
    P_{\mathbb{A}}=P_w,~P_{\mathbb{B}}=P_{>},~ P_{\mathbb{C}}=\sum_{w'\neq w}P_{w'},~\rho_0=\frac{1}{\mathcal{Z}}e^{-\beta H}.
\end{equation}
where $\mathcal{Z}\equiv\operatorname{tr} e^{-\beta H}$ is the partition function. Thus, we obtain
\begin{equation}
    \epsilon_{\mathcal{M}}(w)\leq 10 \sqrt{\mathcal{Z}}\frac{\sqrt{\operatorname{tr}\left( P_{>}e^{-\beta H}\right)}}{\operatorname{tr}\left( P_we^{-\beta H}\right)}.
\end{equation}
Note that we can always upper bound $\mathcal{Z}$ by 
\begin{equation}
   \mathcal{Z}\leq 2^{n}e^{-\beta(E_g-\lambda n)},
\end{equation}
where $E_g$ is the ground state energy of $H_C$ and $E_g-\lambda n$ is the lowest possible ground state energy of $H$ via Weyl's inequality. The lower bound of $\operatorname{tr}( P_w \rho_0)$ can be obtained by using $P_w=\sum_{\ket{\psi} \in w}\ket{\psi}\bra{\psi}$, so that
\begin{equation}
\begin{aligned}
\operatorname{tr}\left( P_we^{-\beta H}\right) &\geq \max_{|\psi\rangle \in w} {\operatorname{tr}( e^{-\beta H}\ket{\psi}\bra{\psi}}) \\
&\geq e^{-\beta(E_g+\lambda n)}.
\end{aligned}
\end{equation}
In the second line, we have used the fact that any state in the cluster $w$ has an energy $E\le E_g+\lambda n$. The upper bound of $\operatorname{tr}\left( P_{>}e^{-\beta H}\right)$ can be obtained by the following inequality, which is defined for an arbitrary projector $P$ to a subspace $S$
\begin{equation}
    \operatorname{tr}(PA)=\operatorname{tr}(PAP) \leq \mathrm{rank}(P) \cdot \lambda_{\max}(PAP)=\mathrm{rank}(P) \cdot \max_{\substack{u\in S \\ \|u\|=1}} \langle u|A|u\rangle\leq \mathrm{rank}(P) \cdot \max_{i,P\ket{\phi_i}\neq 0}\braket{\phi_i|A|\phi_i},
\end{equation}
where $A$ can be any Hermitian operators with all eigenstates denoted by $\ket{\phi_i}$, and we use $\lambda_{\max}(A)$ (different from the notation $\lambda$ for perturbation strength) to denote the maximum eigenvalue of a Hermitian operator $A$. The last step can be easily proved by expanding arbitrary subspace state $u$ in the eigenbasis of $A$. Similarly, if we denote $\ket{\phi_i}$ as eigenstates of perturbed Hamiltonian $H$
\begin{equation}
    \operatorname{tr}\left( P_{>}e^{-\beta H}\right) \leq 2^n \max_{i,P_{>}\ket{\phi_i}\neq 0}\Braket{\phi_i|e^{-\beta H}|\phi_i}.
\end{equation}
From Theorem~\ref{thm:static-Energy-loc}, we can bound the expectation by expanding $\ket{\phi_i}$ in the eigenbasis of $H$, where we swap the roles of $H$ and $H_C$. Suppose the energy of $\ket{\phi}$ in $H$ is $E_\phi\geq E_g+bn$, we know that $\ket{\phi}$ is exponentially localized in the energy window $\mathcal{E}_{\phi}\equiv [E_\phi-dn,E_\phi+dn]$ of $H$, where $d$ can be any constant larger than $\lambda$. By setting $\ket{\phi}=\sum_{j}a_j\ket{\psi_j}$ and $\ket{\psi_j}$ is an eigenstate of $H$ with energy $E_j$, we have the following for any $\ket{\phi}\in\{\ket{\phi_i}\}$.
\begin{equation}
\label{eq:phie^betaH=sum+sum-app}
    \braket{\phi|e^{-\beta H}|\phi}=\sum_{j,E_j\in\mathcal{E}_{\phi}} |a_j|^2 e^{-\beta E_j}+\sum_{j,E_j\notin\mathcal{E}_{\phi}} |a_j|^2 e^{-\beta E_j}    \leq e^{-\beta (E_\phi -dn)}+e^{-\beta E_g}\epsilon^{(2)}_{\lambda,\Delta_q}(d).
\end{equation}
Thus, we have
\begin{equation}
    \mathrm{tr}\left( P_{>}e^{-\beta H}\right) \leq e^{n \ln 2-\beta E_g}\left(e^{-\beta(b-d)n}+\epsilon^{(2)}_{\lambda,\Delta_q}(d)\right).
\end{equation}
Then we can bound $\epsilon_{\mathcal{M}}(w)$ by
\begin{equation}
    \frac{\epsilon_{\mathcal{M}}(w)^2}{100}\leq e^{n(2 \ln 2 +3\beta \lambda) }\left(e^{-\beta(b-d)n}+\epsilon^{(2)}_{\lambda,\Delta_q}(d)\right).
\end{equation}
For convenience, we choose $d=b/4$, we have
\begin{equation}
    \epsilon_{\mathcal{M}}(w)\leq e^{-\xi n},
\end{equation}
where
\begin{equation}
\label{eq:xi-Gibbs-q-app}
    \xi\equiv\min\left(\frac{3\beta(b-4\lambda)}{8},\frac{\frac{b}{4}\ln\frac{b}{4e\lambda}+\lambda}{2\Delta_q}-\frac{3}{2}\beta\lambda\right)-\ln2.
\end{equation}
To ensure that $\xi>0$, we require that
\begin{equation}
\label{eq:lambda-Gibbs-q-app}
    \beta>\frac{8\ln2}{3b},~\lambda<\min\left(\frac{b}{4e}e^{-\Delta_q(3\beta+\frac{8\ln2}{b})},\frac{3\beta b}{4}-2\ln 2\right).
\end{equation}

\appsubsection{Quasi-$q$-local perturbations}

One can verify that for Quasi-$q_{\star}$-local perturbation $V_0$, most of the steps are similar to strictly $q$-local cases, except upper bounding $\mathrm{tr}(P_{>}e^{-\beta H})$. This can be similarly bounded by
\begin{equation}
    \operatorname{tr}\left( P_{>}e^{-\beta H}\right) \leq 2^n \max_{P_{>}\ket{\phi}\neq 0}\Braket{\phi|e^{-\beta H}|\phi}.
\end{equation}
From Theorem~\ref{thm:static-Energy-loc} and Eq.~\eqref{eq:phie^betaH=sum+sum-app} we know that 
\begin{equation}
    \braket{\phi|e^{-\beta H}|\phi}
    \leq e^{-\beta (E_\phi -dn)}+e^{-\beta E_g}\epsilon^{(1)}_{\lambda,\Delta_{q_\star}}(d),
\end{equation}
where the only difference is replacing $\epsilon^{(2)}_{\lambda,\Delta_q}$ with $\epsilon^{(1)}_{\lambda,\Delta_{q_\star}}$. Similarly, we have
\begin{equation}
    \frac{\epsilon_{\mathcal{M}}(w)^2}{100}\leq e^{n(2 \ln 2 +3\beta \lambda) }\left(e^{-\beta(b-d)n}+\epsilon^{(1)}_{\lambda,\Delta_{q_{\star}}}(d)\right).
\end{equation}
For convenience, we choose $d=b/4$, we have
\begin{equation}
    \epsilon_{\mathcal{M}}(w)\leq e^{-\xi n},
\end{equation}
where
\begin{equation}
\label{eq:xi-Gibbs-QuasiQ-app}
    \xi\equiv\min\left(\frac{3\beta(b-4\lambda)}{8},\frac{\frac{b}{4}-\lambda\ln\frac{be}{4\lambda}}{2\Delta_{q_{\star}}}-\frac{3}{2}\beta\lambda\right)-\ln2.
\end{equation}
To ensure that $\xi>0$, we can choose the following requirements: 
\begin{equation}
\label{eq:lambda-Gibbs-QuasiQ-app}
    \beta>\frac{8\ln2}{3b},~\lambda<\min\left(\frac{b}{4e},\frac{3\beta b}{4}-2\ln 2\right),~q_{\star}<\frac{b}{2M}\cdot\frac{e-2}{2e\ln 2+3\beta b}.
\end{equation}

\appsection{Some bounds for nested commutators}
\label{app:proofCommute}
In this section, we provide some useful bounds for the nested commutator $\mathrm{ad}_{H}^m (V)$ along with their asymptotic behavior under different conditions.

\appsubsection{$H$ is $p$-local and $V$ is $q$-local}
We consider a $p$-local $H$ and a $q$-local $V$. They can be written in terms supported on different sites:
\begin{equation}
    H=\sum_{A}h_A, ~V=\sum_{B}v_B,
\end{equation}
where each support contains at most $p$ sites for $A$ and at most $q$ for $B$. In addition, we require $H$ to satisfy 
\begin{equation}
    \sum_{A\ni i}\|h_A\| \leq M
\end{equation}
for any site $i$. We first consider the case where $q<p$. Define $\Delta_{p-1}\equiv2(p-1)M$, the norm of the nested commutators are bounded by
\begin{equation}
    \left\|\mathrm{ad}_{H}^m (V)\right\|\leq m!\Delta_{p-1}^m \left\|V\right\|_{\textnormal{X}},
\end{equation}
where $\|V \|_{\textnormal{X}}\equiv\sum_{X}\|v_X\|$. We do this by induction. For $m=0$ the inequality is trivial. For $m=1$, we have
\begin{equation}
    \mathrm{ad}_{H} (V)=\sum_{B}\sum_{A,A\cap B\neq \emptyset} [h_A,v_B].
\end{equation}
For any single term $v_B$ on the support $B$, it contains $q$ sites and $q\leq p-1$, so
\begin{equation}
\begin{aligned}
    &\left\|\sum_{A,A\cap B\neq \emptyset} [h_A,v_B]\right\|\leq\sum_{i,i \in B}\sum_{A,A\ni i}\left\| [h_A,v_B]\right\|\\
    &\leq \sum_{i,i \in B}2M\|v_B\|\leq 2(p-1)M\|v_B\|= \Delta_{p-1}\|v_B\|,
\end{aligned}
\end{equation} 
Thus,
\begin{equation}
    \|\mathrm{ad}_{H} (V)\|\leq \sum_{B}\Delta_{p-1}\|v_B\|= \Delta_{p-1} \|V\|_{\textnormal{X}}.
\end{equation}
Now we consider $m=2$. Note that we can still write $\mathrm{ad}_H(V)$ as 
\begin{equation}
    \mathrm{ad}_H(V)=\sum_{B} v^{(1)}_B,
\end{equation}
where
\begin{equation}
    v^{(1)}_{B^{(1)}}\equiv v_{B}^{(1)}=\sum_{A,A\cap B\neq \emptyset} [h_A,v_B].
\end{equation}
We use $v^{(1)}_{B^{(1)}} $ to replace the notion of $ v^{(1)}_B$ because it is now supported on $B^{(1)}$, which is an extension of original $B$ and contains at most $2(p-1)$ sites (since $H$ is $p$-local). By the above analysis, we know that the norm of $v^{(1)}_{B^{(1)}} $ is bounded by $\Delta_{p-1} \|v_B\|$. Thus,
\begin{equation}
    \mathrm{ad}_{H}^2 (V)=\sum_{B^{(1)}}\sum_{A,A\cap B^{(1)}\neq \emptyset} [h_A,v^{(1)}_{B^{(1)}}].
\end{equation}
By similar procedure, we obtain
\begin{equation}
    \|\mathrm{ad}_{H}^2 (V)\| \leq 4(p-1)M \Delta_{p-1} \|V\|_{\textnormal{X}}=2 \Delta_{p-1}^2 \|V\|_{\textnormal{X}}.
\end{equation}
Since $B^{(k)}$ grows with at most $(p-1)$ sites from $B^{(k-1)}$, $B^{(k)}$ is supported on at most $k(p-1)$ sites. Then the factor multiplied at $k$-th step is $2k(p-1)M=k\Delta_{p-1}$. Thus, by induction we have
\begin{equation}
    \left\|\mathrm{ad}_{H}^m (V)\right\|\leq m!\Delta_{p-1}^m \left\|V\right\|_{\textnormal{X}}.
\end{equation}
For $q\geq p$, by the similar procedure as above we can obtain the bound
\begin{equation}
\begin{aligned}
    \left\|\mathrm{ad}_{H}^m (V)\right\|&\leq \frac{\Gamma(\frac{q}{p-1}+m)}{\Gamma(\frac{q}{p-1})}\Delta_{p-1}^m \left\|V\right\|_{\textnormal{X}}\\ &<m!\Delta_q^m \left\|V\right\|_{\textnormal{X}}.
\end{aligned}
\end{equation}
Thus, if $H+V$ is $k$-local, we can always have
\begin{equation}
\label{eq:klocalCommuteBound-app}
    \left\|\mathrm{ad}_{H}^m (V)\right\|< m!\Delta_k^m \left\|V\right\|_{\textnormal{X}}.
\end{equation}

\appsubsection{$H$ is mutually commuting and $V$ is $q$-local}

Here we assume that the Hamiltonian $H$ can be decomposed as
\begin{equation}
    H=\sum_{X} h_X,
\end{equation}
where each $h_X$ is supported on $X$ and commutes with each other but can couple arbitrarily many sites, as long as 
\begin{equation}
    \sum_{X,X\ni i}\|h_X\|\leq M.
\end{equation} 
For $m=1$, we have
\begin{equation}
    \mathrm{ad}_{H} (V)=\sum_{B}\sum_{A,A\cap B\neq \emptyset} [h_A,v_B].
\end{equation}
We can obtain the following
\begin{equation}
\begin{aligned}
    &\left\|\sum_{A,A\cap B\neq \emptyset} [h_A,v_B]\right\|\leq\sum_{i,i \in B}\sum_{A,A\ni i}\left\| [h_A,v_B]\right\|\\
    &\leq \sum_{i,i \in B}2M\|v_B\|= 2qM\|v_B\|= \Delta_q\|v_B\|,
\end{aligned}
\end{equation} 
where we define $\Delta_q\equiv 2qM$. For $m=2$, we can still write $\mathrm{ad}_H(V)$ as 
\begin{equation}
    \mathrm{ad}_H(V)=\sum_{B} v^{(1)}_B,
\end{equation}
where
\begin{equation}
    v_{B}^{(1)}=\sum_{A,A\cap B\neq \emptyset} [h_A,v_B].
\end{equation}
Now we consider when $[H_C, v_{B}^{(1)}]$ can be non-zero. If $C\cap B= \emptyset$, then $[H_C,v_B]=0$. By definition, we also have $[H_C,h_A]=0$ for arbitrary $A$. So $[H_C, v_{B}^{(1)}]=0$ if $C\cap B= \emptyset$. Thus, although the support of $v_{B}^{(1)}$ expands from $v_B$, $[H_C, v_{B}^{(1)}]$ is nonzero only if $C\cap B\neq \emptyset$. Thus, we can obtain
\begin{equation}
    \mathrm{ad}_{H}^2 (V)=\sum_{B}\sum_{C,C\cap B\neq \emptyset} [H_C,v_B^{(1)}],
\end{equation}
and thus bound the terms for each initial support $B$ by
\begin{equation}
\begin{aligned}
    &\left\|\sum_{C,C\cap B\neq \emptyset} [H_C,v_B^{(1)}]\right\|\leq\sum_{i,i \in B}\sum_{C,C\ni i}\left\| [H_C,v_B^{(1)}]\right\|\\
    &\leq \sum_{i,i \in B}2M\|v_B^{(1)}\|= 2qM\|v_B^{(1)}\|\leq \Delta_q^2\|v_B\|.
\end{aligned}
\end{equation} 
Thus, by induction we have
\begin{equation}
    \left\|\mathrm{ad}_{H}^m (V)\right\|\leq \Delta_q^m \left\|V\right\|_{\textnormal{X}}.
\end{equation}

\appsubsection{$H$ is mutually commuting and $V$ is quasi-$q$-local}

By the definition of quasi-$q$-local $V=\sum_{B}v_B$ in Appendix~\ref{app:normsTVs}, we have 
\begin{equation}
    \sum_{B,|B|=k}\|v_B\|\leq \frac{\|V\|_{q\textnormal{-X}}}{q+1} e^{-k/q}.
\end{equation}
where the additional factor $\frac{1}{q+1}$ ensures that the total termwise norm $\|V\|_{\mathrm{X}}$ will be bounded by $\|V\|_{q\textnormal{-X}}$ regardless of $q$. Using Eq.~(\ref{eq:klocalCommuteBound-app}), we have the following bounds
\begin{equation}
\begin{aligned}
    \left\|\mathrm{ad}_{H}^m (V)\right\|&\leq (2M)^m\sum_{k=1}^n k^m\sum_{B,|B|=k}\left\|v_{B}\right\|\\
    &\leq\|V\|_{q\textnormal{-X}} \frac{(2M)^m}{q+1}\sum_{k=1}^n k^m e^{-k/q}.    
\end{aligned}
\end{equation}
Here we give a rough but simple estimation to bound the summation at RHS by $m!$. Note that $x^m e^{-x/q}$ is a unimodal function, and for a unimodal function $f(x)$ we have $\sum_{x=a}^b f(x)\leq \int_{a}^bf(x)\mathrm{d}x+\max_{x\in[a,b]} f(x)$. Then we have
\begin{equation}
\begin{aligned}
 \sum_{k=1}^\infty k^m &e^{-k/q}\leq \int_{0}^\infty  x^m e^{-x/q}\mathrm{d}x +\max_k k^m e^{-k/q}\\
    &\leq q^{m+1}m!+(mq/e)^m\\
    &=(q+1)q^m m!(\frac{q}{q+1}+\frac{1}{q+1}\frac{m^m}{e^m m!}).
\end{aligned}
\end{equation}
One can verify that $\frac{m^m}{e^m m!}$ is strictly decreasing for integer $m\geq 1$, so $\frac{m^m}{e^m m!}\leq\frac{1}{e}$, Then we have
\begin{equation}
\begin{aligned}
 \frac{1}{q+1}\sum_{k=1}^\infty k^m e^{-k/q}&\leq q^m m!\frac{q+1/e}{q+1}\\
    &< q^m m!
\end{aligned}
\end{equation}
Thus, we have
\begin{equation}
    \left\|\mathrm{ad}_{H}^m (V)\right\|< \Delta_{q}^m m!\|V\|_{q\text{-X}}.
\end{equation}

\end{document}